\documentclass[12pt,a4paper abstracton]{scrartcl} 
\usepackage{amssymb,amsmath,amsfonts,mathrsfs}
\usepackage[ amsmath, thmmarks, hyperref]{ntheorem}
\usepackage{mathrsfs}
\usepackage{fullpage}
\usepackage{tikz}
\usepackage{color}
\usepackage{microtype}
\usepackage{graphicx}
\usepackage[ colorlinks=true, citecolor=blue]{hyperref}
\usepackage{caption}
\theoremstyle{plain}
\theoremheaderfont{\normalfont\sffamily\bfseries}
 \theorembodyfont{\itshape}
 \theoremseparator{:}
\newtheorem{thm}{Theorem}
\usepackage{mathptmx}
\DeclareRobustCommand\hbar{\mathchar'26\mkern-9mu h}
\def\ben{\begin{equation}}
\def\een{\end{equation}}
\def\bena{\begin{eqnarray}}
\def\eena{\end{eqnarray}}

\newtheorem{lemma}{Lemma}%[section]
\newtheorem{cor}{Corollary}%[section]
\newtheorem{definition}{Definition}%[section]
\makeatletter
\newcommand{\openbox}{\leavevmode
  \hbox to.77778em{%
  \hfil\vrule
  \vbox to.675em{\hrule width.6em\vfil\hrule}%
  \vrule\hfil}}
\gdef\proofSymbol{\openbox}
\newcommand{\proofname}{Proof.}
\newcounter{proof}%
\newenvironment{proof}[1][\proofname]{
  \th@nonumberplain
  \def\theorem@headerfont{\itshape}%
  \normalfont
  \@thm{proof}{proof}{#1}}%
  {\@endtheorem}
\makeatother
\newcommand{\al}{\alpha}

\newcommand{\la}{\lambda}

\newcommand{\vp}{\varphi}
\newcommand{\La}{\Lambda}

\newcommand{\eps}{\varepsilon}
\newcommand{\Lao}{\Lambda_0}

\newcommand{\T}{\mathbb{T}}

\newcommand{\pa}{\partial}

\newcommand{\eq}{\begin{equation}}
\newcommand{\eqe}{\end{equation}}

\newcounter{saveeqn}

\newcommand{\mr}{{\mathbb R}}
\newcommand{\mn}{{\mathbb N}}

\newcommand{\C}{{\mathcal C}}
\newcommand{\e}{\operatorname{e}}

\renewcommand{\L}{{\mathcal L}}

\newcommand{\D}{{\mathcal D}}

\renewcommand{\S}{\mathscr{S}}

\renewcommand{\T}{{\mathbb T}}

\renewcommand{\O}{{\mathcal O}}

\renewcommand{\d}{{\rm d}}

\renewcommand{\D}{{\mathcal D}}

\newcommand{\tkn}{\mathfrak{T}}
\newcommand{\LIR}{M}

\newcommand{\test}{F}
\makeatletter
\newcommand*{\rom}[1]{\expandafter\@slowromancap\romannumeral #1@}
\makeatother
\title{The operator product expansion converges in massless $\varphi_{4}^{4}$-theory}
\author{Jan Holland$^{1,2}$\thanks{\tt jan.holland@uni-leipzig.de}\: , 
Stefan Hollands$^{2}$\thanks{\tt stefan.hollands@uni-leipzig.de}\: and
Christoph Kopper$^{1}$\thanks{\tt kopper@cpht.polytechnique.fr}\:
\\ \\
{\it ${}^{1}$Centre de Physique Th\'eorique, CNRS, UMR 7644} \\
{\it  \'Ecole Polytechnique, F-91128 Palaiseau, France} \medskip \\
{\it ${}^{2}$Institut f\"ur Theoretische Physik, Universit\"at Leipzig
} \\
{\it Br\"uderstr. 16, Leipzig, D-04103, Germany
}  \\
}
\begin{document}
\maketitle
\begin{abstract}
It has been shown recently \cite{Hollands:2011gf} that the mathematical
status of the operator product expansion (OPE) is
better than was expected before: namely considering
{\it massive} Euclidean $\varphi_4^4$-theory in the
perturbative loop expansion, the OPE {\it converges} at
any loop order when considering (as is usually done)
composite operator insertions into correlation functions.

In the present paper we prove the same result for the
{\it massless}  theory. While the short-distance properties
of massive and massless theories may be expected to be similar on physical
grounds, the proof in the massless case requires entirely new techniques.
In our inductive construction we have to control with sufficient precision 
the exceptional momentum singularities of the massless
correlation functions. In fact the bounds we state are expressed
in terms of weight factors associated to certain tree graphs.
Our proof is again based on the flow equations of the
renormalisation group.
\end{abstract}
\section{Introduction}
The operator product expansion (OPE) \cite{Wilson:1969ub, Zimmermann:1973wp} plays an important role in quantum field theory, both from a practical as well as a 
conceptual viewpoint. It states that
\ben\label{OPE}
\O_{A}(x)\O_{B}(y) \sim \sum_C \C_{A B}^{C}(x-y) \O_{C}(y)
\een
where $\{ \O_A \}$ denotes the collection of all local operators of a theory. The numerical coefficients $\C_{AB}^C$ depend on the theory 
under consideration and are called ``Wilson coefficients''. More precisely, the OPE is normally understood as an asymptotic short-distance expansion 
for the operator product $\O_A(x) \O_B(y)$, inserted into some correlation function with other arbitrary ``spectator'' fields. The common opinion is that, in order 
to obtain an approximation as good as $O(|x-y|^\Delta)$ as $|x-y| \to 0$, one should include all terms in the sum up to a sufficiently high operator dimension (depending on $\Delta$). 

Both for practical and conceptual purposes, one would like to understand in quantitative detail how well the OPE actually approximates correlation 
functions in concrete models. For example, one would like to know how the approximation depends on the specific choice of the spectator fields, or whether 
the OPE might even be a convergent -- rather than only asymptotic -- expansion. For this, one must have detailed quantitative bounds on the remainder 
in the OPE. In the case of {\em massive}, perturbative $\varphi^4$-theory in 4-dimensional Euclidean space, such bounds were  established in~\cite{Hollands:2011gf}. 
These bounds showed in particular that the status and range of validity of the OPE is in fact much better than originally anticipated:
The OPE (i.e. sum over $C$) actually converges, in the sense of insertions into a correlation function, for arbitrary but fixed loop order $L$, and for 
arbitrary (!) distances $|x-y|>0$.  

In practice, the OPE is mostly used in the context of asymptotically free gauge theories of Yang-Mills type, and these theories contain {\em massless} fields. At 
first sight, it might appear trivial to extend the convergence results~\cite{Hollands:2011gf} to massless theories -- such as massless $\varphi^4$-theory in the simplest case: After all, the 
OPE is supposed to capture the short distance behaviour of correlation functions, and, as is well-known, massless theories do not differ substantially in this respect from massive ones. 
The difficulty, however, arises from the fact that the OPE must be understood in the sense of an insertion into a correlation function with additional spectator fields. 
Correlation functions are sensitive also to the infrared behaviour of the theory. In the case of massless theories, this is {\em considerably} more involved and must 
be reflected by a more complicated structure of any bounds on these functions. In the present paper, we tackle this problem and are able to show that the 
OPE is still convergent in massless, Euclidean $\varphi^4$-theory. To state our first result, we introduce test functions $\test_i\in\mathcal{S}(\mathbb{R}^4)$ for $i=1, ..., N$ and Schwartz norms\footnote{These are finite for any $\test$ in the space $\mathcal{S}(\mathbb{R}^4)$ by definition.} 
\ben\label{Schwnorms}
\|F\|_{n}:=\sup_{x\in\mathbb{R}^4}|(M^2+x^2)^n F(x)|\, .
\een 
Here $M$ is some fixed renormalisation scale. 
The test functions are used to define ``averaged fields'' $\varphi(\test_i) = \int d^4 x \, \varphi(x) \test_i(x)$. Then our first main result is:

\newpage
\begin{thm}\label{thmope}
For any $N,L\in\mathbb{N}$ there exists $K>0$ such that the remainder of the operator product expansion, carried out up to operators of dimension $D=[A]+[B]+\Delta$, at $L$ loops, is bounded by
\ben
\begin{split}
&\Big| \Big\langle \O_{A}(x)\O_{B}(0)\, \varphi(\test_1)\cdots\varphi(\test_N) \Big\rangle - \sum_{C:[C] \leq D}\C_{A B}^{C}(x) \Big\langle \O_{C}(0)\, \varphi(\test_1)\cdots\varphi(\test_N) \Big\rangle  \Big|_{L-\text{loops}} \\
 &\qquad\leq  \sqrt{{[A]! [B]!}} \ (KM)^{D}\   \frac{|x|^{\Delta}}{\sqrt{\Delta!}}  \
  \LIR^{N}  \sum_{\mu=0}^{(D+2)(N+2L+3)} \sum_{\mu_1+\ldots+\mu_N=\mu   }    \frac{\prod_{i=1}^{N}\|\hat{\test}_{i}\|_{\frac{\mu_{i}}{2}}} { M^{\mu}}   %\mathcal{P}_{2L+\frac{N}{2}}\left( \log_{+} M \right)  
  \, ,
\end{split} 
\label{maineq}
\een
where % $M$ is some fixed renormalisation scale (entering the definition of the correlation functions) and where 
$[A]$ is the dimension of a composite operator $\O_A$.
\end{thm}
One can draw the following conclusions from this theorem:
\begin{enumerate}
\item The remainder of the OPE is a \emph{tempered distribution}, which is of order $O(|x|^\Delta)$ for $|x| \to 0$.  
\item If one restricts to test functions whose Fourier transforms have compact support, the theorem even implies \emph{convergence} of the OPE. This can be seen as follows: Assume that $\hat{\test}_i(p)=0$ for all $|p|\geq P$ and for all $i$. One then has the bound $\|\hat F \|_n< ({\rm cst.} \, P)^{2n}$ on the Schwartz norms. It follows that the right side of \eqref{maineq} behaves as $P^{\Delta (N+2L+3)}/\sqrt{\Delta !}$ for large $\Delta$. Therefore, 
the remainder vanishes in the limit $\Delta\to\infty\,$. 
%(the factorial $\sqrt{\Delta!}$ grows faster than any power ${\rm cst.}^\Delta$).
\item One can also deduce convergence of the OPE under slightly less restrictive conditions on the test functions. For instance, 
if $|\hat F_i(p)|$ decay more rapidly than $e^{-|p|^a}$ for some $a> 2(N+2L+3)$, convergence follows again. 
\end{enumerate}
\noindent To understand better how the remainder behaves as a function of the momenta of the spectator fields, we also derive a bound for a more restricted class of test functions. Namely, consider now test functions such that the support of $\hat \test_1(p_1) \cdots \hat \test_N(p_N)$ contains only configurations of 4-momenta $p_{1},\ldots,p_{N}$ whose magnitude is less than some $P$, and such that their distance to ``exceptional''\footnote{By an exceptional configuration one means a set $(p_1, ..., p_N)$ such that a strict subsum of momenta vanishes. Singularities at such configurations are common in massless theories.} momentum configurations is at
least $\varepsilon>0$ [see \eqref{condition}].
Under these conditions, one has:
\begin{thm}\label{OPEbound1}
The remainder of the operator product expansion, carried out up to operators of dimension $D=[A]+[B]+\Delta$, at $L$ loops, is bounded by
\ben
\begin{split}
&\Big| \Big\langle \O_{A}(x)\O_{B}(0)\, \varphi(\test_1)\cdots\varphi(\test_N) \Big\rangle- \sum_{C:[C] \leq D}\C_{A B}^{C}(x) \Big\langle \O_{C}(0)\, \varphi(\test_1)\cdots\varphi(\test_N) \Big\rangle  \Big|_{L-\text{loops}} \\
 &\leq P^{N} \sqrt{{[A]! [B]!}} \left(K\, M\,  \sup(1,\frac{P}{M})^{(N+2L+1)}\right)^{[A]+[B]}\, \left(\frac{P}{\inf(M, \epsilon)}\right)^{3N} \prod_{i}\sup|\hat{\test}_i| \\
&\times \frac{1}{\sqrt{\Delta!}} \left(K\, M\,|x|\,   \sup(1,\frac{P}{M})^{(N+2L+1)}\right)^{\Delta}\mathcal{P}_{2L+\frac{N}{2}}\left( \log_{+}\frac{P}{\inf(M,\epsilon)}\right) \, ,
\end{split} 
\label{ope31}
\een
where ${K}$ is a constant depending on $N,L$, and $\mathcal{P}_n$ is a polynomial of degree $n$
with nonnegative coefficients which depend on $N,L$.
\end{thm}
One can draw the following conclusions from this theorem:
\begin{enumerate}
\item
For large $\Delta$, the bound on the remainder behaves as  $P^{\Delta(N+2L+1)}/\sqrt{\Delta !}$. Therefore, we conclude again that the OPE converges, but
it converges more slowly for large $P$. This is physically plausible, because $P$ is an upper bound on the momentum space support of the spectator fields. Hence, 
the factorisation phenomenon exhibited by the OPE sets in more slowly if the ``state'' generated by the smeared spectator fields from the ``vacuum'' contains more ``UV modes''. 
 \item
We see that if the quantity $\epsilon$ becomes small, then the bound on the remainder is also larger due to the inverse powers such as $\eps^{-3N}$ and powers of $\log \eps$. This is also physically plausible, 
because $\epsilon$ characterises how close the momenta in the support of  the spectator fields are to becoming ``exceptional''. 
Exceptional momentum configurations are in fact well known to lead to IR singularities in massless theories, as one can see already at tree level from the fact that the propagator associated with a line
is $1/k^2$.  
\end{enumerate}
\noindent To prove theorems~\ref{thmope} and \ref{OPEbound1} we use the renormalisation group flow equation method proposed first in~\cite{Polchinski:1983gv}, and developed significantly further in~\cite{Keller:1990ej, Keller:1991bz, Keller:1992by}. This method characterises the quantities of 
interest in QFT (correlation functions, OPE coefficients, etc.) as solutions to certain flow equations, where the flow parameter $\Lambda$ plays the role of a cutoff. It is possible to establish 
bounds on the solutions of the flow equations which ultimately yield the bound on the remainder presented in the theorems. While the general strategy is thereby rather similar to that 
employed in~\cite{Hollands:2011gf} in the massive case, the structure of the actual bounds and the proofs are very different in the massless case. In fact, we need to make use of, and 
considerably extend the technique of ``tree bounds''. Such tree structures are needed in order to control the (physical) IR-singularities present in massless theories.  Tree structures of this nature have appeared previously in the study of massless theories~\cite{GK, Kopper:2001to}, and we partly rely on -- but considerably extend -- these techniques. 

We believe that our methods can also establish 
a convergence results \ref{thmope}, \ref{OPEbound1} for the OPE of gauge invariant local fields analogous to theorem~\ref{thmope} in Yang-Mills theories. For this, one should combine the tools of the present paper with the well-known BRST-method to isolate the physical degrees of freedom. 
\section{Notation and conventions}
 We use a
standard
multi-index notation.
Our multi-indices are elements $w = (w_1, \dots, w_n) \in \mn^{4n}$,
so
that each  $w_i \in \mn^4$
is a four-dimensional multi-index whose entries are $w_{i,\mu} \in \mn$
and $\mu=1,\dots,4$. We often write $\vec{p}=(p_{1},\ldots,p_{n})\in\mathbb{R}^{4n}$ for tuples of four-vectors. If $f(\vec p)$ is a smooth function on $\mr^{4n}$, we set
\ben\label{multder}
\pa_{\vec{p}}^{w} f(\vec p) = \prod_{i,\mu}
\left( {\pa \over \pa p_{i,\mu}} \right)^{w_{i,\mu}} f(\vec p)\ , 
\quad \mbox{and}\quad 
w! = \prod_{i,\mu} w_{i,\mu}! \, , \quad |w|=\sum_{i,\mu} w_{i,\mu} \, .
\een
Taking derivatives $\partial^w$ of a product of
functions
$f_1 \dots f_n$, such derivatives get distributed over the
factors,
resulting
in the sum of all terms of the form $c_{\{v_i\}} \ \partial^{v_1} f_1
\dots \partial^{v_r} f_r$. Here
each $v_i$ is now a $4n$-dimensional multi-index, where
$v_1+\dots+v_r=w$,
and where
\ben\label{cwest}
c_{\{v_i\}} = \frac{(v_1+\dots+v_r)!}{v_1! \dots v_r!} \le r^{|w|}
\een
is the associated weight factor. We also note the bound
\ben\label{cwest2}
\sum_{v_{1}+\ldots+v_{r}=w\in\mathbb{N}^{4n}}
\leq \sum_{v_{1}+\ldots+v_{r}=w\in\mathbb{N}^{4n}} c_{\{v_{i}\}} = r^{|w|}\, .
\een
 For later convenience, we introduce the shorthand notations
\ben
|\vec{p}|=|(p_{1},\ldots,p_{n})|:=\sup_{{I\subseteq \{1,\ldots,n\}  } }|\sum_{e\in I} p_{e}|\, ,
\een
\ben
|\vec{p}|_{a}=\sup(|\vec{p}|, a), \quad a\in\mathbb{R}\, ,
\een
\ben\label{etadef}
 \eta(\vec{p})=\eta(p_{1},\ldots,p_{n}):=\inf_{\substack{I\subsetneq \{1,\ldots,n\} \\ I\neq \emptyset } }|\sum_{e\in I}p_{e}|\, ,
\een
\ben\label{etabdef}
 \bar\eta(\vec{p}):=\inf_{\substack{I\subseteq \{1,\ldots,n\} \\ I\neq \emptyset } }|\sum_{e\in I}p_{e}|
\een
and
\ben\label{kappadef}
\kappa(\La,\vec{p},\LIR):=\sup(\La, \inf(\eta(\vec{p}), \LIR))\, .
\een
When considering Schwinger functions without operator insertions, we will refer to momentum configurations satisfying $\eta(\vec{p})=0$ as \emph{exceptional}, while in the case with operator insertions the corresponding condition is $\bar\eta(\vec{p})=0$.
We write
\ben
\log_{+}(x)=\log\sup(1,x)
\een
for the ``positive part of the logarithm''. For the Fourier transform, we use the convention
\ben
f(x) = \int_p  \hat f(p)\, \e^{ipx} := \int_{\mr^4} \frac{\d^4 p}{(2\pi)^4}\, 
\e^{ipx} \hat f(p) \, .
\een
\paragraph{Weighted Trees:} Our bounds on the various quantities of interest will be expressed in terms of \emph{weighted trees} (i.e. connected graphs with no loops), which will be defined\footnote{In order to shorten the text we appeal to the reader's intuition, assuming that concepts such as external/internal lines of a tree are evident. For more details see~\cite{GK}.} in the following (cf.~\cite{GK}). 
\begin{definition}\label{deftrees1}
Let $4\leq N\in\mathbb{N}$, $R \in \mathbb{N}\,$, $w\in\mathbb{N}^{4(N-1)}$ 
and $\vec{p}=(p_{1},\ldots,p_{N})\in\mathbb{R}^{4N}$ with $p_{N}=-(p_{1}+\ldots+p_{N-1})$.
Let ${\cal T}_{N,R,w}(\vec{p})\,$ be the set of all weighted trees $T=(\tau,\rho,\sigma)$
 satisfying the following properties:
 \begin{enumerate}
 \item\label{it1} The tree $\tau$ has $N$ external lines and vertices of coordination number $2,3$ or $4$.
 \item\label{it2} Each vertex of coordination number $2$ is incident to one and only one internal line.
 \item\label{it3} Denoting by $V_{2},V_{3}$ the number of vertices of coordination number $2$ and $3$ respectively, we require that
 \ben\label{vertexL}
 2V_{2}+V_{3}\leq R\, .
 \een
  \item\label{it4} To each external line we associate (bijectively) one of the ``momentum four-vectors'' $p_{i}$ in the tuple $\vec{p}=(p_{1},\ldots, p_{N})\in\mathbb{R}^{4N}$. 
  %To the remaining external line we associate the momentum $-(p_{1}+\ldots+p_{N-1})$ (momentum conservation).
  \item\label{it5} We denote by $\mathcal{I}(T)$ the set of internal lines of the tree $\tau$. We associate a momentum $k_{i}\in\mathbb{R}^{4}$ ``flowing through'' the internal line $i\in\mathcal{I}$ by demanding momentum conservation at each vertex (note that this assignment is unique).
 \item\label{it6} To every internal line $i\in\mathcal{I}(T)$ we associate a weight $\rho_{i}\in\{0,1,2\}$. The sum of these weights satisfies the rule
 \ben
\rho:= \sum_{i\in\mathcal{I}(T)} \rho_{i}= N-4\, .
 \een
 \item\label{it7} To every internal line  $i\in\mathcal{I}(T)$ we further associate a weight $\sigma_{i}\in\mathbb{N}$ such that
 \ben
 \sigma:=\sum_{i\in\mathcal{I}(T)} \sigma_{i}=|w|\, \quad,\quad \sum_{i\in\mathcal{I}(T) \atop p_{j}\text{ flows through }i} \hspace{-.5cm}\sigma_{i}=|w_{j}|\, .
 \een
 \item\label{it8} The total weight $\theta_{i}\in\mathbb{N}$ associated to an internal line $i\in\mathcal{I}(T)$ is strictly positive, i.e.
 \ben
 \theta_{i}=\rho_{i}+\sigma_{i}>0\quad \forall i\in\mathcal{I}(T)\, .
 \een
\item\label{it9} To each internal line $i$ with weight $\rho_i=1$ is associated a vertex of coordination number 3 to which this line is incident. To each internal line $i$ with weight $\rho_i=0$ is associated either
\begin{itemize}
\item a vertex of coordination number 2 to which this line is incident
\item or a pair of vertices of coordination number 3 which are connected by this line. 
\end{itemize}
In this way every vertex of coordination number smaller than 4 is 
associated to exactly one internal line with $\rho_i<2\,$.
 \end{enumerate}
\end{definition}
 \paragraph{Remark:}
 For an example of a tree in ${\cal T}_{N,R,w}$, see e.g. fig.\ref{fig:reduction1}. We note a few properties of such trees, which are proven in~\cite{GK}:
\begin{itemize}
\item Nestedness:
${\cal T}_{N,R,w} \subset {\cal T}_{N,R+1,w}\,$
\item Saturation: ${\cal T}_{N,R,w} = {\cal T}_{N,3N-2,w}\,$
for any $R \ge 3N -2$  
\item The total weight associated to the internal lines of a tree is
\ben
\theta=\sum_{i\in\mathcal{I}}\theta_{i}=N+|w|-4
\een
\item The number of internal lines in a tree $T\in{\cal T}_{N,2L,w}$ satisfies the inequality
\ben\label{inlinebound}
|\mathcal{I}(T)| \leq \frac{N-4}{2}+L\, .
\een
This follows from property \ref{it3} in definition \ref{deftrees1} combined with the general formula (see~\cite{GK})
\ben\label{INV}
|\mathcal{I}|=\frac{N-4}{2}-\sum_{n=1}^{\infty}\frac{n-4}{2}\, V_{n}\, ,
\een
where $V_{n}$ is the number of vertices of coordination number $n$.
\item Note also that the trees in ${\cal T}_{N,R,w}$ actually only depend on $(|w_{1}|,\ldots,|w_{N}|)\in\mathbb{N}^{N}$. Thus, the trees in ${\cal T}_{N,R,w}$ and in ${\cal T}_{N,R,w'}$ are the same when $|w_{i}|=|w_{i}'|$ for all $1\leq i\leq N$.
\end{itemize}
%\newpage
Our bounds for the Schwinger functions with operator insertions will be formulated in terms of trees with a ``special vertex''. 
 \begin{definition}\label{deftrees2}
 Let $R,N\in\mathbb{N}$, $w\in\mathbb{N}^{4N}$ and $\vec{p}\in\mathbb{R}^{4N}$. The elements of $\tkn_{N,R,w}(\vec{p})$ are weighted trees $T=(\tau,\rho,\sigma)$ defined as before in def.~\ref{deftrees1}, satisfying in addition the following conditions:
 \begin{itemize}
 \item The tree $\tau$ has one special vertex, called $\mathcal{V}$, of coordination number $N_{\mathcal{V}}\in\{0,\ldots,N\}$, at which momentum conservation is not imposed (all momenta ``flow into'' $\mathcal{V}$). We do not count this vertex in the requirements \ref{it1} - \ref{it3} and \ref{it9} of def.~\ref{deftrees1}.
 \item 
 We define $\vec{p}_{\mathcal{V}}\in \mathbb{R}^{4N_{\mathcal{V}}}$ as the set of momenta ``flowing into'' the vertex $\mathcal{V}$, i.e. the collection of momenta $\Big(({p}_{\mathcal{V}})_{1},\ldots,({p}_{\mathcal{V}})_{N_{\mathcal{V}}}\Big)$ associated to the  lines directly attached to the vertex $\mathcal{V}$.
 \item The requirement \ref{it6} in def.~\ref{deftrees1} is replaced by
  \ben
\rho:= \sum_{i\in\mathcal{I}(T)} \rho_{i}= N-N_{\mathcal{V}}\, ,
 \een
 and the requirement $N\geq 4$ is replaced by $N\in\mathbb{N}$.
 \end{itemize}
 \end{definition}
  \paragraph{Remark:}
See e.g. fig.\ref{fig:LoopCase3} below for an example tree in $\tkn_{N,R,w}$. We again collect a few properties of trees in $\tkn_{N,R,w}(\vec{p})$:
\begin{itemize}
\item The nestedness and saturation properties mentioned in the previous remark are unchanged. 
\item The total weight associated to the internal lines of a tree is
\ben
\theta=\sum_{i\in\mathcal{I}}\theta_{i}=N-N_{\mathcal{V}}+|w|\, .
\een
\item The number of internal lines in a tree $T\in\tkn_{N,2L,w}$ satisfies the inequality
\ben\label{internalboundtkn}
|\mathcal{I}(T)| \leq\begin{cases}  \frac{N-2}{2}+L & \text{for }N\geq 2\\
0 & \text{for }N=0
\end{cases}\, .
\een
\end{itemize}
The proof of the first two properties follows the same arguments as in the $\mathcal{T}_{N,R,w}$ case. 
The bound on the number of internal lines follows again from the general formula \eqref{INV} (the case $N=0$ is trivial in the sense that it only contains the vertex $\mathcal{V}$ and no external or internal lines).

We will later also need a bound on the number of weighted trees in the forests ${\cal T}_{N,2L,w}$ and $\tkn_{N,2L,w}$:
\begin{lemma}\label{lemtreebd}
The cardinality of ${\cal T}_{N,2L,w}(\vec{p})$ satisfies the bound
\ben\label{treecard}
|{\cal T}_{N,2L,w}(\vec{p})| \leq  N!\cdot 4^{3N-2} \cdot [3(|w|+1)]^{\frac{N-4}{2}+L}\, .
\een
\end{lemma}
\begin{proof}
5
Let $T\in{\cal T}_{N,2L,w}(\vec{p})$. Note that the vertices in our trees are ``unlabelled''. An upper bound on the number of different unlabelled trees with $V$ vertices and $N$ external lines is given by $4^{V+N-1}$, see e.g.~\cite[Theorem 8.5.1]{Lovasz:2003ua}.
 How many vertices do our trees $T$ have? Recall that we denote by $V_{i}$ the number of vertices of coordination number $i$. Using the formula (see~\cite{GK})
\ben\label{v3v4bound}
V_{3}+2V_{4}=N-2
\een
as well as (this follows from the requirement that vertices of coordination number $2$ are incident to only one internal line)
\ben\label{v2bound}
V_{2}\leq N\, ,
\een
we find that
\ben
\frac{N-2}{2}\leq V=V_{2}+V_{3}+V_{4} \leq   2N-2\, .
\een
Thus, we find that up to weightings and momentum assignments, the number of trees in ${\cal T}_{N,2L,w}(\vec{p})$ is bounded by
\ben
\sum_{V=N-2/2}^{2N-2} 4^{V+N-1}< %\frac{(4^{2N-1}-1)\cdot 4^{N-1}}{3} \leq 
4^{3N-2} \, .
\een
Let us come to the weightings of our trees. To every internal line we associate a weight $\sigma_{i}\in\{0,\ldots, |w|\}$ and a weight $\rho_{i}\in\{0,1,2\}$. The number of possible weight assignments is therefore bounded by the factor $[3(|w|+1)]^{|\mathcal{I}|}$. Using \eqref{inlinebound}, 
we see that there are at most $[3(|w|+1)]^{\frac{N-4}{2}+L}$ possible assignments of weights. Finally, we also associate to every external line of our trees one of the four momenta $p_{1},\ldots,p_{N}$, which can be realised in $N!$ distinct ways. Combining these estimates, we arrive at the bound \eqref{treecard}.
\end{proof}
\begin{lemma}\label{lemtreebd2}
The cardinality of $\tkn_{N,2L,w}(\vec{p})$ satisfies the bound
\ben\label{treecard2}
|\tkn_{N,2L,w}(\vec{p})| \leq  (N+1)!\cdot 4^{3N} \cdot [3(|w|+1)]^{\frac{|N-2|_{+}}{2}+L} \, .
\een
\end{lemma}
\begin{proof}
For $N>0$, one proceeds essentially as in the proof of lemma \ref{lemtreebd}. Due to the additional vertex $\mathcal{V}$, our bound on the number of trees in $\tkn_{N,2L,w}$ up to decorations is $4^{3N-1}$. Since one vertex in these trees is in fact labelled, we have to multiply this bound by the number of vertices $V< 2N$. We account for this factor by increasing the factorial to $(N+1)!$ and by increasing the exponent of $4$ by one. Noting that the number of internal lines for the trees $\tkn_{N,2L,w}$ is bounded by \eqref{internalboundtkn}, we obtain a factor $[3(|w|+1)]^{\frac{|N-2|}{2}+L}$ from the possible weightings of the internal lines. 

The case $N=0$ corresponds to the trivial situation where we only have the unique ``tree'' containing only the vertex $\mathcal{V}$ and no lines. 
\end{proof}
\section{The flow equation framework}\label{sec:framework}
In this paper we consider the perturbative quantum field theory of a massless scalar field with self-interaction $g\varphi^{4}$ on $4$-dimensional Euclidean space. We adopt the renormalisation group flow equation framework~\cite{Polchinski:1983gv, Wilson:1971bg,Wilson:1971dh,Wegner:1972ih}. In the following  we will give a brief review of the general formalism and define the  objects of interest for the purpose of this paper. See~\cite{Keller:1990ej, GK,Muller:2002he,Kopper:1997vg} for more comprehensive reviews of the flow equation approach, and~\cite{Hollands:2011gf, Keller:1992by,Holland:2012vw} for a discussion of the OPE in the massive case.
\subsection{Connected amputated Schwinger functions}
At first, we formulate our quantum field theory with ultraviolet (UV) cutoff $\La_{0}$ and infrared (IR) cutoff $\La$ in the standard path integral formalism. This requires two main ingredients:
\begin{enumerate}
\item We define the regularised momentum space propagator as
\ben
C^{\La,\La_{0}}(p)=\frac{1}{p^{2}}\, \left[ \exp\left(-\frac{p^{2}}{\La_{0}^{2}}\right)-\exp\left(-\frac{p^{2}}{\La^{2}}\right) \right]\, ,
\een
For ${p}^{2}\neq 0$, one checks that upon removal of the cutoffs, i.e. in the limit $\La\to 0, \La_{0}\to \infty$, we indeed recover the massless propagator $1/p^{2}$. 
\item The interaction Lagrangian is given by
\ben
L^{\Lambda_0}(\varphi,M) = \int \d^4 x \ \bigg( a^{\Lambda_0}
\, \varphi(x)^2
+b^{\Lambda_0} \, (\partial \varphi(x))^2+\left(\frac{g}{4!}+c^{\Lambda_0}\right)
\, \varphi(x)^4 \bigg) \ .
\label{ac}
\een
Here the \emph{basic field} $\varphi $ is assumed to be in the Schwartz space $\S(\mr^4)$.  The counter terms
$a^{\Lambda_0}(\hbar,M),  b^{\Lambda_0}(\hbar,M)  $
and $c^{\Lambda_0}(\hbar,M)$ are formal power series in $\, \hbar$ of order $\geq 1$, whose coefficients will be
adjusted  order by order to satisfy appropriate renormalisation conditions. They also depend on an arbitrary but fixed {\em renormalisation scale} $M>0$, which is introduced in the massless theory in order to avoid infrared divergences (see \eqref{CAGBC1}-\eqref{CAGBC2} below). 
In order to obtain a well defined limit of the quantities of interest, the counterterms need to be chosen as appropriate functions of the ultraviolet cutoff $\Lambda_0$.
\end{enumerate}
The correlation ($=$ Schwinger- $=$ $n$-point-) functions of $n$ basic fields with
cutoff are defined by the expectation values
\ben\label{pathint}
\begin{split}
 \langle \varphi(x_1) \cdots \varphi(x_n) \rangle &\equiv  \mathbb{E}_{\mu^{\Lambda,\Lambda_{0}}} \bigg[\exp \bigg( -\frac{1}{\hbar}
L^{\Lambda_0}\bigg) \, \varphi(x_1) \cdots \varphi(x_n) \bigg] \bigg/ Z^{\Lambda,\Lambda_0} \\
& =
\int \d\mu^{\Lambda,\Lambda_0} \ \exp \bigg( -\frac{1}{\hbar}
L^{\Lambda_0}\bigg) \, \varphi(x_1) \cdots \varphi(x_n)\bigg/ Z^{\Lambda,\Lambda_0} \, .
\end{split}
\een
This expression is simply the standard Euclidean path-integral, but with the free part in the Lagrangian absorbed into the normalised Gaussian measure\footnote{
See the Appendix to Part I of~\cite{glimm} for mathematical details about Gaussian functional integrals. 
} $\d\mu^{\La,\Lao}$ of covariance $\hbar C^{\La,\La_{0}}$.
The normalisation factor $Z^{\Lambda,\Lambda_0}$
is chosen so that $\langle 1 \rangle = 1$. In the perturbative approach to quantum field theory, which we will follow in this paper, the exponentials in the path integral are expanded out and the
Gaussian integrals are then performed. In this way, we obtain a formal series in $\,\hbar$. But we note that, for finite values of the cutoffs $0<\Lambda<\Lambda_0<\infty\, $ and on imposing a finite (space) volume, the functional integral 
\eqref{pathint} exists in the non-perturbative sense. In the perturbative theory it is shown that one can remove the cutoffs, $\Lambda_0 \to
\infty$ and $\Lambda \to 0$, for a suitable choice of the running couplings at each given but fixed order in $\,\hbar$. 
The correct behaviour of these couplings is determined, in the flow equation  framework, by deriving
first a differential equation in the parameter
$\Lambda$ for the Schwinger functions,
and by then defining the running couplings implicitly through the boundary
conditions for this equation.

These differential equations, referred to from now on as \emph{flow equations}, are written more conveniently in
terms of the hierarchy of ``connected, amputated
Schwinger functions'' (CAS's),
whose generating functional is given by the following 
convolution\footnote{The convolution is defined in general by
$(\mu^{\Lambda,\Lambda_0} \star F)(\varphi) =
\int \d\mu^{\Lambda,\Lambda_0}(\varphi') \ F(\varphi+\varphi')$.}
of the Gaussian measure with the exponentiated interaction,
\ben\label{CAGdef}
-L^{\Lambda, \Lambda_0} := \hbar \, \log \left[ \mu^{\Lambda,\Lambda_0}
\star \exp \bigg(-\frac{1}{\hbar} L^{\Lambda_0}
\bigg)\right] - \hbar \log  Z^{\La,\Lao} \ .
\een
The full Schwinger functions can be recovered from the CAS's in the end. 
One can expand the functionals $L^{\Lambda,\Lambda_0}$ as formal power series in terms of Feynman diagrams with $L$ loops, $N$ external legs and propagator $C^{\Lambda,\Lambda_{0}}(p)$. One can show that, indeed, only connected diagrams contribute,
and the (free) propagators on the external legs are removed. While we will 
not use diagrammatic decompositions in terms of Feynman diagrams here, we will analyse the functional (\ref{CAGdef})
in the sense of formal power series (in momentum space)
\ben\label{genfunc}
L^{\Lambda, \Lambda_0}(\varphi,M) := \sum_{N>0}^\infty
\sum_{L=0}^\infty {\hbar^L}
\int \frac{\d^4 p_1}{(2\pi)^{4}} \dots \frac{\d^4 p_N}{(2\pi)^{4}}\ \bar\L^{\Lambda,\Lambda_0}_{N,L}(p_1, \dots, p_N; M)
\,
\hat\varphi(p_1) \cdots \hat\varphi(p_N) \, .
\een
No
statement
is made about the
convergence of these series. %in  $\, \hbar$.
Translation invariance of the connected amputated functions in position space implies that
% their Fourier transforms, denoted
the functions 
$\bar{\L}^{\La,\Lao}_{N,L}(p_1, \dots, p_N; M)$ are supported
at $p_1+\ldots+p_N=0$ (momentum conservation), and thus only depend on $N-1$ independent momenta. We write
%, by a slight abuse of notation
\ben\label{deltaCAG}
\bar{\L}^{\Lambda,\Lambda_0}_{N,L}(p_1, \dots, p_N; M) = \frac{\delta^{4}{(\sum_{i=1}^N
p_i)}}{(2\pi)^{4}}
\, \L^{\Lambda,\Lambda_0}_{N,L}(p_1, \dots, p_{N}; M) \, .
\een
In the following we will use the convention that the variable $p_{N}$ is determined in terms of the remaining $N-1$
 four-vectors by momentum conservation, i.e. $p_{N}=-p_{1}-\ldots -p_{N-1}$.
  One should keep in mind, however, that the functions $\bar{\L}^{\La,\Lao}_{N,L}(p_1, \dots, p_N; M)$ are in fact fully symmetric under permutation of $p_{1},\ldots, p_{N}$.

To obtain the flow equations for the CAS's, we take the $\La$-derivative of
eq.\eqref{CAGdef}:
\ben
\partial_{\La} L^{\La,\Lao} \,=\,
\frac{\hbar}{2}\,
\langle\frac{\delta}{\delta \vp},\dot {C}^{\La}\star
\frac{\delta}{\delta \vp}\rangle L^{\La,\Lao}
\,-\,
\frac{1}{2}\, \langle \frac{\delta}{\delta
  \vp} L^{\La,\Lao},
\dot {C}^{\La}\star
\frac{\delta}{\delta \vp} L^{\La,\Lao}\rangle  +
\hbar \partial_\Lambda \log Z^{\La,\Lao} \ .
\label{fe}
\een
Here we use the following notation:
 We write $\,\dot {C}^{\La}\,$ for the derivative 
$\partial_{\La} {C}^{\La,\Lao}\,$, which, as we note,
does not depend on $\Lao$. Further, by $\langle\ ,\  \rangle$ we denote the standard scalar product in
$L^2(\mathbb{R}^4, \d^4 x)\,$, and $\star$ stands for
convolution in $\mr^4$. As an example, 
\ben
\langle\frac{\delta}{\delta \vp},\dot {C}^{\La}\star
\frac{\delta}{\delta \vp}\rangle = \int \d^4x\, \d^4y \ \dot {C}^{\La}(x-y)
\frac{\delta}{\delta \varphi(x)} \frac{\delta}{\delta \varphi(y)}
\een
is the ``functional Laplace operator''. We can also write the flow equation \eqref{fe} in an expanded version as
\ben\label{CAGFEexpand}
\begin{split}
&\partial_{\Lambda}\L^{\Lambda,\Lambda_{0}}_{N,L}(p_{1},\ldots,p_{N};M)= \left({N+2}\atop{2}\right) \, \int_{k} \dot{C}^{\Lambda}(k)\L^{\Lambda,\Lambda_{0}}_{N+2,L-1}( k, -k,  p_{1},\ldots,p_{N}; M)\\
&-\frac{1}{2}\sum_{\substack{l_{1}+l_{2}=L \\ n_{1}+n_{2}=N+2}} n_{1}n_{2}\,  \mathbb{S}\, \left[  \L^{\Lambda,\Lambda_{0}}_{n_{1},l_{1}}(  p_{1},\ldots,p_{n_{1}-1}, q; M)\,  \dot{C}^{\Lambda}(q)\,   \L^{\Lambda,\Lambda_{0}}_{n_{2},l_{2}}(-q,  p_{n_{1}},\ldots,p_{N}; M) \right]\, ,
\end{split}
\een
with $q=p_{n_{1}}+\ldots+p_{N}=-p_{1} -\ldots -p_{n_{1}-1})$ and where $\mathbb{S}$ is the symmetrisation operator acting on functions of the momenta $(p_{1},\ldots, p_{N})$ by taking the mean value over all permutations $\pi$ of $1,\ldots, N$ satisfying $\pi(1)<\pi(2)<\ldots<\pi(n_{1}-1)$ and $\pi(n_{1})<\ldots<\pi(N)$. We also note that for the theory proposed through \eqref{ac}, only even moments of the CAS's (i.e. even $N$) will be non-vanishing due to the symmetry $\varphi\to-\varphi$. Further, it follows from \eqref{ac} and \eqref{CAGFEexpand} that $\L^{\La,\La_{0}}_{2,0}$ vanishes identically. 

The CAS's are defined uniquely as a solution to this differential equation only after we impose suitable boundary
conditions, which are fixed by adjusting the counter terms $a^{\La_{0}},b^{\La_{0}},c^{\La_{0}}$ in the Lagrangian.  In this paper, we make the following choice (cf.~\cite{GK}):
\begin{align}
\L_{2,L}^{0,\La_{0}}(\vec 0;M)&=0  \label{CAGBC1}\\
\partial_{p_{\mu}}\partial_{p_{\nu}}\L_{2,L}^{0,\La_{0}}(M e_0; M)&=
0  \label{CAGBC3} \\
\L_{4,L}^{0,\La_{0}}(Me_1, Me_2, Me_3;M)&=\frac{g}{4!}\delta_{L,0}\label{CAGBC4} \\
\partial_{\vec{p}}^{w}\L_{N,L}^{\La_{0},\La_{0}}(\vec{p};M)&=0 &\text{ for }N+|w|> 4 \label{CAGBC2}
\end{align}
Here the parameter $M>0$ specifies a \emph{renormalisation scale},  $e_0 \in\mathbb{R}^{4}\setminus\{0\}$ and $(e_{1},e_{2},e_{3},-(e_{1}+e_{2}+e_{3}))\in\mathbb{R}^{4\times 4}$ are unit vectors satisfying $\eta(\vec{e})>0$ (we refer to such combinations of momenta as ``non-exceptional'')\footnote{Changing $e_0, ..., e_3$ corresponds to a finite change of the renormalisation conditions imposed. The dependence of the Schwinger functions on these vectors is not explicitly shown in our notation.}. The parameter $M$ is not needed in the case of massive fields, where one can simply impose renormalisation conditions at zero momentum. This is not possible in the massless case, where Schwinger functions will generally be divergent at vanishing external momentum. 

The last boundary condition, \eqref{CAGBC2}, simply follows by noting that $L^{\Lambda_{0},\Lambda_{0}}=L^{\Lambda_{0}}$, see \eqref{CAGdef}, and by recalling the definition of the interaction $L^{\Lambda_{0}}$, \eqref{ac}. As will turn out in \eqref{N2estweak}, the condition \eqref{CAGBC1} is necessary in order to guarantee IR-finiteness of the massless theory. 

The  CAS's are then uniquely determined by integrating the
flow equations subject to these boundary conditions, see e.g.~\cite{Keller:1990ej,Muller:2002he}. The existence of the limits $\La \to 0, \La_0 \to \infty$ follows from the 
uniform bounds given below\footnote{More precisely, we show in the present paper that Schwinger functions are bounded in the limit $\La \to 0, \La_0 \to \infty$. To show that the limit is in fact convergent (i.e. does not oscillate within the bounds), one has to derive in addition bounds on the $\La_{0}$-derivative of the CAS's. This has been carried out in \cite{GK}.}. 
\subsection{Composite field insertions}\label{subsec:compfields}
In the previous section we have defined Schwinger functions of products of the basic field. We now turn to the so called \emph{composite operators}, or composite fields, which are given by the monomials
\ben\label{compop}
\O_{A}= \partial^{w_{1}}\varphi\cdots \partial^{w_{N}}\varphi\, , \quad A=\{N,w\}\, .
\een
Here $w=(w_{1},\ldots,w_{N})\in\mathbb{N}^{4N}$ is a multi-index [see \eqref{multder}], and we define the canonical dimension of such a field by
\ben
[A]:= N+\sum_{i}|w_{i}| \, .
\een
The Schwinger functions with insertions of composite operators $\O_{A_{1}},\ldots,\O_{A_{r}}$ are obtained by replacing the action $L^{\Lambda_0}$ with an
action containing additional sources, expressed through smooth functionals.
Particular examples of such functionals are {\em local} ones. We consider local functionals
\ben\label{func}
F (\varphi)= \sum_{i=1}^{r} \int \d^4 x \  \O_{A_{i}}(x) \ f^{A_{i}}(x) \,\, ,
\quad f^{A_{i}} \in C^\infty(\mr^4) \, ,
\een
where the composite operators $\O_{A_{i}}$ are as in eq.~\eqref{compop}. We now modify the action $L^{\Lambda_0}$ by adding sources $f^A$ as follows:
\ben\label{LFdef1}
L^{\Lambda_0}\to L^{\Lao}_F:=L^{\Lambda_0}+ F + \sum_{j=0}^{r}
B^{\Lambda_0}_j(\underbrace{F , \cdots , F}_j) 
\een
The expression $B^{\Lambda_0}_j$ represents the counter terms which are needed to eliminate the additional divergences arising from composite field insertions in the limit $\Lambda_0 \to \infty$. It is a symmetric, multilinear map acting on local
functionals of the type~\eqref{func}, and returns again a local functional.
More explicitly, we can write (denoting by $I$ subsets of $\{1,...,r\}$
and writing $\vec{x}_{I}=(x_{i})_{i\in I}\in\mathbb{R}^{4|I|},
\vec{A}_I=(A_{i})_{i\in I}$ as well
as $\langle \vec x \rangle_I=|I|^{-1} \sum_{i \in I} x_i$ for the ``center of mass'')
\ben
\label{countert}
B^{\La_{0}}_{j}(F,\dots,F)
:= \sum_{\substack{I\subseteq\{1,\ldots, r\} \\ |I|=j  }}\, \prod_{i\in I} \int\d^{4} x_{i}\, f^{A_{i}}(x_{i}) \sum_{[C]\leq \sum_{i\in I}[A_{i}]}
b^{\Lambda_0, C}_{\vec{A}_I} (\vec{x}_{I})\ \O_{C}(\langle \vec{x}_I \rangle) %= O(\hbar)
\een
for the counter terms.  The $j=2$ term, for example, is of the form
\ben
B^{\La_{0}}_{2}(F, F)=
\sum_{\{i,j\}\subseteq\{1,\ldots,r\}}\int\d^{4} x_{i}\int\d^{4}x_{j}\, f^{A_{i}}(x_{i})f^{A_{j}}(x_{j}) \sum_{[C]\leq [A_{i}]+[A_{j}]}  b^{\La_{0},C}_{A_{i}A_{j}}(x_{i}-x_{j})\, \O_{C}(\frac{x_{i}+x_{j}}{2})\, .
\een
The counter terms are defined implicitly below through a flow equation and boundary conditions, see eqs.~\eqref{BCunsub}, \eqref{BCL1} and \eqref{BCL2}. 
They are (at least) of order $\,\hbar\,$.
To obtain the Schwinger functions with insertion of the composite operators $\O_{A_{1}},\ldots,\O_{A_{r}}$, we now simply take functional
derivatives with respect to the sources, setting the sources
$f^{A_i} =0$ afterwards:
\ben
\begin{split}
&\langle \O_{A_1}(x_1) \cdots \O_{A_r}(x_r)  \rangle :=
\frac{(-\hbar)^r\delta^r}{\delta f^{A_1}(x_1) \dots \delta f^{A_r}(x_r)}
\ (Z^{\La,\Lao})^{-1} \int \d\mu^{\Lambda,\Lambda_0}
\exp \bigg(-\frac{1}{\hbar}  L^{\Lambda_0}_F(\varphi)
\bigg)\biggr|_{ f^{A_i}=0} \, .
\end{split}
\een
Note that the Schwinger functions from \eqref{pathint} are a special case
of this equation; there we
take $F = \int \d^4 x \ f(x) \ \varphi(x)$, and we have
$B^{\Lambda_0}_j=0$, because no extra counter terms
are required for this insertion. As above, we can define
a corresponding effective action
as
\ben\label{LFins}
-L^{\Lambda,\Lambda_0}_F := \hbar \, \log \, \mu^{\Lambda,\Lambda_0}
\star \exp \bigg(-\frac{1}{\hbar} ( L^{\Lambda_0}
+ F + \sum_{j=0}^\infty B^{\Lambda_0}_j(F^{\otimes j})) \bigg)
- \hbar\log Z^{\La,\Lao}
\een
which now depends on the sources $f^{A_i}$, as well as on $\varphi$. From this modified effective action we determine the generating functionals of the CAS's with
$r$ operator insertions:
\ben
L^{\Lambda,\Lambda_0}(\O_{A_1}(x_1) \otimes \dots \otimes \O_{A_r}(x_r); M)
:=
\frac{\delta^r \ L^{\Lambda,\Lambda_0}_F}{\delta f^{A_1}(x_1) \dots \delta
  f^{A_r}(x_r)}
\,  \Bigg|_{f^{A_i} =  0} \, .
\een
The CAS's with insertions defined this way are multi-linear, as indicated by the
tensor product notation, and symmetric in the insertions. We can also expand the CAS's with insertions in $\hat\varphi$ and $\,\hbar$ again (in momentum space):
\ben
L^{\Lambda,\Lambda_0} \bigg( \bigotimes_{i=1}^r \O_{A_i}(x_i); M \bigg) =\sum_{N,L \ge 0} {\hbar^L} \int \frac{\d^4p_1}{(2\pi)^{4}}\dots \frac{\d^4p_N}{(2\pi)^{4}} \
\L^{\Lambda,\Lambda_0}_{N,L}\bigg( \bigotimes_{i=1}^r \O_{A_i}(x_i);
p_1,
\dots, p_N; M \bigg)
\prod_{j=1}^N \hat{\vp}(p_j) \, .
\een
Due to the insertions in $\L_{N,L}^{\La,\Lao} (\otimes_i
\O_{A_i}(x_i),
\vec p; M)$, there is no restriction on the momentum set
$\vec p$ in this case. Translation invariance, however, implies that the CAS's 
with insertions at a translated set of points $x_j + y\,$ are
obtained
from those
at $y=0\,$ through multiplication by $\e^{iy\sum_{i=1}^n p_i}$, i.e.
\ben\label{CAGtrans}
\L^{\Lambda,\Lambda_0}_{N,L}\bigg( \bigotimes_{i=1}^r \O_{A_i}(x_i+y);
p_1,
\dots, p_N; M \bigg)=
\e^{iy\sum_{i=1}^N p_i}\, \L^{\Lambda,\Lambda_0}_{N,L}\bigg( \bigotimes_{i=1}^r \O_{A_i}(x_i);
p_1,
\dots, p_N; M \bigg)\, .
\een
From \eqref{LFins} we can determine the flow equation for CAS's with insertions. In this paper, we are interested in CAS's with one or two operator insertions. For those cases, the flow equation reads~\cite{Keller:1991bz,Keller:1992by}
\ben
\partial_{\La} L^{\La,\Lao}(\O_{A}) \,=\,
\frac{\hbar}{2}\,
\langle\frac{\delta}{\delta \vp},\dot {C}^{\La}\star
\frac{\delta}{\delta \vp}\rangle L^{\La,\Lao}(\O_{A})
\,-\,
 \langle \frac{\delta}{\delta
  \vp} L^{\La,\Lao}(\O_{A}),
\dot {C}^{\La}\star
\frac{\delta}{\delta \vp} L^{\La,\Lao}\rangle  +\hbar
 \partial_\Lambda \log Z^{\La,\Lao}
\label{FE1ins}
\een
and
\ben
\begin{split}
\partial_{\La} L^{\La,\Lao}(\O_{A}\otimes\O_{B}) \,&=\,
\frac{\hbar}{2}\,
\langle\frac{\delta}{\delta \vp},\dot {C}^{\La}\star
\frac{\delta}{\delta \vp}\rangle L^{\La,\Lao}(\O_{A}\otimes\O_{B})\\
&-\,
 \langle \frac{\delta}{\delta
  \vp} L^{\La,\Lao}(\O_{A}\otimes\O_{B}),
\dot {C}^{\La}\star
\frac{\delta}{\delta \vp} L^{\La,\Lao}\rangle \\
&-\,
 \langle \frac{\delta}{\delta
  \vp} L^{\La,\Lao}(\O_{A}),
\dot {C}^{\La}\star
\frac{\delta}{\delta \vp} L^{\La,\Lao}(\O_{B})\rangle  + \hbar
 \partial_\Lambda \log Z^{\La,\Lao} \ .
\label{FE2ins}
\end{split}
\een
To complete the definition of the CAS's with insertions, we again have to specify boundary conditions on the corresponding flow equation.  For CAS's with one insertion we choose the convention (here $A=\{N',w'\}$)
\ben\label{BCL1}
\partial^{w}_{\vec{p}}\L^{M,\La_{0}}_{N,L}(\O_{A}(0); \vec{0}; M)= i^{|w|}w! \delta_{w,w'}\delta_{N,N'}\delta_{L,0} \quad \text{ for }N+|w|\leq [A]
\een
\ben\label{BCL2}
\partial^{w}_{\vec{p}}\L^{\La_{0},\La_{0}}_{N,L}(\O_{A}(0); \vec{p}; M)=0\quad \text{ for }N+|w|>[A] \quad .
\een
Our freedom to choose boundary conditions different from \eqref{BCL1} can be seen to correspond to composite field redefinitions. Due to the linearity of the flow equation \eqref{FE1ins}, linear superpositions of CAS's with one insertion are again solutions to the system of flow equations, with boundary values given by the corresponding superpositions. 

The simplest choice of boundary conditions in the case of two insertions is
\ben\label{BCunsub}
\partial^{w}_{\vec{p}}\L^{\La_{0},\La_{0}}_{N,L}(\O_{A}(x)\otimes\O_{B}(0); \vec{p}; M)=0\quad \text{for all } w,N,L.
\quad
\een
Imposing these boundary conditions means that no regularising counter
terms for the corresponding operator product are introduced. We also define \emph{subtracted} 
%(''oversubtracted'') 
versions of the CAS functions with two insertions, i.e. versions which are more regular on the diagonal $x=0$. These ``normal products'' also satisfy the flow equation \eqref{FE2ins}, but subject to the boundary conditions
\ben\label{BCL2ins1A}
\partial^{w}_{\vec{p}}\L^{M,\La_{0}}_{D,N,L}(\O_{A}(x)\otimes\O_{B}(0); \vec{0}; M)= 0\quad \text{ for }N+|w|\leq D
\een
\ben\label{BCL2ins2A}
\partial^{w}_{\vec{p}}\L^{\La_{0},\La_{0}}_{D,N,L}(\O_{A}(x)\otimes\O_{B}(0); \vec{p}; M)=0\quad \text{ for }N+|w|>D \quad .
\een
The parameter $D\geq -1$ controls the degree of oversubtraction. For $D=-1$, the normal products coincide with the previously defined Schwinger functions with two insertions. 

In \eqref{BCL1} and \eqref{BCL2ins1A}, we have fixed the renormalisation conditions at the finite scale $\La=M>0$. This differs slightly from the conventions used in the massive case (see~\cite{Hollands:2011gf}), where one usually specifies renormalisation conditions at $\La=0$. As long as no infrared singularities appear, both schemes are strictly equivalent. The boundary conditions at $\La=0$ may be calculated in terms of those at $\La=M$ and vice versa. However, since CAS's at $\La=0$ may diverge at exceptional momenta, in particular at zero momentum, it is technically simpler and safer to impose the boundary conditions at $\La=M>0$ and $\vec{p}=0$. Then we may still calculate the boundary values at $\La=0$ and any non-exceptional momentum configuration as a function of those given in \eqref{BCL1}-\eqref{BCL2ins2A}.

To conclude this section, we state useful identities of the CAS's with insertions, which go by the name of \emph{Lowenstein rules} in the literature. Namely, it is known that~\cite{ Keller:1991bz,Keller:1992by, low}
\ben\label{Low1}
\partial_{x}^{w}\, L^{\La,\La_{0}}(\O_{A}(x)) =  L^{\La,\La_{0}}(\partial_{x}^{w}\O_{A}(x))
\een
and
\ben\label{Low2}
\partial_{x}^{w}\, L_{D}^{\La,\La_{0}}(\O_{A}(x)\otimes\O_{B}(0)) =  L_{D}^{\La,\La_{0}}(\partial_{x}^{w}\O_{A}(x)\otimes\O_{B}(0)) \, .
\een
The relations are non-trivial, since the $x$-derivatives act on the Schwinger functions on the left hand side, whereas they act on the composite operators $\O_{A}$ themselves on the right hand side. To prove these equations, one has to show that both sides satisfy the same linear flow equations and boundary conditions. This can be verified using the definitions above as well as the translation property \eqref{CAGtrans}.
\subsection{The Operator Product Expansion}
We are finally ready to give the definition of the  OPE coefficients in our theory. To have a more compact notation, let us define the operator $\D^{A}$ acting on differentiable and sufficiently regular functionals $F(\varphi)$ of Schwartz space functions $\varphi\in\S(\mathbb{R}^{4})$ by
\ben\label{defD}
\D^{A} F(\varphi) = \left. \frac{(-i)^{|w|}}{N!\,w!}\, \partial_{\vec{p}}^{w}\frac{\delta^{N}}{\delta\hat\varphi(p_{1})\cdots\delta\hat\varphi(p_{N})}\, F(\varphi)\, \right|_{\hat\varphi=0, \vec{p}=0}\quad ,
\een
where $A=\{N,w\}$. Further, let us also define the Taylor expansion operator
\ben\label{defT}
\T^{j}_{x}\, f(x)=\sum_{|w|=j}\, \frac{{x}^{w}}{w!}\, \partial^{w}f({0})
\een
where $f$ is a sufficiently smooth function on $\mathbb{R}^{4}$.  
\begin{definition}[OPE coefficients~\cite{Hollands:2011gf, Keller:1992by}]\label{defOPE}
Let $\Delta:=[C]-([A]+[B])$. The OPE coefficients are defined as
\ben\label{OPEhigh}
\C_{AB}^{C}(x)\,
:=\, \D^{C}\left\{(1-\sum_{j < \Delta}\T^{j}_{{x}}\,) \Big[L^{M,\La_{0}}(\O_{A}(x))L^{M,\La_{0}}(\O_{B}(0)) - L^{M,\Lambda_{0}}_{[C]-1}\left( \O_{A}(x)\otimes\O_{B}(0)\right)  \Big] \right\}\, ,
\een
where $M>0$ is the renormalisation scale introduced above.
\end{definition}
This definition is ultimately motivated by thm.~\ref{OPEbound}.
%
%\newpage
\section{Bounds on Schwinger functions}\label{sec:bounds}
In this section, we prove bounds on Schwinger functions with up to two operator insertions, using an induction scheme based on the renormalisation group flow equations~\cite{Polchinski:1983gv,Keller:1990ej,Muller:2002he}. These bounds are uniform in the cutoffs and imply in particular that it is safe to take the limits $\La \to 0, \La_0 \to \infty$. 
To obtain the bounds for the connected amputated Schwinger functions without insertions, one basically integrates \eqref{CAGFEexpand} over $\Lambda$ for increasing values of $N+2L$ and for given $N+2L$ for increasing values $L$. In this way, the right hand side of the equation is always known inductively. The ``integration constants'' are fixed through the boundary conditions
on the CASs. Since these were given at different values of $\La$ (see \eqref{CAGBC1} - \eqref{CAGBC2}), we have to distinguish the cases $N+|w|\leq 4$ (``relevant terms'') and $N+|w|>4$ (``irrelevant terms'') and adjust the limits of integration accordingly (integrate from $0$ to $\La$ in the former-, and from $\La$ to $\La_{0}$ in the latter case). The procedure for
CAS's with insertions is a relatively straightforward extension of the same idea: The CAS's without insertions enter the construction of CAS's with one insertion via their flow equation, which in turn serve as an input for CAS's with two insertions, etc. 
 
In either case, one wants to take the limits $\La \to 0, \La_0 \to \infty$, and for this one must have suitable uniform bounds on the CASs. The above inductive procedure in principle gives 
a means to obtain these if one can show that the form of the bound is reproduced in the induction step. In massless theories, the behavior of the CAS's in momentum space is more complicated owing to their more complicated IR behaviour, and the bounds will clearly have to take this into account. In particular, straightforwardly taking the limit $m\to 0$ in the known bounds for massive fields such as~\cite{Hollands:2011gf,Keller:1990ej,Keller:1991bz,Keller:1992by,Holland:2012vw} does not work, and a more refined control over the momentum dependence of the Schwinger functions is needed in the massless case. Our methodology for obtaining such bounds on the CAS's with and without insertions can be viewed as an extension of the ideas developed in~\cite{GK, Kopper:2001to}, relying also on~\cite{Hollands:2011gf, Keller:1991bz}. Our results are presented in sections~\ref{subsec:noins} (no insertions),~\ref{sec:1ins} (one insertion) and~\ref{sec:2ins} (two insertions). 

In the following we will set $g=1$ for convenience. The dependence of the Schwinger functions on the coupling constant can in general be  determined by relatively simple considerations. For example, the connected amputated Schwinger functions without insertions $\L_{N,L}^{\La,\La_{0}}$ carry a power of $g^{\frac{N-2}{2}+L}$.
\subsection{Schwinger functions without operator insertion}\label{subsec:noins}
In all what follows, $\LIR > 0\,$ is the renormalisation scale
 and 
 $\mathcal{P}_{L}$ are - each time they appear possibly new - polynomials of degree $L$ with non-negative coefficients depending on $N$ and $L$. $ {\cal T}_{N,2L,w}(\vec{p})$ are the trees specified in def.~\ref{deftrees1},  $k_i\in\mathbb{R}^{4}$ is the momentum flowing through the internal line $i\in\mathcal{I}(T)$ (i.e. the unique momentum associated to that line by momentum conservation at the vertices), 
and $\theta_{i} > 0$ is the total weight associated to the line $i\,$.  The weight function $f_{\La}(k;N,L,\theta)\,$ is defined as
\eq\label{fdef}
f_{\La}(k;N,L,\theta) =   e^{-\frac{k^2}{ \al(N,L) \La^2}}\ \La^{-\theta-1}\, ,
% \sum_{\nu =0}^{\theta-\rho} \sqrt{(\theta-\rho-\nu)!} \ (\frac{|k|}{\La})^{\nu} \ ,
\eqe
where
\eq\label{alphadef}
\al(N,L)= 2\left(\frac{N}{2}+2L\right)^2\ .
\eqe
Our result is:
\begin{thm}\label{thm1}
For $N=2$ and any $L\geq 1$ and $w\in\mathbb{N}^{4}$ with $|w|>2$ there exists a constant $K>0$ such that 
\ben\label{boundN2}
|\pa^w_{\vec{p}} {\cal L}^{\La,\Lao}_{2,L}(\vec{p};\LIR)| 
\le    \, \sqrt{|w|!}\ K^{(4L-2)(|w|+1)}   \int_{\La}^{\La_{0}}\d\lambda\ \frac{e^{-\frac{p^{2}}{\alpha(2,L)\lambda^{2}}} }{\lambda^{|w|-1}} \mathcal{P}_{L}\left( \log_{+}\sup\left(\frac{|p|_\LIR}{\lambda},
 \frac{\lambda}{\LIR}\right)\right)\, .
\een
For any $N\geq4$, any $L\geq 0$ and any multi-index
$ w\in\mathbb{N}^{4(N-1)}$ satisfying $N+|w|> 4$ there exists $K>0$ such that
\ben\label{mainbound}
\begin{split}
|\pa^w_{\vec{p}} {\cal L}^{\La,\Lao}_{N,L}(\vec{p};\LIR)| 
&\le  \sqrt{|w|!}\ K^{(N+4L-4)(|w|+1)}\,   \sum_{T \in {\cal T}_{N,2L,w}(\vec{p})}
\\
&\times  %\La_{\ma(\vec{a},T)+1}^{-S} % \prod_{b\in\vec{a}} \La_{b+1}^{-b+ d_{b,\vec{a}}} 
 \prod_{i \in \mathcal{I}(T)} \int_{\La}^{\La_{0}} \d\la_{i}\ f_{\lambda_{i}}(k_i;N,L,\theta_i) 
\mathcal{P}_{L}\left( \log_{+}\sup\left(\frac{|\vec{p}|_\LIR}{\min_{j\in\mathcal{I}}\lambda_{j}},
 \frac{\max_{j\in\mathcal{I}}\lambda_{j}}{\LIR}\right)\right) \, ,
\end{split}
\een
where $  0 \leq \La \le \Lao$
 and $\vec{p}=(p_{1},\ldots,p_{N}) \in \mathbb{R}^{4N}\,$  with $p_{N}=-(p_{1}+\ldots+p_{N-1})$. 
 \end{thm}
\begin{proof}[Proof of theorem \ref{thm1}:]
To begin with, we recall that CAS functions with odd $N$ vanish identically due to the symmetry $\varphi\to - \varphi$, so the bound \eqref{mainbound} is satisfied trivially in that case. We therefore assume $N$ to be even in the following.  As mentioned above, we prove the claimed bounds inductively, using a standard scheme based on the flow equation~\eqref{CAGFEexpand}~\cite{Polchinski:1983gv,Keller:1990ej,GK,Muller:2002he}.  The induction goes up in $N+2L\geq 2$, for given $N+2L$ ascends in $L$, and for given $N,L$ descends in $|w|$.

Theorem \ref{thm1} states bounds only for values of the parameters satisfying $N+|w|>4$ (``irrelevant terms''). In order to perform the induction, we  will also need bounds for the remaining values of the parameters $N,w$, i.e. for $N+|w|\leq 4$ (``relevant terms''). For this purpose, we will make use of the bounds established in~\cite{GK}, which state that 
\ben\label{boundrel1}
\begin{split}
| {\cal L}^{\La,\Lao}_{4,L}(\vec{p};\LIR)| 
&\le  \,  
\mathcal{P}_{L}\left( \log_{+}\sup\left(\frac{|\vec{p}|_\LIR}{\kappa(\La,\vec{p}, \LIR)},
 \frac{\La}{\LIR}\right)\right)
\end{split}
\een
as well as for $|w|\leq 2$
\ben\label{boundrel2}
\begin{split}
|\partial^{w}_{{p}} {\cal L}^{\La,\Lao}_{2,L}(\vec{p};\LIR)| 
&\le  \,  |p|^{2-|w|}_{\La}
\mathcal{P}_{L-1}\left( \log_{+}\sup\left(\frac{|{p}|_\LIR}{\kappa(\La,{p}, \LIR)},
 \frac{\La}{\LIR}\right)\right)\, ,
\end{split}
\een
with $\kappa(\vec{p},\La,\LIR)$ as defined in eq.\eqref{kappadef}. %and $\vec{p} = (p,-p)$ for $N=2$.
%:=\sup(\La, \inf(\eta(\vec{p}), \LIR))$, with $\eta(\vec{p}):=\inf_{e}|\sum_{e}p_{e}|$. 
%The degrees of the polynomials satisfy the bound $\operatorname{deg}\mathcal{P}_{L}\leq L$. 
Note that this bound implies vanishing of the two point function at zero momentum as $\La\to 0$.

To start the induction, we recall that for $N+2L=2$ the $0$-loop two point function vanishes identically. We then fix an arbitrary $|w_{0}|>0$ and proceed in the order indicated above considering $|w|\leq |w_{0}|$.
To verify the induction step we integrate both terms on the r.h.s. of the flow equation over $\la$ between $\La$ and $\La_{0}$ and use the induction hypothesis (i.e. the bounds stated in the theorem) in order to prove bounds consistent with theorem \ref{thm1}.
\subsubsection{First term on the r.h.s. of the flow equation \eqref{CAGFEexpand}}
Since our inductive bounds in theorem \ref{thm1} have a different structure for the cases $N\geq 4$ and $N=2$, we also have to distinguish these cases in our induction. We may restrict to values $L\geq 1$ for this term.
\paragraph{The case $N\geq 4, |w|>0 $:}
Integrating the flow equation over $\la$ between $\Lambda$ and $\Lambda_{0}$, inserting our inductive bound, \eqref{mainbound}, for the first term on the r.h.s. of the flow equation and using also the formula
\eq
{\dot C}^{\La}(p) = -\frac{2}{\La ^3} \
\e^{-\frac{p^2}{\La ^2}}\ 
\label{33}
\eqe
for the $\La$-derivative of the propagator, we arrive at the bound
\ben\label{irr1st}
\begin{split}
&\left|\int_{\La}^{\La_{0}}\d\la\, \left(N+2 \atop 2 \right) \, \int_{\ell} \dot{C}^{\la}(\ell)\partial_{\vec{p}}^{w}\L^{\la,\Lambda_{0}}_{N+2,L-1}( \ell, -\ell,  \vec{p};\LIR)\right|\\
&\leq\left(N+2 \atop 2 \right)\int_{\Lambda}^{\Lambda_{0}} \d{\lambda} \int_{\ell}\ \frac{2e^{-\frac{\ell^{2}}{\la^{2}}}}{\la^{3}}     \sqrt{|w|!}\ K^{(N+4L-6)(|w|+1)}\,   \sum_{T \in {\cal T}_{N+2,2L-2,(0,0,w)}(\ell, -\ell, \vec{p})}
\\
&\times  %\La_{\ma(\vec{a},T)+1}^{-S} % \prod_{b\in\vec{a}} \La_{b+1}^{-b+ d_{b,\vec{a}}} 
 \prod_{i \in \mathcal{I}(T)} \int_{\la}^{\La_{0}} \d\la_{i}\ f_{\lambda_{i}}(k_i(\ell);N+2,L-1,\theta_i) 
\mathcal{P}_{L-1}\left( \log_{+}\sup\left(\frac{|(\vec{p},\ell)|_\LIR}{\min_{j\in\mathcal{I}}\lambda_{j}},
 \frac{\max_{j\in\mathcal{I}}\lambda_{j}}{\LIR}\right)\right) \\
&\leq   \sqrt{|w|!}\ K^{(N+4L-6)(|w|+1)}\,   \sum_{T \in {\cal T}_{N+2,2L-2,(0,0,w)}(0,0,\vec{p})}
\\
&\times  \int_{\Lambda}^{\Lambda_{0}} \d{\lambda} \, \la
 \prod_{i \in \mathcal{I}(T)} \int_{\la}^{\La_{0}} \d\la_{i}\ f_{\lambda_{i}}(k_i;N,L,\theta_i) 
\mathcal{P}'_{L-1}\left( \log_{+}\sup\left(\frac{|\vec{p}|_\LIR}{\min_{j\in\mathcal{I}}\lambda_{j}},
 \frac{\max_{j\in\mathcal{I}}\lambda_{j}}{\LIR}\right)\right) \\
&\leq   \sqrt{|w|!}\ K^{(N+4L-6)(|w|+1)}\,   \sum_{T \in {\cal T}_{N+2,2L-2,(0,0,w)}(0,0,\vec{p})}
\\
&\times   \prod_{i \in \mathcal{I}(T)} \int_{\La}^{\La_{0}} \d\la_{i}\ f_{\lambda_{i}}(k_i;N,L,\theta_i) 
\mathcal{P}'_{L-1}\left( \log_{+}\sup\left(\frac{|\vec{p}|_\LIR}{\min_{j\in\mathcal{I}}\lambda_{j}},
 \frac{\max_{j\in\mathcal{I}}\lambda_{j}}{\LIR}\right)\right)\cdot \min_{j\in\mathcal{I}}\la_{j}^{2} \, .
\end{split}
\een
Here we have used the notation $k_{i}(\ell)$ in order to indicate that the momenta associated to the internal lines of a tree may  depend on the integration variable $\ell$. 
The second inequality is obtained with the help of lemma \ref{lablemma1} for $d=0$ (see the appendix), which allows us to bound the $\ell$-integral. To see that the conditions stated in this lemma are met, we note that the number of factors in the product over $i\in\mathcal{I}(T)$ is bounded  by $\frac{N-4}{2}+L$, as required (this number corresponds to $N'$ in the lemma). We have further absorbed some constants as well as the $N$ and $L$ dependent factors into the ``new'' (larger) polynomials $\mathcal{P}'_{L-1}$. Note that the degree of the polynomial does not change in the process, so it is at most equal to $L-1$. The last inequality follows from the bound
\ben\label{lambdaintsest}
\begin{split}
\int_{\Lambda}^{\La_{0}} \d{\lambda}\, \la \, \int_{\la}^{\La_{0}}\d\la_{1}\ldots \int_{\la}^{\La_{0}}\d\la_{|\mathcal{I}|} \, &=\, \int_{\La}^{\La_{0}}\d\la_{1}\ldots \int_{{\La}}^{\La_{0}}\d\la_{|\mathcal{I}|}\,  \int_{\Lambda}^{\min_{i}\la_{i}} \d{\lambda}\, \la \\
& \leq \int_{\La}^{\La_{0}}\d\la_{1}\ldots \int_{{\La}}^{\La_{0}}\d\la_{|\mathcal{I}|}\, \min_{i}\la_{i}^{2}\, .
\end{split}
\een
The factor $\min_{j\in\mathcal{I}}\la_{j}^{2}$ now allows us to ``reduce'' the weights $f_{\la_{i}}$ in \eqref{irr1st}: 
\begin{itemize}
\item Let us fix an internal line $a\in\mathcal{I}(T)$. If the weight factor of this internal line satisfies $\theta_{a}>1$, then the trivial bound $f_{\lambda_{a}}(k_a;N,L,\theta_a)\cdot \min_{j\in\mathcal{I}}\la_{j}\leq f_{\lambda_{a}}(k_a;N,L,\theta_a-1)$ holds.
\item If $\theta_{a}=1$, we can use one of the factors $\min_{j\in\mathcal{I}}\la_{j}$ to ``remove'' the corresponding integral over $\la_{a}$ using the bound (see lemma~\ref{theta1lem} in the appendix for the proof)
\ben\label{theta1}
\begin{split}
&\int_{\La}^{\La_{0}} \d\la_{a}\ f_{\lambda_{a}}(k_a;N,L,\theta_a=1) 
\mathcal{P}_{L-1}\left( \log_{+}\sup\left(\frac{|\vec{p}|_\LIR}{\min_{j\in\mathcal{I}}\la_{j}},
 \frac{\max_{j\in\mathcal{I}}\lambda_{j}}{\LIR}\right)\right)\cdot \min_{j\in\mathcal{I}}\la_{j}^{2}\\
 &\leq \mathcal{P}'_{L-1}\left( \log_{+}\sup\left(\frac{|\vec{p}|_\LIR}{\min_{j\in\mathcal{I}\setminus \{a\}}\la_{j}},
 \frac{\max_{j\in\mathcal{I}\setminus \{a\}}\la_{j}}{\LIR}\right)\right)\cdot \min_{j\in\mathcal{I}\setminus \{a\}}\lambda_{j} \ \ ,
\end{split}
\een
where $\mathcal{P}'_{L-1}$ is again a larger polynomial. 
In terms of our trees, the process of removing the integral over $\la_{a}$ corresponds to removing the internal line $a$.
\end{itemize}
In summary, following~\cite{GK, Kopper:2001to}, we can use the factor $\min_{j\in\mathcal{I}}\la_{j}^{2}$ in \eqref{irr1st} to decrease the weight of any two internal lines. We can organise this ``reduction procedure'' systematically on the level of trees, which will yield a relation between trees in ${\cal T}_{N+1,2L-1,(0,w)}(0,\vec{p})$ and trees in ${\cal T}_{N,2L,w}(\vec{p})$. Thus, by applying this reduction procedure twice, we can relate the trees appearing in \eqref{irr1st} to the ones which appear in the claimed bound \eqref{mainbound}.
\paragraph{Reduction procedure for trees:} Given a tree $T\in{\cal T}_{N+1,L-1,(0,w)}(\ell,p_{1},\ldots,p_{N})$, where $w\in\mathbb{N}^{4(N-1)}$, we define the reduced tree $\mathcal{R}_{\ell}(T)\in{\cal T}_{N,L,{w}}(p_{1},\ldots, p_{N})$ by the following successive operations (see fig.\ref{fig:reduction1} for a visualisation and see~\cite{GK} for a proof that the procedure indeed yields a tree in ${\cal T}_{N,L,w}$):
\begin{enumerate}\label{reduct}
\item Delete the external line with associated momentum $\ell$ %labelled by the index $i\in\{1,\ldots,N\}$ 
from the tree $T$. Imposing momentum conservation at the vertices, it follows that the momenta associated to the internal lines of the resulting tree $\mathcal{R}_{\ell}(T)$ do not depend on $\ell$.
\item Reduce by one the weight $\rho_{r}>0$ of an internal line $r\in\mathcal{I}(T)$ which is either\footnote{The cases a) and b) are exclusive, and there always exists a line $r$ satisfying one of the conditions (see~\cite{GK}).}
\begin{enumerate}
\item adjacent to the vertex where the external line with momentum $\ell$ was suppressed. 
\item adjacent to an internal line $r_{0}$ with weight $\rho_{r_{0}}=0$ which in turn is adjacent to the removed external line. 
\end{enumerate}
\item If an internal line has acquired the weight $\theta_{j}=0$ in this process, delete that line by merging the adjacent vertices.\footnote{Property \ref{it9} in definition \ref{deftrees1} guarantees that this procedure does not produce vertices with coordination number larger than four.}
\item If a vertex has acquired coordination number two, and if it is adjacent to two internal lines, then fuse these lines into one, adding up their weights.
\end{enumerate}
\begin{figure}[htbp]
\begin{center}
\includegraphics{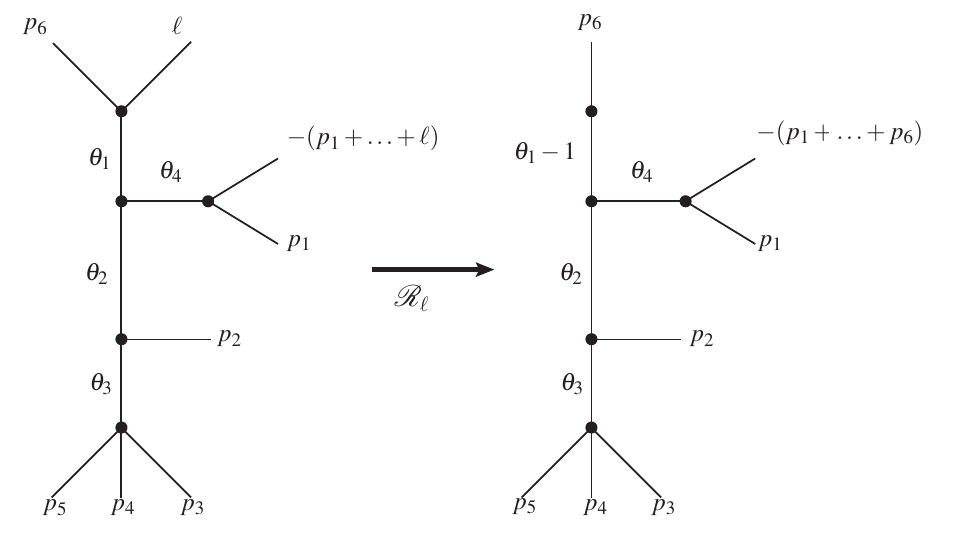}
\end{center}
\caption{The reduction procedure acting on a tree $T\in\mathcal{T}_{8,4,w}$.}
\label{fig:reduction1}
\end{figure}
The last point demands  further explanation. It follows from the bound (note that two lines connected by a vertex of coordination number two carry the same momentum)
\ben\label{internalmerge}
\begin{split}
&\int_{\La}^{\La_{0}}\d\la_{i} \, \la_{i}^{-\theta_{i}-1}\, \e^{-\frac{k^{2}}{\alpha(N,L)\la_{i}^{2}}}\cdot \int_{\La}^{\La_{0}}\d\la_{j} \, \la_{j}^{-\theta_{j}-1}\, \e^{-\frac{k^{2}}{\alpha(N,L)\la_{j}^{2}}} \,\mathcal{P}_{L-1}\left( \log_{+}\sup\left(\frac{|\vec{p}|_\LIR}{\min(\la_{i},\lambda_{j})},
 \frac{\max(\la_{i},\lambda_{j})}{\LIR}\right)\right)\\
&=\int_{\La}^{\La_{0}}\d\la_{i} \, \la_{i}^{-\theta_{i}-1}\, \e^{-\frac{k^{2}}{\alpha(N,L)\la_{i}^{2}}}\cdot \int_{\la_{i}}^{\La_{0}}\d\la_{j} \, \la_{j}^{-\theta_{j}-1}\, \e^{-\frac{k^{2}}{\alpha(N,L)\la_{j}^{2}}} \,\mathcal{P}_{L-1}\left( \log_{+}\sup\left(\frac{|\vec{p}|_\LIR}{\lambda_{i}},
 \frac{\la_{j}}{\LIR}\right)\right)+ ``(i\leftrightarrow j)''\\ 
 &\leq \int_{\La}^{\La_{0}}\d\la_{i} \, \la_{i}^{-\theta_{i}-\theta_{j}-1}\, \e^{-\frac{k^{2}}{\alpha(N,L)\la_{i}^{2}}}\,\mathcal{P}_{L-1}'\left( \log_{+}\sup\left(\frac{|\vec{p}|_\LIR}{\la_{i}},
 \frac{\la_{i}}{\LIR}\right)\right) \ \ ,
\end{split}
\een
where we used lemma~\ref{lemmalambdalarge} to bound the $\la_{j}$-integral.
Thus, applying the reduction operation twice to remove the external lines carrying the loop momenta $\ell,-\ell$ from the trees in ${\cal T}_{N+2,2L-2,(0,0,w)}(\ell,-\ell, \vec{p})$, we arrive at the bound
\ben\label{R2bound}
\begin{split}
&\text{r.h.s. of \eqref{irr1st}} \leq  \sqrt{|w|!}\ K^{(N+4L-6)(|w|+1)}\,   \sum_{\substack{\mathcal{R}_{\ell}\circ\mathcal{R}_{-\ell}(T) \\ T \in {\cal T}_{N+2,2L-2,(0,0,w)}(\ell,-\ell,\vec{p}) }}
\\
&\times  \prod_{i \in \mathcal{I}(\mathcal{R}_{\ell}\circ\mathcal{R}_{-\ell}(T))} \int_{\La}^{\La_{0}} \d\la_{i}\ f_{\lambda_{i}}(k_i;N,L,\theta_i)
\mathcal{P}_{L-1}\left( \log_{+}\sup\left(\frac{|\vec{p}|_\LIR}{\min_{j\in\mathcal{I}}\lambda_{j}},
 \frac{\max_{j\in\mathcal{I}}\lambda_{j}}{\LIR}\right)\right)
 \, .
 \end{split}
\een
%
%\vspace{.5cm}
As mentioned above, we know from~\cite{GK} that the reduction operation $\mathcal{R}$ maps trees from  $\mathcal{T}_{N+1,2L-1,(0,w)}$ into $\mathcal{T}_{N,2L,w}$. However, different trees $T\in\mathcal{T}_{N+1,2L-1,(0,w)}$ may yield the same reduced tree $T'=\mathcal{R}(T)\in\mathcal{T}_{N,2L,w}$, i.e. the reduction map is not injective. Since we want to write \eqref{R2bound} as a sum over trees in $\mathcal{T}_{N,2L,w}$, we have to know how many different trees $T$ can yield the same reduced tree $T'$ and multiply our bound by this factor. We claim that this number is bounded by the following expression:
\ben\label{reductionbound}
%V_{3}(T')+V_{2}(T')+|\mathcal{I}(T')|\cdot 2(|w|+1) \leq
 2L+N+ ({N-4}+2L)\cdot (|w|+1) \ .
\een
%different trees $T\in\mathcal{T}_{N+1,2L-1,w}$. 
To verify this bound, we fix a tree $T'\in\mathcal{T}_{N,2L,w}$ and ``invert'' the reduction procedure, i.e. we attach an additional  external line to the tree $T'$ in any possible way such that it yields a tree  $T\in\mathcal{T}_{N+1,2L-1,(0,w)}$. How many different trees $T$ can be obtained this way?
\begin{description}
\item[Vertices:]  We can attach the additional external line to any vertex of coordination number two or three of the tree $T'$ and end up with a different tree in $T$. As there are  $V_{2}(T')+V_{3}(T')\leq 2L$ [cf. item \ref{it3} in def.~\ref{deftrees1}] vertices of this kind, we obtain up to $2L$ different trees this way. 
\item[Lines:] We can attach  the additional external line to any internal or external line of $T'$, creating a vertex of coordination number three and splitting the initial line into two. If the initial line was internal, its weights $\sigma,\rho$ are distributed over the two resulting lines, which can be done in $2(\sigma+1)\leq 2(|w|+1)$ ways. As there are $|\mathcal{I}(T')|\leq \frac{N-4}{2}+L$ internal lines in $T'$, we obtain up to $(N-4+2L)(|w|+1)$ different trees $T$ this way. If the initial line was external, the weights of the resulting internal line are uniquely determined to be $\rho=0=\sigma$. Thus, we obtain up to $N$ different trees this way. In total, we can bound the number of different trees $T$ obtained by attaching a line to $T'$ by $(N-4+2L)(|w|+1)+N$.
\end{description}
Since we apply the reduction operation twice on trees $T\in\mathcal{T}_{N+2,2L-2,(0,0,w)}$, we pick up a factor 
\ben\label{totalredfactor}
 \left[2L+N+ ({N-4}+2L)\cdot (|w|+1)\right]^{2} %\left[2L+N+ ({N+1-4}+2(L-1))\cdot (|w|+1)\right]
\een
in total. Hence, we infer that the following bound holds:
\ben\label{R2boundremove}
\begin{split}
&\text{r.h.s. of \eqref{R2bound}} 
\leq  (|w|+1)^{2}\, \sqrt{|w|!}\ K^{(N+4L-6)(|w|+1)}\,   \sum_{T \in {\cal T}_{N,2L,w}(\vec{p})}\\
&\times  \prod_{i \in \mathcal{I}(T)} \int_{\La}^{\La_{0}} \d\la_{i}\ f_{\lambda_{i}}(k_i;N,L,\theta_i)
\mathcal{P}'_{L-1}\left( \log_{+}\sup\left(\frac{|\vec{p}|_\LIR}{\min_{j\in\mathcal{I}}\lambda_{j}},
 \frac{\max_{j\in\mathcal{I}}\lambda_{j}}{\LIR}\right)\right)
 \, .
 \end{split}
\een
Here we have again absorbed $N$ and $L$ dependent factors into the polynomials $\mathcal{P}'_{L-1}$. We find that this contribution satisfies the inductive bound \eqref{mainbound}, provided that $K$ is chosen sufficiently large such that
\ben
K^{-2(|w|+1)}\,  (|w|+1)^{2}\leq 1\, .
\een
\paragraph{The case $N= 2, |w|>2$:}
Following the same steps as in \eqref{irr1st}, we find that this contribution is bounded by
\ben\label{irr1stN=2}
\begin{split}
&\left|\int_{\La}^{\La_{0}}\d\la\, \left(4 \atop 2 \right) \, \int_{\ell} \dot{C}^{\la}(\ell)\partial_{{p}}^{w}\L^{\la,\Lambda_{0}}_{4,L-1}( \ell, -\ell ,{p},-p;\LIR)\right|\leq    \sqrt{|w|!}\ K^{(4L-4)(|w|+1)}\, \sum_{T \in {\cal T}_{4,2L-2,(0,0,w)}(0,0,p,-p)}  \\
&\times   \prod_{i \in \mathcal{I}(T)} \int_{\La}^{\La_{0}} \d\la_{i}\ f_{\lambda_{i}}(k_i;2,L,\theta_i) 
\mathcal{P}_{L-1}\left( \log_{+}\sup\left(\frac{|{p}|_\LIR}{\min_{j\in\mathcal{I}}\lambda_{j}},
 \frac{\max_{j\in\mathcal{I}}\lambda_{j}}{\LIR}\right)\right)\cdot \min_{j\in\mathcal{I}}\la_{j}^{2} \, .
\end{split}
\een
There are three types of trees in the set ${ {\cal T}_{4,2L-2,(0,0,w)}(0,0,p,-p)}$, see fig.\ref{fig:4trees} (see~\cite{GK} for more details). 
\begin{figure}[htbp]
\begin{center}
\includegraphics[width=\textwidth]{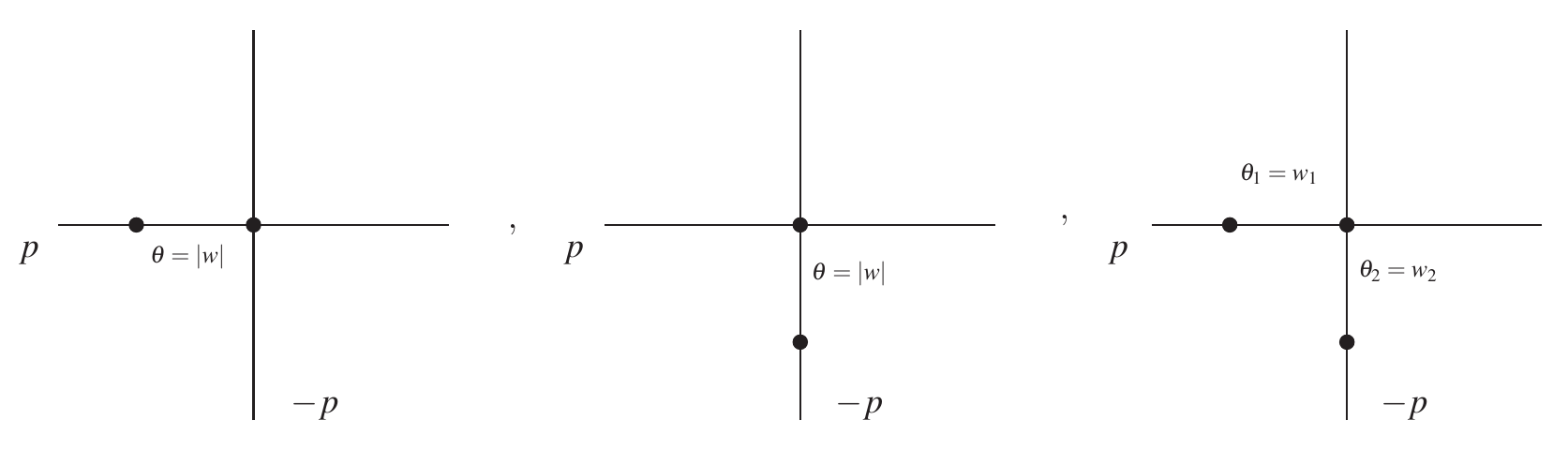}
\end{center}
\caption{The trees in ${\cal T}_{4,2L-2,(0,0,w)}(0,0,p,-p)$. Here $w_{1}+w_{2}=w$ and $|w_{i}|>0$.}
\label{fig:4trees}
\end{figure}

The contribution of each of those trees is bounded by
\ben\label{irr1stN=2b}
\begin{split}
&   \sqrt{|w|!}\ K^{(4L-4)(|w|+1)} \int_{\La}^{\La_{0}} \d\la\ \la^{2-|w|-1} \e^{-\frac{p^{2}}{\alpha(2,L) \lambda^{2}}}
\mathcal{P}_{L-1}\left( \log_{+}\sup\left(\frac{|{p}|_\LIR}{\lambda},
 \frac{\lambda}{\LIR}\right)\right) 
\end{split}
\een
in accordance with  the claimed bound \eqref{boundN2}.
 For the trees with one internal line this follows trivially from \eqref{irr1stN=2}, and for the trees with two internal lines we have also used the bound \eqref{internalmerge}. It remains to bound the cardinality of ${\cal T}_{4,2L-2,(0,0,w)}(0,0,p,-p)$. The two trees with one internal line are unique, and there are $|w|+1$ combinations of trees with two internal lines. Thus, we pick up a factor of $(|w|+1)$ from the sum over $T\in {\cal T}_{4,2L-2,(0,0,w)}(0,0,p,-p)$, and we verify that the contribution at hand satisfies the claimed bound, \eqref{boundN2}, provided that $K$ is chosen large enough that
\ben
(|w|+1)\leq K^{2(|w|+1)}\, .
\een
\subsubsection{Second term on the r.h.s. of the flow equation \eqref{CAGFEexpand}}\label{subsec:secondterm}
We distinguish the following cases in the sum over $n_{1}+n_{2}=N+2$ in \eqref{CAGFEexpand}: 
\begin{enumerate}
\item The case $n_{1},n_{2}\geq 4$
\item The case  $n_{1}\geq 4$, $n_{2}=2$
\item The case $n_{1}=2=n_{2}$
\end{enumerate}
\paragraph{For $n_{1},n_{2}\geq 4$:} 
For the sake of brevity, it will be useful to combine the inductive bound \eqref{mainbound} and the known bound \eqref{boundrel1} into one slightly weaker bound by writing for $N\geq 4$ and any $w\in\mathbb{N}^{4(N-1)}$
 \ben\label{mainboundcomb}
\begin{split}
|\pa^w {\cal L}^{\La,\Lao}_{N,L}(\vec{p};\LIR)| 
&\le  \sqrt{|w|!}\ K^{(N+4L-4)(|w|+1)}\,   \sum_{T \in {\cal T}_{N,2L,w}(\vec{p})}
\\
&\times  %\La_{\ma(\vec{a},T)+1}^{-S} % \prod_{b\in\vec{a}} \La_{b+1}^{-b+ d_{b,\vec{a}}} 
 \prod_{i \in \mathcal{I}(T)} \int_{\La}^{\La_{0}} \d\la_{i}\ f_{\lambda_{i}}(k_i;N,L,\theta_i) 
\mathcal{P}_{L}\left( \log_{+}\sup\left(\frac{|\vec{p}|_\LIR}{\La},
 \frac{\max_{j\in\mathcal{I}}(\lambda_{j},\La)}{\LIR}\right)\right) \, ,
\end{split}
\een
where we use the convention
\ben\label{emptyprod}
\prod_{i \in \mathcal{I}(T)} \int_{\La}^{\La_{0}} \d\la_{i}\ f_{\lambda_{i}}(k_i;N,L,\theta_i)  =1 \qquad \text{if }\mathcal{I}(T)=\emptyset\, .
\een
Integrating the flow equation \eqref{CAGFEexpand}, inserting the bound \eqref{mainboundcomb} and making use of the inequality~\cite{Hollands:2011gf}
\eq
\Bigl| \pa^{w}_{q} \e^{-\frac{q^2}{\La ^2}} | \ \le\
k\ \La ^{-|w|}\ \sqrt{|w|!}\ 2^{\frac{|w|}{2}}\ \e^{-\frac{q^2}{2\La ^2}}\
\ , \qquad k=1.086\ldots\, ,
\label{w}
\eqe 
we find that an entry in the sum over $n_{i}$ in the second term on the r.h.s. of the flow equation satisfies the bound
\ben\label{irr2nd}
\begin{split}
&\left|\int_{\La}^{\La_{0}}\d\la  \sum_{\substack{ {v_{1}+v_{2}+v_{3}=w}\\ l_{1}+l_{2}=L }} c_{\{v_{i}\}} \partial_{\vec{p}}^{v_{1}} \L^{\la,\Lambda_{0}}_{n_{1},l_{1}}(   p_{1},\ldots,p_{n_{1}-1},q;\LIR)\,  \partial_{\vec{p}}^{v_{3}}\dot{C}^{\la}(q)\,   \partial_{\vec{p}}^{v_{2}}\L^{\la,\Lambda_{0}}_{n_{2},l_{2}}(-q,  p_{n_{1}},\ldots,p_{N};\LIR)  \right| \\
&\leq \int_{\Lambda}^{\Lambda_{0}} \d{\lambda} \hspace{-.4cm} \sum_{\substack{ {v_{1}+v_{2}+v_{3}=w}\\ (v_{1})_{j}=0\, \forall j\geq n_{1} \\l_{1}+l_{2}=L }}c_{\{v_{i}\}}2k \frac{\sqrt{|v_{3}|!\, 2^{|v_{3}|} }\  e^{-\frac{q^{2}}{2\la^{2}}}}{\la^{3+|v_{3}|}}\times \bigg[  \sqrt{|v_{1}|!}\ K^{(n_{1}+4l_{1}-4)(|v_{1}|+1)}\,  \sum_{T_{1} \in {\cal T}_{n_{1},2l_{1},\tilde{v}_{1}}(\vec{p}_{1})}
\\
&\times  %\La_{\ma(\vec{a},T)+1}^{-S} % \prod_{b\in\vec{a}} \La_{b+1}^{-b+ d_{b,\vec{a}}} 
 \prod_{i \in \mathcal{I}(T_{1})} \int_{\la}^{\La_{0}} \d\la_{i}\ f_{\lambda_{i}}(k_i;n_{1},l_{1},\theta_i) 
\mathcal{P}_{l_{1}}\left( \log_{+}\sup\left(\frac{|\vec{p}_{1}|_\LIR}{\la},
 \frac{\max\limits_{j\in\mathcal{I}(T_{1})}(\lambda_{j},\la)}{\LIR}\right)\right)
\bigg]\\
&\times\bigg[   \sqrt{|v_{2}|!}\ K^{(n_{2}+4l_{2}-4)(|v_{2}|+1)}\sum_{T_{2} \in {\cal T}_{n_{2},2l_{2},\tilde{v}_{2}}(\vec{p}_{2})}
\\
&\times  %\La_{\ma(\vec{a},T)+1}^{-S} % \prod_{b\in\vec{a}} \La_{b+1}^{-b+ d_{b,\vec{a}}} 
 \prod_{i \in \mathcal{I}(T_{2})} \int_{\la}^{\La_{0}} \d\la_{i}\ f_{\lambda_{i}}(k_i;n_{2},l_{2},\theta_i) 
\mathcal{P}_{l_{2}}\left( \log_{+}\sup\left(\frac{|\vec{p}_{2}|_\LIR}{\la},
 \frac{\max\limits_{j\in\mathcal{I}(T_{2})}(\lambda_{j},\la)}{\LIR}\right)\right)
\bigg]\, .
\end{split}
\een
Here we use the notation $\vec{p}_{1}=(p_{1},\ldots,p_{n_{1}-1},q), \vec{p}_{2}=(-q,p_{n_{1}},\ldots,p_{N})$ and $q$ as in the flow equation, see \eqref{CAGFEexpand}. We made use of the fact that the sum over $v_{1}$ is non-vanishing only if $(v_{1})_{j}=0$ for $j\geq n_{1}$, since the CAS $\L^{\la,\Lambda_{0}}_{n_{1},l_{1}}(   p_{1},\ldots,p_{n_{1}-1},q;\LIR)$ does not depend on $p_{n_{1}},\ldots, p_{N-1}$. Here we used the notation $(v_{1})_{j}\in\mathbb{N}^{4}$ for the four-indices in the tuple $v_{1}=((v_{1})_{1},\ldots, (v_{1})_{N-1})\in\mathbb{N}^{4(N-1)}$. The multi-indices $\tilde{v}_{i}\in\mathbb{N}^{4(n_{i}-1)}$ appearing in the trees are related to the summation indices $v_{i}\in\mathbb{N}^{4(N-1)}$ via
\ben
\tilde{v}_{1}:=((v_{1})_{1},\ldots,(v_{1})_{n_{1}-1})
\een
\ben
\tilde{v}_{2}:=(\sum_{i=1}^{n_{1}-1} (v_{2})_{i},(v_{2})_{n_{1}},\ldots,(v_{2})_{N-1})\, .
\een
We can ``merge'' the logarithmic polynomials with the help of the inequality
\ben\label{logsmerge}
\begin{split}
&\mathcal{P}_{l_{1}}\left( \log_{+}\sup\left(\frac{|\vec{p}_{1}|_\LIR}{\la},
 \frac{\max\limits_{j\in\mathcal{I}(T_{1})}(\lambda_{j},\la)}{\LIR}\right)\right)\cdot \mathcal{P}_{l_{2}}\left( \log_{+}\sup\left(\frac{|\vec{p}_{2}|_\LIR}{\la},
 \frac{\max\limits_{j\in\mathcal{I}(T_{2})}(\lambda_{j},\la)}{\LIR}\right)\right)\\
 &\qquad\leq \mathcal{P}_{L}\left( \log_{+}\sup\left(\frac{|\vec{p}|_\LIR}{\la},
 \frac{\max\limits_{j\in\mathcal{I}(T_{1})\cup\mathcal{I}(T_{2})} (\lambda_{j},\la) }{\LIR}\right)\right)\, .
\end{split}
\een
Noting also that 
\ben\label{expfbound}
\frac{ e^{-\frac{q^{2}}{2\la^{2}}}}{\la^{3+|v_{3}|}}\leq f_{\la}(q;N,L,\theta=2+|v_{3}|)\, ,
\een
and that
\ben\label{factorialtriv}
|v_{1}|!|v_{2}|! |v_{3}|!\leq |w|!\quad \text{for } v_{1}+v_{2}+v_{3}=w\, ,
\een 
 we find that \eqref{irr2nd} is smaller than
\ben\label{2ndirrnestedint}
\begin{split}
&  \sqrt{|w|!2^{|w|}}\ K^{(N+4L-6)(|w|+1)}  \sum_{\substack{ {v_{1}+v_{2}+v_{3}=w}\\ (v_{1})_{j}=0\, \forall j\geq n_{1}  \\ l_{1}+l_{2}=L }}c_{\{v_{i}\}} \int_{\Lambda}^{\Lambda_{0}} \d{\lambda} f_{\la}(q,N,L,\theta=2+|v_{3}|)\sum_{{T_{1} \in {\cal T}_{n_{1},2l_{1},\tilde{v}_{1}}(\vec{p}_{1}) } \atop {T_{2} \in {\cal T}_{n_{2},2l_{2},\tilde{v}_{2}}(\vec{p}_{2}) } }\\
&  \times\hspace{-.3cm}
 \prod_{i \in \mathcal{I}(T_{1})\cup\mathcal{I}(T_{2}) } \int\limits_{\La}^{\La_{0}} \d\la_{i}\ f_{\lambda_{i}}(k_i;N,L,\theta_i) 
\mathcal{P}_{L}\left( \log_{+}\sup\left(\frac{|\vec{p}|_\LIR}{\la},
 \frac{\max\limits_{j\in\mathcal{I}(T_{1})\cup\mathcal{I}(T_{2})}(\la_{j},\la)}{\LIR}\right)\right)\, .
\end{split}
\een
Here we have also made use of the bound $f_{\lambda_{i}}(k_i;n_{a},l_{a},\theta_i)\leq f_{\lambda_{i}}(k_i;N,L,\theta_i)$, where $a\in\{1,2\}$. This inequality follows directly from our definition of the weight functions, eq.\eqref{alphadef}, combined with the inequality\footnote{\label{footni}Recall that $\mathcal{L}^{\La, \La_{0}}_{2,0}(p)=0$ by definition, so non-vanishing contributions even satisfy, in view of $n_{1}+n_{2}=N+2$ and $l_{1}+l_{2}=L$,
\ben
n_{2}/2+2l_{2}\geq 2\quad \Rightarrow\quad n_{1}/2+2l_{1}=N/2+2L+1-(n_{2}/2+2l_{2})\leq N/2+2L-1\, .
\een
Exchanging the roles of the indices $1$ and $2$, one also verifies $n_{2}/2+2l_{2}\leq N/2+2L-1$.
 } $n_{a}/2+2l_{a}\leq N/2+2L$.
 
It remains again to express this bound in terms of trees $T\in\mathcal{T}_{N,2L,w}(\vec{p})$ in order to verify consistency with our hypothesis \eqref{mainbound}. This is achieved as follows (see fig.\ref{fig:merge1}):
\begin{figure}[htbp]
\begin{center}
\includegraphics[width=10cm]{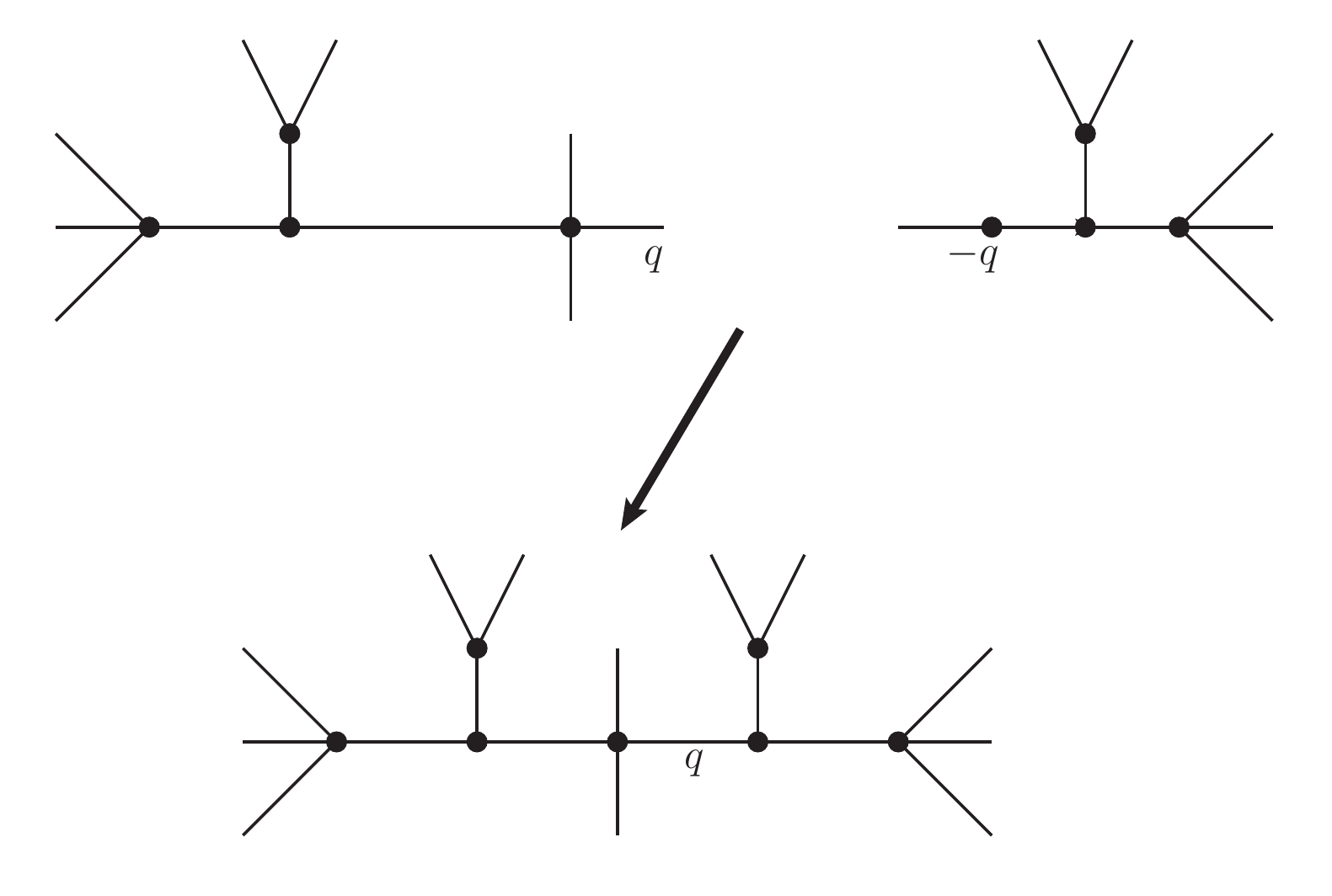}
\end{center}
\caption{Merging two trees $T_{1}\in\mathcal{T}_{8,2,\tilde{v}_{1}}, T_{2}\in\mathcal{T}_{6,4,\tilde{v}_{2}}$ into a tree $T\in\mathcal{T}_{12,6,w}$. The weight $\theta_{i}$ of the new internal line connecting the two trees is set to be $|v_{3}|+\sigma_{j}+2$, where the index $j$ corresponds to the internal line adjacent to the external line with momentum $-q$ in the tree $T_{2}$ on the top right.}
\label{fig:merge1}
\end{figure}

To every pair of trees $T_{1} \in {\cal T}_{n_{1},2l_{1},\tilde{v}_{1}}(\vec{p_{1}})$ and $T_{2} \in {\cal T}_{n_{2},2l_{2},\tilde{v}_{2}}(\vec{p_{2}})$ we associate the tree $T\in\mathcal{T}_{N,2L,w}(\vec{p})$, which results from joining $T_{1}$ and $T_{2}$ along the external lines carrying momentum $q$ and $-q$, respectively (note that the trees have such a line by definition). The resulting new internal line therefore carries momentum $-q$, and we associate the weights $\rho=2, \sigma=|v_{3}|$ (thus $\theta=2+|v_{3}|$) to it. If one of the external lines with momenta $q,-q$ in $T_{1}$  and $T_{2}$ is adjacent to a vertex of coordination number two (as in the right tree in fig.\ref{fig:merge1}), then this procedure yields a vertex of coordination number two adjacent to two internal lines in the tree $T$. In that case, we remove this vertex and fuse the adjacent lines, adding up their weights (just as in item 4. of the reduction procedure, see page \pageref{reduct}). One can show that 
this procedure indeed yields a tree $T\in\mathcal{T}_{N,2L,w}(\vec{p})$ (see~\cite{GK}).

For fixed $v_{1},v_{2},v_{3}$, this merging procedure is injective (different combinations of trees $T_{1},T_{2}$ yield different trees $T$). Using \eqref{cwest2} to bound the sum over $v_{1},v_{2},v_{3}$ and using $\sum_{l_{1}+l_{2}=L}=L+1$, we thus find that \eqref{2ndirrnestedint} is bounded by
\ben\label{2ndirrboundfinal}
\begin{split}
& \sqrt{|w|! 2^{|w|}} \, 3^{|w|}\, K^{(N+4L-6)(|w|+1)}   \sum_{T \in {\cal T}_{N,2L,w}(\vec{p})}
\\
&\times% \prod_{b\in\vec{a}} \La_{b+1}^{-b+ d_{b,\vec{a}}} 
 \  \prod_{i \in \mathcal{I}(T) } \int_{\La}^{\La_{0}} \d\la_{i}\ f_{\lambda_{i}}(k_i;N,L,\theta_i) 
\mathcal{P}_{L}\left( \log_{+}\sup\left(\frac{|\vec{p}|_\LIR}{\min_{j\in\mathcal{I}}\lambda_{j}},
 \frac{\max_{j\in\mathcal{I}}\lambda_{j}}{\LIR}\right)\right)\, ,
\end{split}
\een
where we have absorbed the factor $L+1$ into the logarithmic polynomial. 
We conclude that the claimed bound \eqref{mainbound} is satisfied by this contribution  provided that $K$ is chosen large enough that
\ben\label{Kcond1}
 3^{|w|}\cdot{2^{|w|/2}} \leq K^{2(|w|+1)} \, .
\een
\paragraph{For $n_{1}\geq 4,n_{2}= 2$:} 
For the two point function $|\pa^{v_{2}} {\cal L}^{\La,\Lao}_{2,l_{2}}(q,-q;\LIR)|$ we can use
\ben\label{N2estweak}
%\begin{split}
|\pa^{v_{2}} {\cal L}^{\La,\Lao}_{2,l_{2}}(q,-q;\LIR)| 
%&\le  \prod_{i=1}^{S_{2}} \int_{\La_{i}}^{\La_{0}}\frac{\d\La_{i+1}} { \La_{i+1}}\, \sqrt{|w_{2}|!}\ K^{(4l_{2}-2)(|w_{2}|+1)}\ \La_{S_{2}+1}^{2-|w_{2}|}e^{-\frac{q^{2}}{\alpha(2,l_{2})\La_{S_{2}+1}^{2}}} \\
%&\qquad\times\mathcal{P}_{2,l_{2}}\left( \log\sup\left(\frac{|q|_\LIR}{\kappa(\La_{1},q, \LIR)},
% \frac{\La_{S_{2}+1}}{\LIR}\right)\right)\\
\leq   \sqrt{|v_{2}|!}\ K^{(4l_{2}-2)(|v_{2}|+1)}\, |q|_{\La}^{2} {\La}^{-|v_{2}|}\, \mathcal{P}_{l_{2}-1}\left( \log_{+}\sup\left(\frac{|q|_\LIR}{\La},
 \frac{{\La}}{\LIR}\right)\right)\, .
%\end{split}
\een
If $|v_{2}|\leq 2$, this bound follows trivially from \eqref{boundrel2}, while for $|v_{2}|>2$ it follows from \eqref{boundN2} after  application of formula \eqref{Lambdaint} from the appendix and after absorbing $l_{2}$ dependent factors into the polynomials $\mathcal{P}_{l_{2}-1}$. Let us also mention that the inequality \eqref{boundrel2}, and thus also \eqref{N2estweak}, relies crucially on the boundary condition \eqref{CAGBC1}. 
Using the inequality \eqref{N2estweak} along with our inductive bound \eqref{mainboundcomb}, we find that the contribution at hand is bounded by
\ben\label{irr2ndB}
\begin{split}
&\left|\int_{\La}^{\La_{0}}\d\la   \sum_{\substack{ {v_{1}+v_{2}+v_{3}=w}\\ l_{1}+l_{2}=L }}c_{\{v_{i}\}} \partial_{\vec{p}}^{v_{1}}\L^{\la,\Lambda_{0}}_{N,l_{1}}(   p_{1},\ldots,p_{N-1},q;\LIR)\,  \partial_{\vec{p}}^{v_{3}}\dot{C}^{\la}(q)\,   \partial_{\vec{p}}^{v_{2}}\L^{\la,\Lambda_{0}}_{2,l_{2}}(- q,q;\LIR)  \right| \\
&\leq\int_{\Lambda}^{\Lambda_{0}} \d\la   \sum_{\substack{ {v_{1}+v_{2}+v_{3}=w}\\ l_{1}+l_{2}=L }}c_{\{v_{i}\}} 2k \frac{\sqrt{|v_{3}|!\, 2^{|v_{3}|} }\  e^{-\frac{q^{2}}{2\la^{2}}}}{\la^{3+|v_{3}|}}  \times \bigg[  \sqrt{|v_{1}|!}\ K^{(N+4l_{1}-4)(|v_{1}|+1)}\,  \sum_{T\in {\cal T}_{N,2l_{1},v_{1}}(\vec{p})}
\\
&\times  %\La_{\ma(\vec{a},T)+1}^{-S} % \prod_{b\in\vec{a}} \La_{b+1}^{-b+ d_{b,\vec{a}}} 
 \prod_{i \in \mathcal{I}(T)} \int_{\la}^{\La_{0}} \d\la_{i}\ f_{\lambda_{i}}(k_i;N,l_{1},\theta_i) 
\mathcal{P}_{l_{1}}\left( \log_{+}\sup\left(\frac{|\vec{p}|_\LIR}{\la},
 \frac{\max_{j\in\mathcal{I}}(\lambda_{j},\la)}{\LIR}\right)\right)
\bigg]\\
&\times \sqrt{|v_{2}|!}\ K^{(4l_{2}-2)(|v_{2}|+1)}\, |q|_{\la}^{2}\ \la^{-|v_{2}|}\, \mathcal{P}_{l_{2}-1}\left( \log_{+}\sup\left(\frac{|q|_\LIR}{\la},
 \frac{\la}{\LIR}\right)\right)\\
&\leq  K^{(N+4L-6)(|w|+1)}\sqrt{|w|!\, 2^{|w|}} \sum_{\substack{ {v_{1}+v_{2}+v_{3}=w}\\ l_{1}+l_{2}=L }}c_{\{v_{i}\}} \int_{\Lambda}^{\Lambda_{0}} \d\la\,  f_{\la}(q,N,L,\theta=|v_{2}|+|v_{3}|)   \sum_{T \in {\cal T}_{N,2l_{1},v_{1}}(\vec{p})}  \\
&\times \prod_{i \in \mathcal{I}(T)} \int_{\La}^{\La_{0}} \d\la_{i}\ f_{\lambda_{i}}(k_i;N,L,\theta_i) 
\mathcal{P}_{L-1}\left( \log_{+}\sup\left(\frac{|\vec{p}|_\LIR}{\la},
 \frac{\max_{j\in\mathcal{I}}(\lambda_{j},\la)}{\LIR}\right)\right)\, .
\end{split}
\een
Here we made use of the elementary bound 
\ben\label{qest}
|q|_{\la}^{2}e^{-\frac{q^{2}}{4\la^{2}}}\leq  2\la^{2}\, .
\een
To see that \eqref{irr2ndB} is indeed bounded by \eqref{mainbound}, we associate to every tree $T \in {\cal T}_{N,2l_{1},v_{1}}(\vec{p}) $ another tree $T' \in {\cal T}_{N,2L,w}(\vec{p}) $ as follows:
\begin{itemize}
\item Add a vertex of coordination number two to the external line of $T$ carrying momentum $q$, as shown in fig.~\ref{fig:merge2}. To the resulting new internal line $i$ associate the weight  $\rho_{i}=0$, $\sigma_{i}=|v_{2}|+|v_{3}|$.
\item If this procedure yields two internal lines adjacent to a vertex of coordination number two, fuse these lines into one by removing that vertex and add up their weights (as in item 4. of the reduction procedure).
\end{itemize}
\begin{figure}[htbp]
\begin{center}
\includegraphics{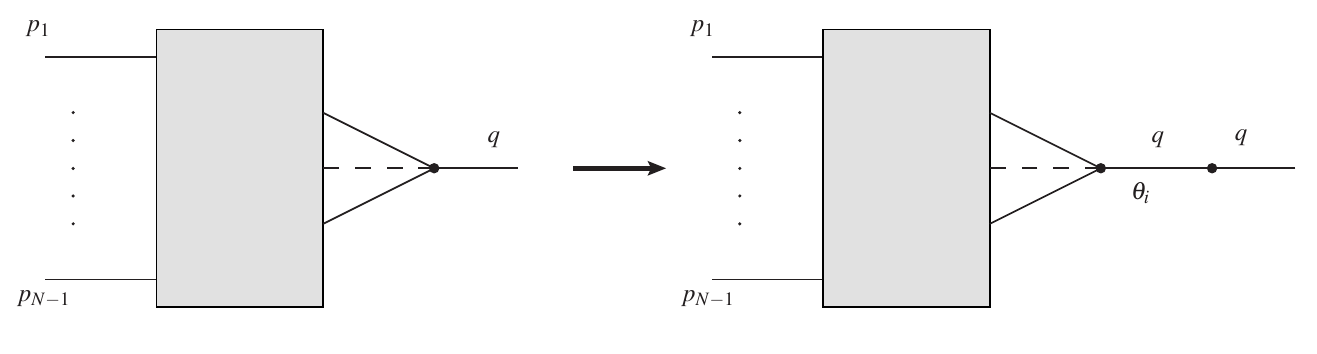}
\end{center}
\caption{If the external line carrying momentum $q$ is incident to a vertex of coordination number greater than $2$, we separate this line into two by adding a vertex of coordination number $2$. The resulting new internal line, which carries momentum $q$, is assigned the weight $\theta_{i}=|v_{2}|+|v_{3}|$.}
\label{fig:merge2}
\end{figure}
Again, for fixed values of $v_{1},v_{2},v_{3}$, this procedure is injective. Using \eqref{cwest2} to bound the sum over the $v_{i}$, and bounding the sum over $l_{1},l_{2}$ again by a factor $L+1$, which we absorb into the polynomial $\mathcal{P}_{L-1}$, it follows that the expression \eqref{irr2ndB} is smaller than 
\ben\label{case4merge}
\begin{split}
&   K^{(N+4L-6)(|w|+1)}3^{|w|}\sqrt{|w|!\, 2^{|w|}}% \left(\sum_{w_{1}+w_{2}+w_{3}=w}c_{\{w_{i}\}}\right)  
\sum_{T \in {\cal T}_{N,2L,w}(\vec{p})}  \\
&\times \prod_{i \in \mathcal{I}(T)} \int_{\La}^{\La_{0}} \d\la_{i}\ f_{\lambda_{i}}(k_i;N,L,\theta_i) 
\mathcal{P}_{L-1}\left( \log_{+}\sup\left(\frac{|\vec{p}|_\LIR}{\min_{j\in\mathcal{I}}\lambda_{j}},
 \frac{\max_{j\in\mathcal{I}}\lambda_{j}}{\LIR}\right)\right)
\end{split}
\een
which is consistent with our claim \eqref{mainbound}, provided that \eqref{Kcond1} holds true.
\paragraph{For $n_{1}=2=n_{2}$:} 
We make use of \eqref{N2estweak} and bound the second term in the flow equation by
\ben
\begin{split}
&\left|\int_{\La}^{\La_{0}}\d\la  \sum_{\substack{ {v_{1}+v_{2}+v_{3}=w}\\ l_{1}+l_{2}=L }}c_{\{v_{i}\}} \partial_{{p}}^{v_{1}}\L^{\la,\Lambda_{0}}_{2,l_{1}}(   p,-p;\LIR)\,  \partial_{{p}}^{v_{3}}\dot{C}^{\la}(-p)\,   \partial_{{p}}^{v_{2}}\L^{\la,\Lambda_{0}}_{2,l_{2}}(- p,p;\LIR)  \right| \\
&\leq\int_{\La}^{\La_{0}}\d\la\,\sum_{\substack{ {v_{1}+v_{2}+v_{3}=w}\\ l_{1}+l_{2}=L }}c_{\{v_{i}\}} 2k \frac{\sqrt{|v_{3}|!\, 2^{|v_{3}|} }\  e^{-\frac{p^{2}}{2\la^{2}}}}{\la^{3+|v_{3}|}}\\
&\times   \sqrt{|v_{1}|!}\ K^{(4l_{1}-2)(|v_{1}|+1)}\,|p|_{\la}^{2}\, \la^{-|v_{1}|}\, \mathcal{P}_{l_{1}-1}\left( \log_{+}\sup\left(\frac{|p|_\LIR}{\la},
 \frac{\la}{\LIR}\right)\right)\\
 &\times  \sqrt{|v_{2}|!}\ K^{(4l_{2}-2)(|v_{2}|+1)}|p|_{\la}^{2}\, \la^{-|v_{2}|}\, \mathcal{P}_{l_{2}-1}\left( \log_{+}\sup\left(\frac{|p|_\LIR}{\la},
 \frac{\la}{\LIR}\right)\right)\\
 &\leq \int_{\La}^{\La_{0}}\d\la\,\sum_{{v_{1}+v_{2}+v_{3}=w}}\hspace{-.4cm}  \frac{c_{\{v_{i}\}}\sqrt{|w|!\, 2^{|w|} }\  e^{-\frac{p^{2}}{4\la^{2}}}}{\la^{|w|-1}} K^{(4L-4)(|w|+1)} \mathcal{P}_{L-2}\left( \log_{+}\sup\left(\frac{|p|_\LIR}{\la},
 \frac{\la}{\LIR}\right)\right)\, .
\end{split}
\een
To obtain the last inequality we used the bound
\ben
|p|_{\la}^{4}\, e^{-\frac{p^{2}}{4\la^{2}}}\leq 9\, \la^{4} 
\een
and we absorbed some constants into the polynomials $\mathcal{P}_{L-2}$. The bound is compatible with our claim \eqref{boundN2}, provided that
\ben
K^{-2|w|}\,3^{|w|}\, 2^{|w|/2}\leq 1\, .
\een

%\vspace{.5cm}

We have verified that any item $(n_{1},n_{2})$ in the second term on the r.h.s. of the flow equation \eqref{CAGFEexpand} satisfies the inductive bounds \eqref{boundN2} and \eqref{mainbound}.  It remains to bound the sum over these parameters, for which we use 
\ben
\sum_{n_{1}+n_{2}=N+2}n_{1}n_{2}\leq (N+1)^{3}\, .
\een 
Absorbing this factor into the logarithmic polynomials, we find that the complete second term on the r.h.s. of the flow equation satisfies the claimed bounds. This finishes the proof of theorem~\ref{thm1}.
\end{proof}

%\vspace{.5cm}

For $\Lambda=0$ we may now state the following bounds on the massless (connected amputated) Schwinger functions:
\begin{cor}\label{cor1}
For any $N,L\in\mathbb{N}$ and $w\in\mathbb{N}^{4(N-1)}$ there exists $K>0$ such that for $\eta(\vec{p})>0$
\ben\label{corbound1}
\begin{split}
|\pa^w {\cal L}^{0,\Lao}_{N,L}(\vec{p};\LIR)| 
&\le  {|w|!}\ K^{(N+4L-3)(|w|+1)}\,  \eta(\vec{p})^{-(N+|w|-4)}
\mathcal{P}_{L}\left( \log_{+}\frac{|\vec{p}|_\LIR}{\inf(\eta(\vec{p}),M)}
\right) \, ,
\end{split}
\een
 where $\eta$ is defined as in \eqref{etadef}.
\end{cor}
\begin{proof}
For $N+|w|\leq 4$, this follows directly from the bounds \eqref{boundrel1} and \eqref{boundrel2}. For $N+|w|>4$, we start from our bounds established in theorem \ref{thm1} and first rewrite the product of $\la_{i}$ integrals as follows:
\ben\label{intprodorder}
\begin{split}
&\prod_{i \in \mathcal{I}(T)} \int_{0}^{\La_{0}} \d\la_{i}\ f_{\lambda_{i}}(k_i;N,L,\theta_i) 
\mathcal{P}_{L}\left( \log_{+}\sup\left(\frac{|\vec{p}|_\LIR}{\min_{j\in\mathcal{I}}\lambda_{j}},
 \frac{\max_{j\in\mathcal{I}}\lambda_{j}}{\LIR}\right)\right) \\
& =\sum_{\pi\in\mathfrak{S}(\mathcal{I})} \int_{0}^{\La_{0}}\d\la_{\pi_{1}} \int_{\la_{\pi_{1}}}^{\La_{0}}\d\la_{\pi_{2}}\ldots  \int_{\la_{\pi_{|\mathcal{I}|-1}}}^{\La_{0}}\d\la_{\pi_{|\mathcal{I}|}} \\
&\qquad\times f_{\lambda_{\pi_{1}}}(k_{\pi_{1}};N,L,\theta_{\pi_{1}})\cdots f_{\lambda_{\pi_{|\mathcal{I}|}}}(k_{\pi_{|\mathcal{I}|}};N,L,\theta_{\pi_{|\mathcal{I}|}})\mathcal{P}_{L}\left( \log_{+}\sup\left(\frac{|\vec{p}|_\LIR}{\la_{\pi_{1}}},
 \frac{\lambda_{\pi_{|\mathcal{I}|}}}{\LIR}\right)\right)\, ,
\end{split}
\een
where $\mathfrak{S}(\mathcal{I})$ is the set of permutations of the internal lines of the tree $T$ (i.e. the \emph{symmetric group} on $\mathcal{I}$).  
This way of writing the integral allows us to apply lemma \ref{lemmalambdalarge} (see the appendix) in order to bound the integral over $\la_{\pi_{|\mathcal{I}|}}$%the $\la_{i}$ for $i\in\mathcal{I}\setminus\{r\}$. Thus
\ben
\begin{split}
&\prod_{i \in \mathcal{I}(T)} \int_{0}^{\La_{0}} \d\la_{i}\ f_{\lambda_{i}}(k_i;N,L,\theta_i) 
\mathcal{P}_{L}\left( \log_{+}\sup\left(\frac{|\vec{p}|_\LIR}{\min_{j\in\mathcal{I}}\lambda_{j}},
 \frac{\max_{j\in\mathcal{I}}\lambda_{j}}{\LIR}\right)\right) \\
 & \leq\sum_{\pi\in\mathfrak{S}(\mathcal{I})} \int_{0}^{\La_{0}}\d\la_{\pi_{1}} \int_{\la_{\pi_{1}}}^{\La_{0}}\d\la_{\pi_{2}}\ldots  \int_{\la_{\pi_{|\mathcal{I}|-2}}}^{\La_{0}}\d\la_{\pi_{|\mathcal{I}|-1}} \\
 &\qquad\times  f_{\lambda_{\pi_{1}}}(k_{\pi_{1}};N,L,\theta_{\pi_{1}})\cdots f_{\lambda_{\pi_{|\mathcal{I}|-1}}}(k_{\pi_{|\mathcal{I}|-1}};N,L,\theta_{\pi_{|\mathcal{I}|-1}})\cdot \la_{\pi_{|\mathcal{I}|-1}}^{-\theta_{\pi_{|\mathcal{I}|}}} \mathcal{P}_{L}\left( \log_{+}\sup\left(\frac{|\vec{p}|_\LIR}{\la_{\pi_{1}}},
 \frac{\lambda_{\pi_{|\mathcal{I}|-1}}}{\LIR}\right)\right)\, ,
\end{split}
\een
where $L$ dependent factors were absorbed into the polynomials $\mathcal{P}_{L}$. Repeating this procedure, we can bound all the $\la_{\pi_i}$-integrals except the first one (over $\la_{\pi_{1}}$), which leads us to
\ben
\begin{split}
&\prod_{i \in \mathcal{I}(T)} \int_{0}^{\La_{0}} \d\la_{i}\ f_{\lambda_{i}}(k_i;N,L,\theta_i) 
\mathcal{P}_{L}\left( \log_{+}\sup\left(\frac{|\vec{p}|_\LIR}{\min_{j\in\mathcal{I}}\lambda_{j}},
 \frac{\max_{j\in\mathcal{I}}\lambda_{j}}{\LIR}\right)\right) \\
 & \leq\sum_{\pi\in\mathfrak{S}(\mathcal{I})} \int_{0}^{\La_{0}}\d\la_{\pi_{1}}  \, f_{\lambda_{\pi_{1}}}(k_{\pi_{1}};N,L,\theta_{\pi_{1}})\cdot \la_{\pi_{1}}^{-\theta_{\pi_{2}}-\ldots-\theta_{\pi_{|\mathcal{I}|}}}   \mathcal{P}_{L}\left( \log_{+}\sup\left(\frac{|\vec{p}|_\LIR}{\la_{\pi_{1}}},
 \frac{\lambda_{\pi_{1}}}{\LIR}\right)\right)\\
 & \leq(\frac{N-4}{2}+L)! \int_{0}^{\La_{0}}\d\la  \, e^{ - \frac{\eta(\vec{p})^{2}}{\alpha(N,L) \la^{2} } }\, \la^{-(N+|w|-3)}  \mathcal{P}_{L}\left( \log_{+}\sup\left(\frac{|\vec{p}|_\LIR}{\la},
 \frac{\lambda}{\LIR}\right)\right)\, .
\end{split}
\een
In the last line we used the condition $\sum_{i}\theta_{i}=N-4+|w|$, the obvious relation 
\ben\label{etabound}
\inf_{r\in\mathcal{I}(T)}|k_{r}|\ge \eta(\vec{p}) \quad \text{ for }T\in\mathcal{T}_{N,2L,w}(\vec{p})
\een
 as well as the bound \eqref{inlinebound} on the number of internal lines of  a tree, which allows us to bound the sum over index permutations. To bound the remaining $\la$-integral, we use lemma \ref{lemmaLambdaint} from the appendix, which yields
\ben
\begin{split}
&\prod_{i \in \mathcal{I}(T)} \int_{0}^{\La_{0}} \d\la_{i}\ f_{\lambda_{i}}(k_i;N,L,\theta_i) 
\mathcal{P}_{L}\left( \log_{+}\sup\left(\frac{|\vec{p}|_\LIR}{\min_{j\in\mathcal{I}}\lambda_{j}},
 \frac{\max_{j\in\mathcal{I}}\lambda_{j}}{\LIR}\right)\right) \\
 & \leq (\frac{N-4}{2}+L)!\cdot \sqrt{(N-4+|w|)!}\left(\frac{\sqrt{\alpha}}{\eta(\vec{p})}\right)^{N+|w|-4}  \mathcal{P}_{L}\left( \log_{+}\frac{|\vec{p}|_\LIR}{\inf(|\eta(\vec{p})|, M)}\right)\\
 & \leq \sqrt{ |w|!\, (2\alpha)^{|w|}}   \eta(\vec{p})^{-(N+|w|-4)}  \mathcal{P}_{L}\left( \log_{+}\frac{|\vec{p}|_\LIR}{\inf(|\eta(\vec{p})|, M)}\right)\, .
\end{split}
\een
The last inequality follows by absorbing $N$ and $L$ dependent factors into the polynomial $\mathcal{P}_{L}$. Using this bound  in theorem \ref{thm1}, we arrive at
\ben
\begin{split}
|\pa^w {\cal L}^{0,\Lao}_{N,L}(\vec{p};\LIR)| 
&\le  {|w|!}\ K^{(N+4L-4)(|w|+1)}\, (2\alpha)^{\frac{|w|}{2}}\hspace{-.5cm}  \sum_{T \in {\cal T}_{N,2L,w}(\vec{p})}
 \hspace{-.5cm}    \eta(\vec{p})^{-(N+|w|-4)}  \mathcal{P}_{L}\left( \log_{+}\frac{|\vec{p}|_\LIR}{\inf(|\eta(\vec{p})|, M)}\right)\, .
\end{split}
\een
It remains to bound the sum over $T \in {\cal T}_{N,2L,w}$.  Here we use lemma \ref{lemtreebd}, which states that $\mathcal{T}_{N,2L,w}$ satisfies the bound $|\mathcal{T}_{N,2L,w}|\leq N!\cdot 4^{3N-2} \cdot [3(|w|+1)]^{\frac{N-4}{2}+L}$. Absorbing the purely $N$ dependent part into the polynomial $\mathcal{P}_{L}$ and choosing the constant $K$ large enough in order to guarantee (here we denote by $K_{0}$ the constant from theorem \ref{thm1})
\ben
(|w|+1)^{\frac{N-4}{2}+L}\, (2\alpha)^{|w|/2} \, K_{0}^{(N+4L-4)(|w|+1)} \leq K^{(N+4L-3)(|w|+1)}
\een
we finally arrive at the bound \eqref{corbound1}.
\end{proof}
%

%\vspace{.3cm}

\noindent We can also use theorem \ref{thm1} to bound the smeared, connected Schwinger functions:
\begin{cor}\label{cor1b}
For any $N,L\in\mathbb{N}$ there exists $K>0$ such that
\ben\label{cor1beq}
\left|\int_{p_{1},\ldots,p_{N}}  \L_{N,L}^{0,\infty}\left( \vec{p}\right)\ \prod_{i=1}^{N} \frac{\hat{\test}_{i}(p_{i})}{p_i^2} \right| \leq K\,   \LIR^{N}\ \sum_{\mu=0}^{N+2} \ \sum_{{ \mu_1+\ldots+\mu_N=\mu }  }   \frac{\prod_{i=1}^{N}\|\hat{\test}_{i}\|_{\mu_i/2}}{M^\mu} 
\een
for arbitrary test functions $\test_i\in\mathcal{S}(\mathbb{R}^4)$ and for $\|\cdot\|_n$ as in \eqref{Schwnorms}.
\end{cor}
The bound implies that the Schwinger functions are tempered distributions. See appendix \ref{appdistr} for the proof of an analogous bound for the Schwinger functions with two operator insertions. The proof of corollary \ref{cor1b} follows along the same lines. 
\subsection{Schwinger functions with one insertion}\label{sec:1ins}
Using our bounds from the previous subsection, we next give bounds on the CAS's with one insertion. 
Recall from section~\ref{subsec:compfields} that the CAS's with insertion of the composite operator $\O_{A}$, where $A=\{N',w'\}$ is a multi-index, are defined through the flow equation 
\ben\label{CAGFEInsertion}
\begin{split}
&\partial_{\La}\partial_{\vec{p}}^{w}\L^{\La,\Lambda_{0}}_{N,L}(\O_{A};\vec{p};\LIR)= \left(N+2 \atop 2 \right) \, \int_{k} \dot{C}^{\La}(k)\partial_{\vec{p}}^{w}\L^{\La,\Lambda_{0}}_{N+2,L-1}(\O_{A}; k, -k,  \vec{p};\LIR)  \\
&-\hspace{-.3cm}\sum_{\substack{l_{1}+l_{2}=L \\ n_{1}+n_{2}=N+2 \\ v_{1}+v_{2}+v_{3}=w }}\hspace{-.3cm} n_{1}n_{2}\, c_{\{v_{i}\}} \mathbb{S}\, \left[ \partial_{\vec{p}}^{v_{1}}\L^{\La,\Lambda_{0}}_{n_{1},l_{1}}(\O_{A};   p_{1},\ldots,p_{n_{1}-1},q;\LIR)\,  \partial_{\vec{p}}^{v_{3}}\dot{C}^{\La}(q)\,  \partial_{\vec{p}}^{v_{2}} \L^{\La,\Lambda_{0}}_{n_{2},l_{2}}(-q,  p_{n_{1}},\ldots,p_{N};\LIR) \right]\, ,
\end{split}
\een
where $q=p_{n_{1}}+\ldots+p_{N}$,  through the boundary conditions
\ben\label{BCL1b}
\partial^{w}_{\vec{p}}\L^{\LIR, \La_{0}}_{N,L}(\O_{A}(0); \vec{0}; \LIR)= i^{|w|}w! \delta_{w,w'}\delta_{N,N'}\delta_{L,0} \quad \text{ for }N+|w|\leq [A]\, ,
\een
\ben\label{BCL2b}
\partial^{w}_{\vec{p}}\L^{\La_{0},\La_{0}}_{N,L}(\O_{A}(0); \vec{p};\LIR)=0\quad \text{ for }N+|w|>[A] \quad ,
\een
where $0<\LIR<\La_{0}$ is a finite renormalisation scale, and through the translation property \eqref{CAGtrans}. For the sake of simplicity we assume $N'$ to be even in the following, i.e. we consider \emph{even} monomials in $\varphi$. The odd case can be treated similarly. It follows directly from the flow equation and boundary conditions that $\L^{\La,\Lambda_{0}}_{0,0}(\O_{A};\LIR)=0$ (assuming $[A]>0\,$).  For other values of $N,L$, we have: 
\begin{thm}\label{thmCAGins1}
For any $N,L\in\mathbb{N}$ with $N+L>0$ and $w\in\mathbb{N}^{4N}$  there exists $\,K>0\,$ such that for $0\leq\La\leq M$, and the following bound holds:
\ben\label{CAG1boundeq}
\begin{split}
&|\pa^w {\cal L}^{\La,\Lao}_{N,L}(\O_{A}(0);\vec{p};\LIR)| 
\le  \sqrt{|w|!\, |w'|!}\ K^{(2N+8L-4)(|w|+1)}K^{[A](N/2+2L)^3 } \sum_{v_{1}+v_{2}=w} 
  \\
&\times \sum_{T \in \tkn_{N,2L,v_{1}}(\vec{p})} \LIR^{[A]-N_{\mathcal{V}}-|v_{2}|}  \sum_{\mu=0}^{[A](N+2L+1)}\frac{1}{\sqrt{\mu!}}\
\left(\frac{|\vec p_{\mathcal{V}}|}{\LIR}\right)^{\mu}     \prod_{i \in \mathcal{I}(T)} \int_{\La}^{\La_{0}} \d\la_{i}\ f_{\lambda_{i}}(k_i;N,L,\theta_i) 
\\
&\times\,
 \mathcal{P}_{2L+\frac{N}{2}-1}\left( \log_{+}\sup\left( \frac{|\vec{p}|}{\LIR}, \frac{|\vec{p}|_\LIR}{\min_{j\in\mathcal{I}}\lambda_{j}},
 \frac{\max_{j\in\mathcal{I}}\lambda_{j}}{\LIR}\right)\right)\, .
\end{split}
\een
Here 
%$\mathcal{P}_{n}$ are polynomials of degree at most $n$ with non-negative coefficients depending on $N$ and $L$, with 
$f$ is given in \eqref{fdef} and the trees $\tkn_{N,2L,w}$ and 
momentum $\vec{p}_{\mathcal{V}}$ are defined in def.~\ref{deftrees2}.% and with $D=[A]$.
\end{thm}
\paragraph{Remark:} We use the convention \eqref{emptyprod} for the case $\mathcal{I}=\emptyset$, and it is understood that\\  $\sup\left( \frac{|\vec{p}|}{\LIR}, \frac{|\vec{p}|_\LIR}{\min_{j\in\mathcal{I}}\lambda_{j}},
 \frac{\max_{j\in\mathcal{I}}\lambda_{j}}{\LIR}\right)= \frac{|\vec{p}|}{\LIR}$ in that case.
\begin{proof}
Note first that we can again restrict to even $N$, since the CAS functions with insertion of an even monomial in $\varphi$ vanish due to the symmetry $\varphi\to- \varphi$. Our general strategy of proof is as follows:
\begin{enumerate}
\item 
Our bounds obtained for massive fields~\cite{Hollands:2011gf} imply corresponding bounds for massless fields as long as the infrared cutoff $\Lambda$ is 
above the chosen renormalisation scale $M$.\footnote{The underlying reason for this is that for $\La\geq M$ one has $|\dot{C}^{\La}(p;m=0)|\leq \frac{1}{e} |\dot{C}^{\La}(p;m=M)|$, i.e. the massless propagator is bounded by the one with mass $M$ up to a constant.} More precisely, we obtain\footnote{
Strictly speaking, the bounds in~\cite{Hollands:2011gf} are derived for different boundary conditions. The condition \eqref{BCL1} is taken at $M=0$ there. %In fact, these bounds diverge logarithmically in the limit $m\to 0$, even for finite $\La=M$. 
  The proof of the bounds however remains largely unaffected, so we are not going to repeat it here. We note that, using our boundary conditions \eqref{BCL1}, one can replace the mass parameter $m$ appearing in the logarithms by the scale $M$. This is found by adapting the inequality (89) in \cite{Hollands:2011gf} to our situation, where $\La\geq M$ in the induction.
\label{foot}
}
\ben\label{boundCAG1}
\begin{split}
|\partial^{w}_{\vec{p}}\L^{\LIR,\Lambda_{0}}_{ N,L}(\O_{A}(0); \vec{p}; \LIR)|\leq\, & \sqrt{|w|!\, |w'|!}\, K^{(2N+8L-4)|w|}\, K^{[A](N/2+2L)^{3}}\,\LIR^{[A]-N-|w|}  \\
\times & \sum_{\mu=0}^{[A](N+2L+1)}\frac{1}{\sqrt{\mu!}}\left(\frac{|\vec{p}|}{\LIR}\right)^{\mu}\,  \mathcal{P}_{2L+\frac{N}{2}-1}\left( \log_{+}\frac{|\vec{p}|}{\LIR}\right)\, .
\end{split}
\een
\item We can then bound $|\pa^w {\cal L}^{\La,\Lao}_{N,L}(\O_{A};\vec{p};\LIR)|$ for $0\leq \La\leq \LIR$ by integrating the flow equation \eqref{CAGFEInsertion} over $\la$ between $\La$ and $\LIR$, using the bound established in step 1 as boundary condition at the upper limit of integration. This amounts to ``integrating out'' the momenta below the chosen scale $M$ (IR-region).
\end{enumerate}
Thus, we only need to complete the second step mentioned above, i.e. the integration of the flow equation \eqref{CAGFEInsertion} from $\Lambda = M$ to small $\Lambda$. In this step, we will again use the induction scheme based on the flow equations, which goes up in $N+2L$, for given $N+2L$ 
ascends in $L$, and for given $N,L$ descends in $|w|$. At each induction step, we have 
three contributions: A boundary term from $\Lambda = M$, the integral of the first term on the right side, and the integral of the second term on the right side from $\La = M$ to general $\La$.
We look at these separately.

\subsubsection{Boundary contributions}\label{subsec:BC}
When integrating the flow equation \eqref{CAGFEInsertion} between $\La$ and $M$, we obtain boundary contributions from the upper limit of integration, i.e.
\ben
\L^{\La,\La_{0}}_{N,L}(\O_{A};\vec{p})=\L^{M,\La_{0}}_{N,L}(\O_{A};\vec{p}) - \int_{\La}^{M}\d\la\, \text{``r.h.s. of eq.\eqref{CAGFEInsertion}''}\, .
\een
These boundary contributions $\L^{M,\La_{0}}_{N,L}(\O_{A};\vec{p})$ satisfy the bound \eqref{boundCAG1}, which, crucially, is consistent with our hypothesis \eqref{CAG1boundeq} (it corresponds to the case $N_{\mathcal{V}}=N$, i.e. all external legs directly attached to the vertex $\mathcal{V}$). 

With boundary contributions taken care of, we are now ready to verify the induction step, i.e. we verify that the integral over the r.h.s. of the flow equation \eqref{CAGFEInsertion} reproduces the inductive bound \eqref{CAG1boundeq}.

\subsubsection{First term on the r.h.s. of the flow equation \eqref{CAGFEInsertion}}
Inserting our inductive bound, \eqref{CAG1boundeq}, for the first term in the flow equation, \eqref{CAGFEInsertion}, and integrating over $\la$, we find the bound

\ben\label{CAGins1st1}
\begin{split}
&\left|\int_{\La}^{M}\d\la\, \left(N+2 \atop 2 \right) \, \int_{\ell} \dot{C}^{\la}(\ell)\partial_{\vec{p}}^{w} \L^{\la,\Lambda_{0}}_{N+2,L-1}(\O_{A}(0); 
\ell, -\ell,  \vec{p};\LIR) \right|\\
&\leq \left(N+2 \atop 2 \right) \sqrt{|w|!\, |w'|!}\ K^{(2N+8L-8)(|w|+1)}K^{[A](N/2+2L-1)^3 } \sum_{v_{1}+v_{2}=w} \sum_{T \in \tkn_{N+2,2L-2,(0,0,v_{1})}(\ell,-\ell, \vec{p})}
  \\
&\times\LIR^{[A]-N_{\mathcal{V}}-|v_{2}|} \int_{\Lambda}^{\LIR} \d\la \int_{\ell}\ \frac{2e^{-\frac{\ell^{2}}{\la^{2}}}}{\la^{3}} \sum_{\mu=0}^{[A](N+2L+1)}\frac{1}{\sqrt{\mu!}}\
\left(\frac{|\vec p_{\mathcal{V}}(\ell)|}{\LIR}\right)^{\mu} \prod_{i \in \mathcal{I}(T)} \int_{\la}^{\La_{0}} \d\la_{i}   \\
&\times    f_{\lambda_{i}}(k_i(\ell);N+2,L-1,\theta_i) 
 \mathcal{P}_{2L+\frac{N}{2}-2}\left( \log_{+}\sup\left(\frac{|(\vec{p},\ell)|}{\LIR},\frac{|(\vec{p},\ell)|_\LIR}{\min_{j\in\mathcal{I}}\lambda_{j}},
 \frac{\max_{j\in\mathcal{I}}\lambda_{j}}{\LIR}\right)\right)\, .
 \end{split}
\een
Here we have again used the notation $k_{i}(\ell)$ and $\vec{p}_{\mathcal{V}}(\ell)$ in order to indicate that these momenta may implicitly depend on the ``loop variable'' $\ell$. Recall from \eqref{internalboundtkn} that $|\mathcal{I}(T)|\leq \frac{N}{2}+L$, which allows us to use lemma \ref{lablemma1}
in order to bound the $\ell$-integral. Noting further that $(N/2+2L-1)^{3}=(N/2+2L)^{3}-3(N/2+2L)(N/2+2L-1)-1$, we arrive at the bound (absorbing constants and purely $N,L$ dependent factors into the polynomial $\mathcal{P}$)
\ben\label{CAGins1st2}
\begin{split}
&|\text{r.h.s. of \eqref{CAGins1st1}}|\leq  \sqrt{|w|!\, |w'|!}\ K^{(2N+8L-8)(|w|+1)}K^{[A](N/2+2L)^3} K^{-3[A](N/2+2L)(N/2+2L-1)-[A]}\  2^{[A](N+2L+1)} 
  \\
&\times \sum_{v_{1}+v_{2}=w} \sum_{T \in \tkn_{N+2,2L-2,(0,0,v_{1})}(0,0,\vec{p})}\LIR^{[A]-N_{\mathcal{V}}-|v_{2}|}  \int_{\Lambda}^{\LIR} \d\la\ \la \sum_{\mu=0}^{[A](N+2L+1)}\frac{1}{\sqrt{\mu!}}\
\left(\frac{|\vec p_{\mathcal{V}}|}{\LIR}\right)^{\mu}  \\
&\times   \prod_{i \in \mathcal{I}(T)} \int_{\la}^{\La_{0}} \d\la_{i}\ f_{\lambda_{i}}(k_i;N,L,\theta_i) 
 \mathcal{P}_{2L+\frac{N}{2}-2}\left( \log_{+}\sup\left( \frac{|\vec{p}|}{\LIR}, \frac{|\vec{p}|_\LIR}{\min_{j\in\mathcal{I}}\lambda_{j}},
 \frac{\max_{j\in\mathcal{I}}\lambda_{j}}{\LIR}\right)\right)\, .
\end{split}
\een
We can bound the $\la$-integral as in \eqref{lambdaintsest}. Taking into account also the fact $\la\leq M$, this yields a factor $\min_{j\in\mathcal{I}}(\la_{j}, M)^{2}$. Thus, we have the bound
\ben\label{CAGins1st1case2b}
\begin{split}
& |\text{r.h.s. of \eqref{CAGins1st2}}|\leq \sqrt{|w|!\, |w'|!}\ K^{(2N+8L-8)(|w|+1)}K^{[A](N/2+2L)^3} K^{-3[A](N/2+2L)(N/2+2L-1)-[A]} \ 2^{[A](N+2L+1)} 
  \\
&\times \sum_{v_{1}+v_{2}=w} \sum_{T \in \tkn_{N+2,2L-2,(0,0,v_{1})}(0,0,\vec{p})}\LIR^{[A]-N_{\mathcal{V}}-|v_{2}|}   \sum_{\mu=0}^{[A](N+2L+1)}\frac{1}{\sqrt{\mu!}}\
\left(\frac{|\vec p_{\mathcal{V}}|}{\LIR}\right)^{\mu}\prod_{i \in \mathcal{I}(T)} \int_{\La}^{\La_{0}} \d\la_{i}  \\
&\times    f_{\lambda_{i}}(k_i;N,L,\theta_i) 
 \mathcal{P}_{2L+\frac{N}{2}-2}\left( \log_{+}\sup\left(\frac{|\vec{p}|}{\LIR},\frac{|\vec{p}|_\LIR}{\min_{j\in\mathcal{I}}\lambda_{j}},
 \frac{\max_{j\in\mathcal{I}}\lambda_{j}}{\LIR}\right)\right) \cdot \min_{j\in\mathcal{I}}(\la_{j},M)^{2} \, .
\end{split}
\een
Note that each factor of $\min_{j\in\mathcal{I}}(|\la_{j}|,M)$ can be absorbed by decreasing the value $\theta_{i}$ for some internal line $i\in\mathcal{I}(T)$ [cf. the discussion following \eqref{lambdaintsest}], or alternatively by decreasing the value of $N_{\mathcal{V}}$ (keeping $N-N_{\mathcal{V}}$ fixed). This procedure of decreasing weights can again be organised with the help of a ``reduction procedure on trees''. 
\paragraph{Reduction procedure for trees \rom{2}:} Given a tree $T\in{\tkn}_{N+1,L-1,(0,w)}(\ell,\vec{p})$ we define the reduced tree $\mathcal{R}_{\ell}(T)\in{\tkn}_{N,L,w}(\vec{p})$ as follows:
\begin{enumerate}
\item If the external line with momentum $\ell$ is not directly attached to the vertex $\mathcal{V}$, then proceed as specified on page \pageref{reduct}. This reduces the number of external lines, $N$, as well as the weight $\theta_{j}$ of some adjacent  internal line by one. The reduction of the weight $\theta_{j}$ corresponds to absorbing a factor of $\la_{j}$.
\item If the external line with momentum $\ell$ is directly attached to the vertex $\mathcal{V}$, then simply delete this line. This reduces the value of both $N$ and $N_{\mathcal{V}}$ by one.
 %while  keeping $N-N_{\mathcal{V}}$ fixed. 
 This procedure corresponds to absorbing a factor of $M$.
\end{enumerate}
Thus, by applying this reduction procedure twice, we can remove the external legs with associated momenta $\ell,-\ell$, which allows us to express the bound \eqref{CAGins1st1case2b} in terms of trees in ${\tkn}_{N,L,w}(\vec{p})$ (see fig. \ref{fig:LoopCase3} for an example). 
\begin{figure}[htbp]
\begin{center}
\includegraphics{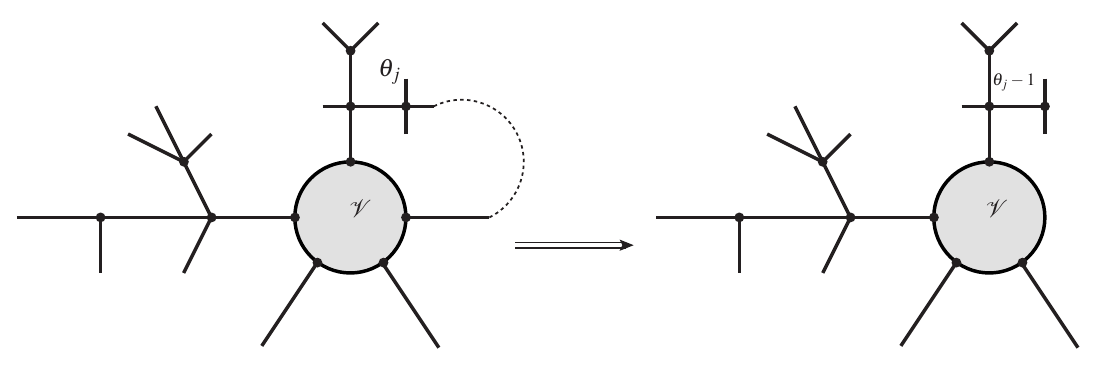}
\end{center}
\caption{To express our bound in terms of trees $T' \in \tkn_{N,2L,v_{1}}(\vec{p})$, we delete from $T \in \tkn_{N+2,2L-2,(0,0,v_{1})}(\ell,-\ell,\vec{p})$ the external lines carrying momentum $\ell,-\ell$ using the reduction procedure.}
\label{fig:LoopCase3}
\end{figure}

As in the case without insertion, the reduction procedure is not injective, i.e. different trees in $T\in{\tkn}_{N+1,L-1,(0,w)}(\ell,\vec{p})$ may yield the same reduced tree (see the discussion following \eqref{R2bound}). To account for this, we again multiply our bound by the factor \eqref{totalredfactor} (it is not hard to check that this factor is again suitable). The resulting bound is 
\ben\label{CAGins1st1case2bfin}
\begin{split}
&\left|\int_{\La}^{M}\d\la\, \left(N+2 \atop 2 \right) \, \int_{\ell} \dot{C}^{\la}(\ell)\partial_{\vec{p}}^{w}\L^{\la,\Lambda_{0}}_{N+2,L-1}(\O_{A}(0); \ell, -\ell,  \vec{p};\LIR) \right|\\
& \leq (|w|+1)^{2} \sqrt{|w|!\, |w'|!}\ K^{(2N+8L-8)(|w|+1)}K^{[A](N/2+2L)^3} K^{-3[A](N/2+2L)(N/2+2L-1)-[A]} \ 2^{[A](N+2L+1)} 
  \\
&\times \sum_{v_{1}+v_{2}=w} \sum_{T \in \tkn_{N,2L,v_{1}}(\vec{p})} \LIR^{[A]-N_{\mathcal{V}}-|v_{2}|} \sum_{\mu=0}^{[A](N+2L+1)}\frac{1}{\sqrt{\mu!}}\
\left(\frac{|\vec p_{\mathcal{V}}|}{\LIR}\right)^{\mu}    \\
&\times   \prod_{i \in \mathcal{I}(T)} \int_{\La}^{\La_{0}} \d\la_{i}\ f_{\lambda_{i}}(k_i;N,L,\theta_i) 
 \mathcal{P}_{2L+\frac{N}{2}-2}\left( \log_{+}\sup\left(\frac{|\vec{p}|}{\LIR}, \frac{|\vec{p}|_\LIR}{\min_{j\in\mathcal{I}}\lambda_{j}},
 \frac{\max_{j\in\mathcal{I}}\lambda_{j}}{\LIR}\right)\right)
\end{split}
\een
which is consistent with our hypothesis \eqref{CAG1boundeq} provided that
\ben
 K^{-3[A](N/2+2L)(N/2+2L-1)-[A]-4(|w|+1)} \ 2^{[A](N+2L+1)}\, (|w|+1)^{2} \leq 1 \, .
\een
\subsubsection{Second term on the r.h.s. of the flow equation \eqref{CAGFEInsertion}}
We distinguish the cases $n_{2}\geq 4$ and $n_{2}=2$ in the sum over $n_{1}+n_{2}=N+2$.
\paragraph{\textbf{For $n_{2}\geq 4$:}} 
The bound \eqref{mainboundcomb} on the CAS's without insertion combined with the inductive bound \eqref{CAG1boundeq} for the CAS's with one insertion implies that this contribution to the second term on the r.h.s. of the flow equation satisfies the bound
\ben\label{irr2ndins}
\begin{split}
&\left|\int_{\La}^{M}\d\la\sum_{\substack{ {v_{1}+v_{2}+v_{3}=w} \\ n_{1}+n_{2}=N+2,\, n_{2}\geq 4 \\  {l_{1}+l_{2}=L} } } n_{1}n_{2}c_{\{v_{i}\}} \partial_{\vec{p}}^{v_{1}}\L^{\la,\Lambda_{0}}_{n_{1},l_{1}}(\O_{A}(0);   \vec{p}_{1};\LIR)\,  \partial_{\vec{p}}^{v_{3}}\dot{C}^{\la}(q)\,  \partial_{\vec{p}}^{v_{2}} \L^{\la,\Lambda_{0}}_{n_{2},l_{2}}( \vec{p}_{2};\LIR)\right|\\
&\leq\int_{\Lambda}^{\LIR} \d\la \sum_{\substack{  n_{1}+n_{2}=N+2,\, n_{2}\geq 4 \\  {l_{1}+l_{2}=L} } }\sum_{\substack{v_{1}+v_{2}+v_{3}=w \\ (v_{2})_{i}=0=(v_{3})_{i}\, \forall i <n_{1}  }}  n_{1}n_{2}c_{\{v_{i}\}} 2k \frac{\sqrt{|v_{3}|!\, 2^{|v_{3}|} }\  e^{-\frac{q^{2}}{2\la^{2}}}}{\la^{3+|v_{3}|}}\\
&  \times \bigg[  \sqrt{|v_{1}|!\, |w'|!}\ K^{(2n_{1}+8l_{1}-4)(|v_{1}|+1)}K^{[A](n_{1}/2+2l_{1})^3 } \sum_{u_{1}+u_{2}=v_{1}} \sum_{T_{1} \in \tkn_{n_{1},2l_{1},\tilde{u}_{1}}(\vec{p}_{1})}
  \\
&\times \LIR^{[A]-N_{\mathcal{V}}-|u_{2}|} \sum_{\mu=0}^{[A](n_{1}+2l_{1}+1)}\frac{1}{\sqrt{\mu!}}\
\left(\frac{|\vec p_{\mathcal{V}}|}{\LIR}\right)^{\mu}     \prod_{i \in \mathcal{I}(T_{1})} \int_{\la}^{\La_{0}} \d\la_{i}\ f_{\lambda_{i}}(k_i;n_{1},l_{1},\theta_i) 
\\
&\times\,
 \mathcal{P}_{2l_{1}+\frac{n_{1}}{2}-1}\left( \log_{+}\sup\left(\frac{|\vec{p}_{1}|}{\LIR},\frac{|\vec{p}_{1}|_\LIR}{\min\limits_{j\in\mathcal{I}(T_{1})}\lambda_{j}},
 \frac{\max\limits_{j\in\mathcal{I}(T_{1})}\lambda_{j}}{\LIR}\right)\right)
\bigg]\\
&\times\bigg[    \sqrt{|v_{2}|!}\ K^{(n_{2}+4l_{2}-4)(|v_{2}|+1)}\,   \sum_{T_{2} \in {\cal T}_{n_{2},2l_{2},\tilde{v}_{2}}(\vec{p}_{2})}
\\
&\times  %\La_{\ma(\vec{a},T)+1}^{-S} % \prod_{b\in\vec{a}} \La_{b+1}^{-b+ d_{b,\vec{a}}} 
 \prod_{i \in \mathcal{I}(T_{2})} \int_{\la}^{\La_{0}} \d\la_{i}\ f_{\lambda_{i}}(k_i;n_{2},l_{2},\theta_i) 
\mathcal{P}_{l_{2}}\left( \log_{+}\sup\left(\frac{|\vec{p}_{2}|_\LIR}{\la},
 \frac{\max\limits_{j\in\mathcal{I}(T_{2})}(\lambda_{j},\la)}{\LIR}\right)\right) \bigg]\, .
\end{split}
\een
Here we used the notation $\vec{p}_{1}=(p_{1},\ldots,p_{n_{1}-1},q)$ and $\vec{p}_{2}=( p_{n_{1}},\ldots,p_{N},-q)$ with $q=p_{n_{1}}+\ldots+p_{N}$, and we write $v_{i}=((v_{i})_{1},\ldots,(v_{i})_{N})$ with $(v_{i})_{j}\in\mathbb{N}^{4}$. Note that the l.h.s. of \eqref{irr2ndins} vanishes unless $(v_{2})_{i}=0=(v_{3})_{i}$ for $i<n_{1}$, since neither $q$, nor $\vec{p}_{2}$ depend on the variables $p_{1},\ldots,p_{n_{1}-1}$. This explains the restriction in the sum over $v_{1},v_{2},v_{3}$ in the second line. In the summation over the trees, we have used the notation
\ben
\tilde{u}_{1}:=((u_{1})_{1},\ldots,(u_{1})_{n_{1}-1},\quad 
\sum_{j=n_{1}}^{N}(u_{1})_{j})\in\mathbb{N}^{4n_{1}}\ \ ,
\een
\ben
\tilde{v}_{2}:=((v_{2})_{n_{1}},\ldots,(v_{2})_{N})\in\mathbb{N}^{4(n_{2}-1)}\ .
\een
Using the inequality $n_{1}/2+2l_{1}\leq N/2+2L-1$ (cf. footnote \ref{footni}) in order to bound the exponent of $K$, bounding  the factorial factors with the help of \eqref{factorialtriv} and combining the two polynomials into one as in \eqref{logsmerge} (the degree of this polynomial is at most $2l_{1}+n_{1}/2-1+l_{2}\leq 2L+N/2-1$), we find that \eqref{irr2ndins} is smaller than
\ben\label{irr2ndinsb}
\begin{split}
&\int_{\Lambda}^{\La_{0}} \d\la \hspace{-.4cm} \sum_{\substack{  n_{1}+n_{2}=N+2,\, n_{2}\geq 4 \\  {l_{1}+l_{2}=L} } }\sum_{\substack{v_{1}+v_{2}+v_{3}=w \\ (v_{2})_{i}=0=(v_{3})_{i}\, \forall i <n_{1}  }} n_{1}n_{2}c_{\{v_{i}\}} 2k \frac{\sqrt{|w|!\, 2^{|w|} }\  e^{-\frac{q^{2}}{2\la^{2}}}}{\la^{3+|v_{3}|}}\\
&  \times \bigg[  \sqrt{ |w'|!}\ K^{(2N+8L-6)(|w|+1)}K^{[A](N/2+2L-1)^3 } \sum_{u_{1}+u_{2}=v_{1}} \sum_{{T_{1} \in \tkn_{n_{1},2l_{1},\tilde{u}_{1}}(\vec{p}_{1})}\atop {T_{2} \in {\cal T}_{n_{2},2l_{2},\tilde{v}_{2}}(\vec{p}_{2})}}
  \\
&\times \LIR^{[A]-N_{\mathcal{V}}-|u_{2}|}\, \sum_{\mu=0}^{[A](N+2L+1)}\frac{1}{\sqrt{\mu!}}\
\left(\frac{|\vec p_{\mathcal{V}}|}{\LIR}\right)^{\mu}    \prod_{i \in \mathcal{I}(T_{1})\cup\mathcal{I}(T_{2})} \int_{\la}^{\La_{0}}\d\la_{i} f_{\la_{i}}(k_i;N,L,\theta_i)\\
&\times\mathcal{P}_{2L+\frac{N}{2}-1}\left( \log_{+}\sup\left(\frac{|\vec{p}|}{\LIR},\frac{|\vec{p}|_\LIR}{\la},
 \frac{\max\limits_{j\in\mathcal{I}(T_{1})\cup\mathcal{I}(T_{2})}(\lambda_{j},\la)}{\LIR}\right)\right)
\bigg] \ \ . \\
\end{split}
\een
Here we also used the bound $ f_{\la_{i}}(k_i;n_{a},l_{a},\theta_i)\leq  f_{\la_{i}}(k_i;N,L,\theta_i)$, where $a\in\{1,2\}$, which holds since $n_{a}/2+2l_{a}\leq N/2+2L$. 

In order to express the bound \eqref{irr2ndinsb} in terms of trees $T \in \tkn_{N,2L,u_{1}+v_{2}+v_{3}}(\vec{p})$ 
we perform a ``gluing procedure'', which is analogous to the one  described in fig.\ref{fig:merge1} and the corresponding discussion. Namely, we take the trees $T_{1} \in \tkn_{n_{1},2l_{1},\tilde{u}_{1}}(\vec{p}_{1})$ and $T_{2} \in {\cal T}_{n_{2},2l_{2},\tilde{v}_{2}}(\vec{p}_{2})$ and join them along the external lines carrying momentum $q$ and $-q$, respectively. The resulting new internal line therefore carries momentum $q$, and we associate the weights $\rho=2$, $\sigma=|v_{3}|$ to it (see fig.\ref{fig:Gluecase1} for a sketch of the procedure). If the procedure yields a vertex of coordination number two adjacent to two internal lines (this is the case if one of the joined external lines is adjacent to a vertex of coordination number two), we remove this vertex and join the internal lines, adding up their weights.
\begin{figure}[htbp]
\begin{center}
\includegraphics{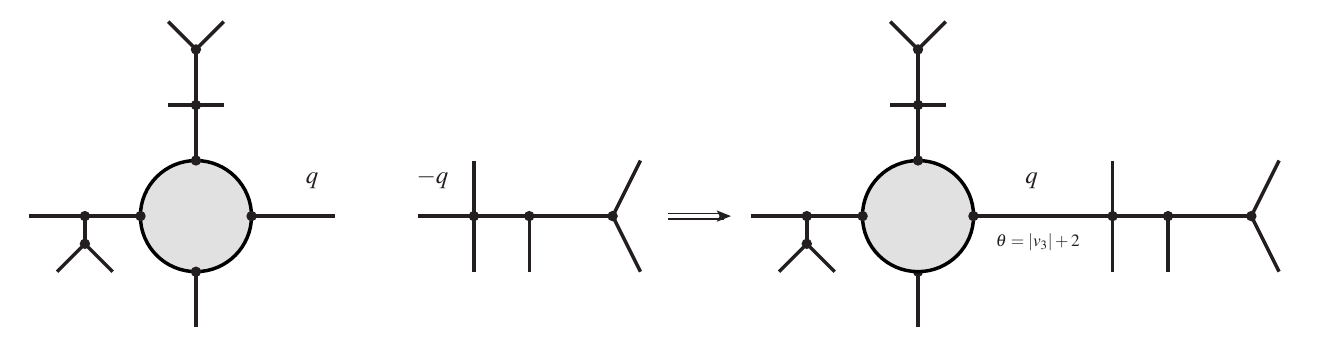}
\end{center}
\caption{Gluing the trees $T_{1} \in \tkn_{n_{1},2l_{1},\tilde{u}_{1}}(\vec{p}_{1})$ and $T_{2} \in {\cal T}_{n_{2},2l_{2},\tilde{v}_{2}}(\vec{p}_{2})$ in the way sketched above, we obtain a tree $T \in \tkn_{N,2L,u_{1}+v_{2}+v_{3}}(\vec{p})$ }
\label{fig:Gluecase1}
\end{figure}

It is straightforward to check that the resulting tree is indeed in $\tkn_{N,2L,u_{1}+v_{2}+v_{3}}(\vec{p})$. Note also that for fixed values of  $u_{1},v_{2},v_{3}$, the procedure is injective (i.e. it yields different trees $T$ for different $T_{1},T_{2}$). Thus, renaming the summation variables $u_{1}+v_{2}+v_{3}\to v_{1}$ and $u_{2}\to v_{2}$ and using the inequality \eqref{cwest2} to bound the sum over different combinations of $u_{1}, v_{2}, v_{3}$, we conclude that \eqref{irr2ndinsb} implies the bound
\ben\label{irr2ndinsc}
\begin{split}
&2^{|w|/2}\cdot 3^{|w|}\,  \sum_{{v_{1}+v_{2}=w}} \bigg[  \sqrt{|w|! |w'|!}\ K^{(2N+8L-6)(|w|+1)}K^{[A](N/2+2L-1)^3 } 
  \\
&\times  \sum_{{T \in \tkn_{N,2L,v_{1}}(\vec{p})}} \LIR^{[A]-N_{\mathcal{V}}-|v_{2}|} \sum_{\mu=0}^{[A](N+2L+1)}\frac{1}{\sqrt{\mu!}}\
\left(\frac{|\vec p_{\mathcal{V}}|}{\LIR}\right)^{\mu}     \prod_{i \in \mathcal{I}(T)} \int_{\La}^{\La_{0}}\d\la_{i} f_{\la_{i}}(k_i;N,L,\theta_i)
\\
&\times\, 
 \mathcal{P}_{2L+\frac{N}{2}-1}\left( \log_{+}\sup\left(\frac{|\vec{p}|}{\LIR},\frac{|\vec{p}|_\LIR}{\min_{j\in\mathcal{I}}\lambda_{j}},
 \frac{\max_{j\in\mathcal{I}}\lambda_{j}}{\LIR}\right)\right)
\bigg]\, .
\end{split}
\een
Here the sums over $l_{1},l_{2}$ and $n_{1},n_{2}$ have been bounded by $L+1$ and $N^{3}$, respectively, and we have absorbed these factors into the polynomial $\mathcal{P}$. Raising the upper limit of integration for the integral over $\la$ to $\La_{0}$,
% and renaming integration variables, 
 we therefore arrive at a bound that is consistent with our claim, \eqref{CAG1boundeq}, provided that $K$ is chosen large enough that
\ben\label{127bound}
2^{|w|/2}\cdot 3^{|w|} \leq K^{2|w|}\, .
\een
\paragraph{\textbf{For $n_{2}=2$:}}
In this case, we combine the bound \eqref{N2estweak} with our inductive bound \eqref{CAG1boundeq}, which implies that the contribution at hand satisfies the inequality
\ben\label{irr2ndins1c}
\begin{split}
&\left|\int_{\La}^{M}\d\la\sum_{\substack{v_{1}+v_{2}+v_{3}=w \\ l_{1}+l_{2}=L  }}2N\,c_{\{v_{i}\}} \partial_{\vec{p}}^{v_{1}}\L^{\la,\Lambda_{0}}_{N,l_{1}}(\O_{A}(0);   p_{1},\ldots,p_{N};\LIR)\,  \partial_{\vec{p}}^{v_{3}}\dot{C}^{\la}(p_{N})\,  \partial_{\vec{p}}^{v_{2}} \L^{\la,\Lambda_{0}}_{2,l_{2}}(  -p_{N},p_{N};\LIR)\right|\\
&\leq\int_{\Lambda}^{\LIR} \d\la  \sum_{\substack{v_{1}+v_{2}+v_{3}=w \\ l_{1}+l_{2}=L   }}2Nc_{\{v_{i}\}} 2k \frac{\sqrt{|v_{3}|!\, 2^{|v_{3}|} }\  e^{-\frac{q^{2}}{2\la^{2}}}}{\la^{3+|v_{3}|}}  \Bigg[  \sqrt{|v_{1}|!\, |w'|!}\ K^{(2N+8l_{1}-4)(|v_{1}|+1)}K^{[A](N/2+2l_{1})^3 }  \\
&  \times\sum_{u_{1}+u_{2}=v_{1}} \sum_{T \in \tkn_{N,2l_{1},u_{1}}(\vec{p})}
 \LIR^{[A]-N_{\mathcal{V}}-|u_{2}|} \sum_{\mu=0}^{[A](N+2l_{1}+1)}\frac{1}{\sqrt{\mu!}}\
\left(\frac{|\vec p_{\mathcal{V}}|}{\LIR}\right)^{\mu}     \prod_{i \in \mathcal{I}(T)} \int_{\la}^{\La_{0}} \d\la_{i}\ f_{\lambda_{i}}(k_i;N,l_{1},\theta_i) 
\\
&\times\,
 \mathcal{P}_{2l_{1}+\frac{N}{2}-1}\left( \log_{+}\sup\left(\frac{|\vec{p}|}{\LIR},\frac{|\vec{p}|_\LIR}{\min_{j\in\mathcal{I}}\lambda_{j}},
 \frac{\max_{j\in\mathcal{I}}\lambda_{j}}{\LIR}\right)\right)
\Bigg]\\
&\times \sqrt{|v_{2}|!}\ K^{(4l_{2}-2)(|v_{2}|+1)}\, |q|_{\la}^{2}\, \la^{|v_{2}|}\, \mathcal{P}_{l_{2}-1}\left( \log_{+}\sup\left(\frac{|q|_\LIR}{\la},
 \frac{\la}{\LIR}\right)\right)\\
&\leq \int_{\Lambda}^{\La_{0}} \d\la  \sum_{\substack{v_{1}+v_{2}+v_{3}=w \\ l_{1}+l_{2}=L  }} \frac{c_{\{v_{i}\}}\sqrt{ 2^{|w|} }\  e^{-\frac{q^{2}}{4\la^{2}}}}{\la^{1+|v_{2}|+|v_{3}|}}    \sqrt{|w|!\, |w'|!}\ K^{(2N+8L-6)(|w|+1)}K^{[A](N/2+2L-1)^3 } \\
&\times\Bigg[\sum_{u_{1}+u_{2}=v_{1}} \sum_{T \in \tkn_{N,2l_{1},u_{1}}(\vec{p})}\LIR^{[A]-N_{\mathcal{V}}-|u_{2}|} \, \sum_{\mu=0}^{[A](N+2L+1)}\frac{1}{\sqrt{\mu!}}\
\left(\frac{|\vec p_{\mathcal{V}}|}{\LIR}\right)^{\mu}   \\
&\times\prod_{i\in\mathcal{I}(T)} \int_{\La}^{\La_{0}} \d\la_{i}\ f_{\lambda_{i}}(k_i;N,L,\theta_i)   \mathcal{P}_{2L+\frac{N}{2}-2}\left( \log_{+}\sup\left(\frac{|\vec{p}|}{\LIR},\frac{|\vec{p}|_\LIR}{\min\limits_{j\in\mathcal{I}}(\lambda_{j},\la)},
 \frac{\max\limits_{j\in\mathcal{I}}(\lambda_{j},\la)}{\LIR}\right)\right)
\Bigg]\, .
\end{split}
\een
Here we used \eqref{qest}, combined the logarithmic polynomials as in \eqref{logsmerge} and absorbed $N$ and $L$ dependent factors into the new polynomial. We can express the bound in terms of trees $T' \in \tkn_{N,2L,u_{1}+v_{2}+v_{3}}(\vec{p})$ by adding a vertex of coordination number two to the external line of $T \in \tkn_{N,2l_{1},u_{1}}(\vec{p})$ carrying momentum $q$, see fig.\ref{fig:Gluecase1c}. To the resulting new internal line, which also carries momentum $q$, we associate the weights $\rho=0$, $\sigma=|v_{2}|+|v_{3}|$. If this process has created a vertex of coordination number two adjacent to two internal lines, then we remove this vertex and fuse the adjacent lines, adding up their weights. 
\begin{figure}[htbp]
\begin{center}
\includegraphics{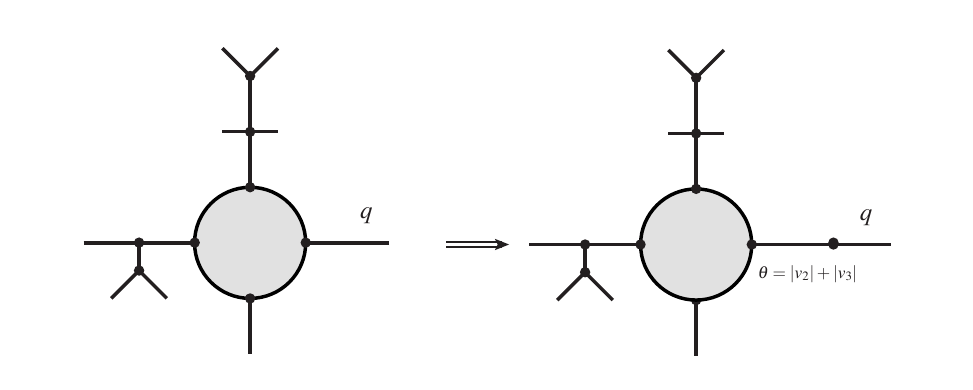}
\end{center}
\caption{We add a vertex of coordination number two to the external line in $T_{1} \in \tkn_{N,2l_{1},u_{1}}(\vec{p})$ which carries momentum $q$. This yields a tree $T \in \tkn_{N,2L,u_{1}+v_{2}+v_{3}}(\vec{p})$.}
\label{fig:Gluecase1c}
\end{figure}
\noindent Renaming the summation variables ($u_{1}+v_{2}+v_{3}\to v_{1}$, $u_{2}\to v_{2}$) and using \eqref{cwest2} we find that the inductive bound is reproduced, provided that \eqref{127bound} holds true.

This concludes the proof of theorem~\ref{thmCAGins1}.
\end{proof}
%\vspace{.5cm}
\noindent We can again simplify the bound established in theorem \ref{thmCAGins1} into a version without $\la_{i}$-integrals:
\begin{cor}\label{cor2}
For  $N,L,w$ as in theorem \ref{thmCAGins1} and non-exceptional momenta $\vec{p}\in\mathbb{R}^{4N}$ (i.e. for $\bar\eta(\vec{p})>0$), there exists $K>0$ such that
\ben\label{corbound2}
\begin{split}
|\pa^w {\cal L}^{0,\La_{0}}_{N,L}(\O_{A}(0),\vec{p};\LIR)| 
&\le  {|w|!} \sqrt{|w'|!} \ 
K^{(2N+8L-3)(|w|+1)}K^{[A](N/2+2L)^3}\, M^{[A]} \  
\inf(\bar\eta(\vec{p}),M)^{-(N+|w|)}\\
&\times\sup(1, \frac{|\vec{p}|}{M})^{[A](N+2L+1)}\mathcal{P}_{2L+\frac{N}{2}-1}\left( \log_{+}\frac{|\vec{p}|_\LIR}{\inf(\bar\eta(\vec{p}),M)}\right) \, ,
\end{split}
\een
 where $\bar\eta$  is defined in \eqref{etabdef}.
\end{cor}
\begin{proof}
The proof is largely analogous to the proof of corollary \ref{cor1}. Starting from the bound stated in theorem \ref{thmCAGins1}, we first rewrite the $\la_{i}$-integrals in the form \eqref{intprodorder}. This allows us to use lemma \ref{lemmalambdalarge} repeatedly to bound all but one of these integrals, and to finally use lemma \ref{lemmaLambdaint} to bound the final integral.  Absorbing $N$ and $L$ dependent factors in the polynomial, we thereby arrive at the bound
\ben
\begin{split}
&|\pa^w {\cal L}^{0,\Lao}_{N,L}(\O_{A}(0);\vec{p};\LIR)| 
\le  \sqrt{|w|!\, |w'|!}\ K^{(2N+8L-4)(|w|+1)}K^{[A](N/2+2L)^3 } \sum_{v_{1}+v_{2}=w} \sum_{T \in \tkn_{N,2L,v_{1}}(\vec{p})}
  \\
&\times \LIR^{[A]-N_{\mathcal{V}}-|v_{2}|}\,  \sum_{\mu=0}^{[A](N+2L+1)}\frac{1}{\sqrt{\mu!}}\
\left(\frac{|\vec p_{\mathcal{V}}|}{\LIR}\right)^{\mu}   \bar\eta(\vec{p})^{-(N-N_{\mathcal{V}}+|v_{1}|)}\cdot \sqrt{|v_{1}|![2\alpha(N,L)]^{|v_{1}|}}
\\
&\times\,
 \mathcal{P}_{2L+\frac{N}{2}-1}\left( \log_{+}\frac{|\vec{p}|_\LIR}{\inf(\bar\eta(\vec{p}), M)}\right)\, ,
\end{split}
\een
where we used the fact that the weights of the tress $T \in \tkn_{N,2L,v_{1}}$ satisfy $\sum_{i}\theta_{i}=N-N_{\mathcal{V}}  +|v_{1}|$. We also made use of the inequality\footnote{In the case where $N_{\mathcal{V}}=1$, an internal line may be assigned the sum of \emph{all} momenta $p_{1}+\ldots+p_{N}$, which is not possible in the case without insertions. This explains why $\bar\eta$ appears in the bound \eqref{etabound2} whereas $\eta$ appears in~\eqref{etabound}, and it also explains why the notion of exceptional momenta is slightly different for the cases with and without insertions.}
\ben\label{etabound2}
\inf_{r\in\mathcal{I}(T)}|k_{r}|\geq \bar\eta(\vec{p}) \quad \text{ for } T\in \tkn_{N,2L,v_{1}}(\vec{p})\, .
\een
Note that $N-N_{\mathcal{V}}+|v_{1}|$ is always positive, which allows us to use the inequality 
\ben
\bar\eta(\vec{p})^{-(N-N_{\mathcal{V}}+|v_{1}|)}\leq \inf(M,\bar\eta(\vec{p}))^{-(N-N_{\mathcal{V}}+|v_{1}|)}\, .
\een
To bound the cardinality of the set $\tkn_{N,2L,v_{1}}$ we make use of lemma \ref{lemtreebd2}, which tells us that we pick up a factor $ 4^{3N} \cdot [3(|w|+1)]^{\frac{N-2}{2}+L}$.  Absorbing  $N$ and $L$ dependent factors into the polynomial and using  \eqref{cwest2}, 
we verify the bound claimed in the corollary for $K$  large enough such that (here $K_{0}$ is the constant from theorem \ref{thmCAGins1})
\ben
2^{3|w|/2}(|w|+1)^{\frac{N-2}{2}+L}\, [\alpha(N,L)]^{|w|/2}  K_{0}^{(N+4L-4)(|w|+1)} \leq K^{(N+4L-3)(|w|+1)}\, .
\een 
\end{proof}
\noindent Theorem \ref{thmCAGins1} implies the following bound for the smeared, connected Schwinger functions with one insertion:
\begin{cor}\label{cor2b}
For any $N,L\in\mathbb{N}$ there exists $K>0$ such that
\ben\label{cor2beq}
\begin{split}
&\left|\int_{\vec{p}}  \L_{N,L}^{0,\infty}\left( \O_{A}(0);\vec{p}\right)\, \prod_{i=1}^{N}  \frac{\hat{\test}_{i}(p_{i})}{p_i^2} \right| \leq  \sqrt{[A]!}\ K^{[A]} \, 
  \LIR^{[A]+N}  \sum_{\mu=0}^{([A]+2)(N+2L+1)}\hspace{-.3cm} \sum_{ \mu_1+\ldots+\mu_N=\mu   }    \frac{\prod\limits_{i=1}^{N}\|\hat{\test}_{i}\|_{\frac{\mu_{i}}{2}}} { M^{\mu}}   \, ,
\end{split}
\een
for arbitrary test functions $\test_i\in\mathcal{S}(\mathbb{R}^4)$ and with $\|\cdot\|_n$ as in \eqref{Schwnorms}.
\end{cor}
The proof of this corollary follows essentially the same arguments as that of corollary \ref{cor3b}, which can be found in appendix \ref{appdistr}.
\subsection{Schwinger functions with two insertions}\label{sec:2ins}
We next derive bounds on (subtracted)
 CAS's with two insertions using our previous results on the CAS's without and with one insertion. 
The CAS's with insertion of two composite operators $\O_{A}=\{N',w'\}$ and $\O_{B}=\{N'',w''\}$ and with subtraction to degree $D\leq[A]+[B]$ were defined in section~\ref{subsec:compfields} via the flow equation 
\ben\label{CAGFE2Insertion}
\begin{split}
&\partial_{\La}\L^{\La,\Lambda_{0}}_{D,N,L}(\O_{A}(x)\otimes\O_{B}(0);\vec{p};\LIR)= \left(N+2 \atop 2 \right) \, \int_{\ell} \dot{C}^{\La}(\ell)\L^{\La,\Lambda_{0}}_{D,N+2,L-1}(\O_{A}(x)\otimes\O_{B}(0); \ell, -\ell,  \vec{p};\LIR)  \\
&-\hspace{-.2cm}\sum_{{l_{1}+l_{2}=L} \atop{n_{1}+n_{2}=N+2}}\hspace{-.2cm} n_{1}n_{2}\,  \mathbb{S}\, \bigg[  \L^{\La,\Lambda_{0}}_{D,n_{1},l_{1}}(\O_{A}(x)\otimes\O_{B}(0);   p_{1},\ldots,p_{n_{1}-1},q;\LIR)\,  \dot{C}^{\La}(q)\,   \L^{\La,\Lambda_{0}}_{n_{2},l_{2}}(-q,  p_{n_{1}},\ldots,p_{N};\LIR) \\
&\qquad+\int_{\ell} \L^{\La,\Lambda_{0}}_{n_{1},l_{1}}(\O_{A}(x);   p_{1},\ldots,p_{n_{1}-1}, \ell;\LIR)\,  \dot{C}^{\La}(\ell)\,   \L^{\La,\Lambda_{0}}_{n_{2},l_{2}}(\O_{B}(0);-\ell,  p_{n_{1}},\ldots,p_{N};\LIR) \bigg]
\end{split}
\een
and the boundary conditions (see \eqref{BCL2ins1A} and \eqref{BCL2ins2A})
\ben\label{BCLins21}
\partial^{w}_{\vec{p}}\L^{\LIR, \La_{0}}_{D,N,L}(\O_{A}(x)\otimes\O_{B}(0); \vec{0}; \LIR)= 0 \quad \text{ for }N+|w|\leq D
\een
\ben\label{BCLins22}
\partial^{w}_{\vec{p}}\L^{\La_{0},\La_{0}}_{D,N,L}(\O_{A}(x)\otimes\O_{B}(0); \vec{p};\LIR)=0\quad \text{ for }N+|w|>D \quad .
\een
For simplicity, we again assume that both $N'$ and $N''$ are even. One can show: 
\begin{thm}\label{thm3}
For $M\geq \La \geq 0\,$, for any $N,L\in\mathbb{N}, w\in\mathbb{N}^{4N}$ and for $[A]+[B]\geq D\geq -1$ there exists a constant $\,K>0\,$ such that the bound
\ben\label{CAG2boundeq}
\begin{split}
&|\pa^w {\cal L}^{\La,\Lao}_{D,N,L}(\O_{A}(x)\otimes\O_{B}(0);\vec{p};\LIR)| \\
&\le  \sqrt{(|w|+[A]+[B]-D)!\, |w'|!\,|w''|!}\ K^{(2N+8L-3)(|w|+1)}K^{([A]+[B])(N/2+2L)^3 } \sum_{v_{1}+v_{2}=w}\, \sum_{T \in \tkn_{N,2L+2,v_{1}}(\vec{p})}
  \\
&\times\LIR^{D-N_{\mathcal{V}}} \cdot \sup\left(M, \frac{1}{|x|} \right)^{[A]+[B]-D}\sup\left(|x|,\frac{1}{M}\right)^{|v_{2}|} \cdot  \sum_{\mu=0}^{([A]+[B])(N+2L+3)}\frac{1}{\sqrt{\mu!}}\
\left(\frac{|\vec p_{\mathcal{V}}|}{\LIR}\right)^{\mu}    \\
&\times \prod_{i \in \mathcal{I}(T)} \int_{\La}^{\La_{0}} \d\la_{i}\ f_{\lambda_{i}}(k_i;N-2,L+1,\theta_i)  \,
 \mathcal{P}_{2L+\frac{N}{2}}\left( \log_{+}\sup\left(\frac{|\vec{p}|}{\LIR},\frac{|\vec{p}|_\LIR}{\min_{j\in\mathcal{I}}\lambda_{j}},
 \frac{\max_{j\in\mathcal{I}}\lambda_{j}}{\LIR}\right)\right)
\end{split}
\een
holds, with the same notation as in theorem \ref{thmCAGins1}.
\end{thm}
\begin{proof}
The l.h.s. of \eqref{CAG2boundeq} vanishes for odd $N$ due to the $\varphi\to -\varphi$ symmetry.  For even $N$, our strategy is again to prove this bound in two steps:
\begin{enumerate}
\item Derive a bound for $|\pa^w {\cal L}^{M,\Lao}_{D,N,L}(\O_{A}(x)\otimes\O_{B}(0);\vec{p};\LIR)|$ using the flow equation \eqref{CAGFE2Insertion} and boundary conditions  \eqref{BCLins21}, \eqref{BCLins22}.
\item Derive a bound for $|\pa^w {\cal L}^{\La,\Lao}_{D,N,L}(\O_{A}(x)\otimes\O_{B}(0);\vec{p};\LIR)|$ with $0\leq\La\leq M$, employing again the flow equation \eqref{CAGFE2Insertion}, but using the bounds derived in step 1 as boundary conditions at $\La=M$.
\end{enumerate}
To implement the first step, we appeal, as in the proof of theorem \ref{thmCAGins1}, to the results of~\cite{Hollands:2011gf,Holland:2012vw} (see in particular corollary 1 in \cite{Holland:2012vw}). The bounds given there imply in the limit $m\to 0$ for finite infrared cutoff $\La = M$ (the subtlety discussed in footnote \ref{foot} applies here as well):
\ben\label{boundCAG2}
\begin{split}
&|\partial^{w}_{\vec{p}}\L^{\LIR,\Lambda_{0}}_{ D,N,L}(\O_{A}(x)\otimes\O_{B}(0); \vec{p}; \LIR)|\\
&\quad\leq  K^{(2N+8L-3)|w|}\, K^{([A]+[B])(N/2+2L)^{3}}\, \sqrt{(|w|+[A]+[B]-D)!\, |w'|!\, |w''|!}  \\
&\quad\times \frac{\LIR^{D-N-|w|}}{|x|^{[A]+[B]-D}}\cdot  \sum_{\mu=0}^{([A]+[B])(N+2L+3)}\frac{1}{\sqrt{\mu!}}\left(\frac{|\vec{p}|}{\LIR}\right)^{\mu}\,  \mathcal{P}_{2L+\frac{N}{2}}\left( \log_{+}\frac{|\vec{p}|}{\LIR}\right) \, .
\end{split}
\een
Note that this bound is consistent with our hypothesis \eqref{CAG2boundeq}. It corresponds to the case $N_{\mathcal{V}}=N$, i.e. all external legs of the tree are directly attached to the special vertex $\mathcal{V}$. This guarantees that boundary contributions, which appear when integrating the flow equation, satisfy the bound  \eqref{CAG2boundeq} (cf. section \ref{subsec:BC}).
It remains to carry out step 2, i.e. integrating the flow equation between $\La$ and $M$. Let us start the induction.
\subsubsection{First and second term on the r.h.s. of the flow equation \eqref{CAGFE2Insertion}}
To verify that these contributions are consistent with our induction hypothesis, \eqref{CAG2boundeq}, we follow essentially the same steps as in the proof of our bound for the CAS's with one insertion (cf. section \ref{sec:1ins}). Since the necessary adjustments are minor, we refrain from repeating the lengthy calculations.
\subsubsection{Third term on the r.h.s. of the flow equation \eqref{CAGFE2Insertion}}
To begin with, we make use of the translation properties of the CAS's with one insertion in order to write
\ben\label{partint}
\begin{split}
&\Big|\partial_{\vec{p}}^{w}\int_{\La}^{M}\d\la\int_{\ell}\hspace{-.2cm} \sum_{\substack{ l_{1}+l_{2}=L \\ n_{1}+n_{2}=N+2  }}\hspace{-.5cm} n_{1}n_{2} \L^{\la,\Lambda_{0}}_{n_{1},l_{1}}(\O_{A}(x);   p_{1},\ldots,p_{n_{1}-1}, \ell;\LIR)\,  \dot{C}^{\la}(\ell)\,   \L^{\la,\Lambda_{0}}_{n_{2},l_{2}}(\O_{B}(0);-\ell,  p_{n_{1}},\ldots,p_{N};\LIR)\Big|\\
&=\Big|\partial_{\vec{p}}^{w}\int_{\La}^{M}\d\la\int_{\ell} \sum_{\substack{ l_{1}+l_{2}=L \\ n_{1}+n_{2}=N+2  }}n_{1}n_{2}\, e^{i (p_{1}+\ldots+p_{2n_{1}-1}+\ell)x}\\
&\qquad\qquad\times \L^{\la,\Lambda_{0}}_{n_{1},l_{1}}(\O_{A}(0);   \vec{p}_{1}, \ell;\LIR)\,  \dot{C}^{\la}(\ell)\,   \L^{\la,\Lambda_{0}}_{n_{2},l_{2}}(\O_{B}(0);-\ell, \vec{p}_{2};\LIR)\Big|\\
&\leq\sum_{v_{1}+v_{2}=w}\sum_{\substack{ l_{1}+l_{2}=L \\ n_{1}+n_{2}=N+2  }}n_{1}n_{2}\, c_{\{v_{i}\}} |x|^{|v_{2}|} \int_{\La}^{M}\d\la\int_{\ell}\\ %e^{i (p_{1}+\ldots+p_{n_{1}-1}+\ell)x}\, \\
&\qquad\qquad\times \Big| [\prod_{i=1}^{n_{1}-1}\partial_{p_{i}}^{(v_{1})_{i}} \L^{\la,\Lambda_{0}}_{n_{1},l_{1}}(\O_{A}(0);  \vec{p}_{1}, \ell;\LIR)]\,  \dot{C}^{\la}(\ell)\,  [\prod_{i=n_{1}}^{N}\partial_{p_{i}}^{(v_{1})_{i}}  \L^{\la,\Lambda_{0}}_{n_{2},l_{2}}(\O_{B}(0);-\ell, \vec{p}_{2};\LIR)]  \Big|\, .
\end{split}
\een
In order to bound this ``source term'', we make use of our bounds on the CAS's with one insertion derived above in section \ref{sec:1ins}. Using also \eqref{cwest} for $c_{\{v_{i}\}}$, this yields
\ben\label{CAG23rd}
\begin{split}
|\text{r.h.s. of \eqref{partint}}| &\leq \sum_{\substack{ l_{1}+l_{2}=L \\ n_{1}+n_{2}=N+2  }}  \sum_{\substack{(v_{1},v_{2})+v_{3}=w \\ v_{1}\in\mathbb{N}^{4(n_{1}-1)} \\ v_{2}\in\mathbb{N}^{4(n_{2}-1)}  }}n_{1}n_{2}\,2^{|w|} |x|^{|v_{3}|} \int_{\La}^{M}\d\la \int_{\ell}  \frac{2  e^{-\frac{\ell^{2}}{\la^{2}}}}{\la^{3}}\\
&\times \bigg[\sqrt{|v_{1}|!\, |w'|!}\ K_{1}^{(2n_{1}+8l_{1}-4)(|v_{1}|+1)}K_{1}^{[A](n_{1}/2+2l_{1})^3 } \sum_{u_{1}+u_{2}=v_{1}} \sum_{T_{1} \in \tkn_{n_{1},2l_{1},(0,u_{1})}(\ell, \vec{p}_{1})}
  \\
&\times\LIR^{[A]-N_{\mathcal{V}_{1}}-|u_{2}|}  \sum_{\mu=0}^{[A](n_{1}+2l_{1}+1)}\frac{1}{\sqrt{\mu!}}\
\left(\frac{|\vec p_{\mathcal{V}_{1}}(\ell)|}{\LIR}\right)^{\mu}    \prod_{i \in \mathcal{I}(T_{1})} \int_{\la}^{\La_{0}} \d\la_{i}\ f_{\lambda_{i}}(k_i(\ell);n_{1},l_{1},\theta_i) 
\\
&\times\,
 \mathcal{P}_{2l_{1}+\frac{n_{1}}{2}-1}\left( \log_{+}\sup\left(\frac{|(\vec{p}_{1},\ell)|}{\LIR},\frac{|(\vec{p}_1,\ell)|_\LIR}{\min\limits_{j\in\mathcal{I}(T_{1})}\lambda_{j}},
 \frac{\max\limits_{j\in\mathcal{I}(T_{1})}\lambda_{j}}{\LIR}\right)\right) \bigg]\\
 &\times \bigg[\sqrt{|v_{2}|!\, |w''|!}\ K_{1}^{(2n_{2}+8l_{2}-4)(|v_{2}|+1)}K_{1}^{[B](n_{2}/2+2l_{2})^3 } \sum_{u_{3}+u_{4}=v_{2}} \sum_{T_{2} \in \tkn_{n_{2},2l_{2},(0,u_{3})}(-\ell, \vec{p}_{2})}
  \\
&\times\LIR^{[B]-N_{\mathcal{V}_{2}}-|u_{4}|}  \sum_{\mu=0}^{[B](n_{2}+2l_{2}+1)}\frac{1}{\sqrt{\mu!}}\
\left(\frac{|\vec p_{\mathcal{V}_{2}}(\ell)|}{\LIR}\right)^{\mu}    \prod_{i \in \mathcal{I}(T_{2})} \int_{\la}^{\La_{0}} \d\la_{i}\ f_{\lambda_{i}}(k_i(\ell);n_{2},l_{2},\theta_i) 
\\
&\times\,
 \mathcal{P}_{2l_{2}+\frac{n_{2}}{2}-1}\left( \log_{+}\sup\left(\frac{|(\vec{p}_{2},\ell)|}{\LIR},\frac{|(\vec{p}_2,\ell)|_\LIR}{\min\limits_{j\in\mathcal{I}(T_{2})}\lambda_{j}},
 \frac{\max\limits_{j\in\mathcal{I}(T_{2})}\lambda_{j}}{\LIR}\right)\right) \bigg] \, .
\end{split}
\een
Here we denote by $K_{1}$ the constant (called $K$ there) appearing in our bound for the CAS's with one insertion, \eqref{CAG1boundeq}. Since $2\leq n_{1},n_{2}\leq N$, one has
 \ben
 [A](n_{1}/2+2l_{1})^{3}\cdot [B](n_{2}/2+2l_{2})^{3}\leq ([A]+[B]) (N/2+2L)^{3}
 \een
and also
 \ben\label{musummerge}
 [A](n_{1}+2l_{1}+1)+[B](n_{2}+2l_{2}+1)\leq ([A]+[B])(N+2L+1)\, .
 \een
 The relation \eqref{musummerge} allows us to combine the sums over $\mu_{i}$:
 \ben
 \begin{split}
&\left[ \sum_{\mu_{1}=0}^{[A](n_{1}+2l_{1}+1)}\frac{1}{\sqrt{\mu_{1}!}}\
\left(\frac{|\vec p_{\mathcal{V}_{1}}|}{\LIR}\right)^{\mu_{1}}\right]\times \left[ \sum_{\mu_{2}=0}^{[B](n_{2}+2l_{2}+1)}\frac{1}{\sqrt{\mu_{2}!}}\
\left(\frac{|\vec p_{\mathcal{V}_{2}}|}{\LIR}\right)^{\mu_{2}}\right]\\
& \leq \sqrt{2}^{([A]+[B])(N+2L+1)}\sum_{\mu=0}^{([A]+[B])(N+2L+1)}\frac{1}{\sqrt{\mu!}}\
\left(\frac{\sup(|\vec p_{\mathcal{V}_{1}}|,|\vec p_{\mathcal{V}_{2}}|)}{\LIR}\right)^{\mu}\, .
\end{split}
 \een
 For the weight functions, we use again that
%
%\ben
$  f_{\lambda_{i}}(k_i;n_{a},l_{a},\theta_i)  \leq  f_{\lambda_{i}}(k_i;N,L,\theta_i) $ .
%\een

 The loop integral over $\ell$ can again be bounded with the help of lemma \ref{lablemma1}\footnote{Note that $|\mathcal{I}(T_{1})|+|\mathcal{I}(T_{2})|\leq \frac{n_{1}+n_{2}-4}{2}+l_{1}+l_{2}= \frac{N-2}{2}+L$, so the conditions in the lemma are met.}, which yields a factor $2^{([A]+[B])(N+2L+1)}$ along with some purely $N,L$ dependent terms. Combining the two logarithmic polynomials and renaming summation indices $u_{1}\to v_{1}, (u_2,u_4) \to v_3, u_3 \to v_2,
v_3 \to v_4 $, we may infer
\ben\label{CAG23rdb}
\begin{split}
&|\text{r.h.s. of \eqref{CAG23rd}}|\leq 4^{([A]+[B])(N+2L+1)}\, 2^{|w|} \sum_{\substack{ l_{1}+l_{2}=L \\ n_{1}+n_{2}=N+2  }} n_{1}n_{2}  \\
&\times  \,   \sqrt{|w|!\, |w'|!|w''|!}\quad  K_{1}^{(2N+8L-4)(|w|+1)}K_{1}^{([A]+[B])(N/2+2L)^3 }\sum_{(v_{1},v_{2})+v_{3}+v_{4}=w}   \\
&\times \sum_{ {T_{1} \in \tkn_{n_{1},2l_{1},(0,v_{1})}(0,\vec{p}_{1})} \atop{T_{2} \in \tkn_{n_{2},2l_{2},(0,v_{2})}(0,\vec{p}_{2})}} M^{[A]+[B]-N_{\mathcal{V}_{1}}-N_{\mathcal{V}_{2}}-|v_{3}|} \, |x|^{|v_{4}|} \,
 \sum_{\mu=0}^{([A]+[B])(N+2L+1)}  \frac{1}{\sqrt{\mu!}}\
\left(\frac{\sup(|\vec p_{\mathcal{V}_{1}}|,|\vec{p}_{\mathcal{V}_{2}}|)}{\LIR}\right)^{\mu} \\
&\times  \int_{\La}^{M}\d\la\, \la  \prod_{i \in\mathcal{I}(T_{1})} \int_{\la}^{\La_{0}} \d\la_{i}\ f_{\lambda_{i}}(k_i;N-2,L+1,\theta_i) \prod_{j \in\mathcal{I}(T_{2})} \int_{\la}^{\La_{0}} \d\la_{j}\ f_{\lambda_{i}}(k_j;N-2,L+1,\theta_j) \\
&\times \mathcal{P}_{2L+\frac{N}{2}}\left( \log_{+}\sup\left(\frac{|\vec{p}|}{\LIR},\frac{|\vec{p}|_\LIR}{\min\limits_{j\in\mathcal{I}(T_{1})\cup \mathcal{I}(T_{2})}\lambda_{j}},
 \frac{\max\limits_{j\in\mathcal{I}(T_{1})\cup\mathcal{I}(T_{2})}\lambda_{j}}{\LIR}\right)\right)\, .
\end{split}
\een
The $\la$-integral is bounded as in \eqref{lambdaintsest}, leading to a factor of $\min_{j}(\la_{j},M)^{2}$. This additional factor can again be absorbed  by our reduction operation on trees (cf. pages \pageref{Lambdaint} 
and \pageref{CAGins1st1case2b}). As we have explained on page \pageref{CAGins1st1case2b}, this procedure ``eats'' a factor  $\min_{j}(\la_{j},M)$ and maps a tree $T \in \tkn_{N+1,2L,(0,w)}(0,\vec{p})$ to a tree $T' \in \tkn_{N,2L+1,w}(\vec{p})$. Applying this procedure twice to remove the external lines corresponding to the loop momenta $\ell,-\ell$ (recall also that we pick up the additional factor \eqref{reductionbound} due to the fact that the reduction is not-injective), we find 
\ben\label{CAG23rdc}
\begin{split}
&|\text{r.h.s. of \eqref{CAG23rdb}}| \leq 2^{|w|}4^{([A]+[B])(N+2L+1)}(|w|+1)^{2} \sum_{\substack{ l_{1}+l_{2}=L \\ n_{1}+n_{2}=N+2  }}n_{1}n_{2}    \\
&\times \sqrt{|w|!\, |w'|!|w''|!}\quad  K_{1}^{(2N+8L-4)(|w|+1)}K_{1}^{([A]+[B])(N/2+2L)^3 } \sum_{(v_{1},v_{2})+v_{3}+v_{4}=w}  \\
&\times \sum_{ {T_{1} \in \tkn_{n_{1}-1,2l_{1}+1,v_{1}}(\vec{p}_{1})} \atop{T_{2} \in \tkn_{n_{2}-1,2l_{2}+1,v_{2}}(\vec{p}_{2})}}   M^{[A]+[B]-N_{\mathcal{V}_{1}}-N_{\mathcal{V}_{2}}-|v_{3}|}\, |x|^{|v_{4}|} \,
 \sum_{\mu=0}^{([A]+[B])(N+2L+1)}  \frac{1}{\sqrt{\mu!}}\
\left(\frac{\sup(|\vec p_{\mathcal{V}_{1}}|,|\vec{p}_{\mathcal{V}_{2}}|)}{\LIR}\right)^{\mu} \\
&\times  \prod_{i \in\mathcal{I}(T_{1})\cup \mathcal{I}(T_{2})} \int_{\La}^{\La_{0}} \d\la_{i}\ f_{\lambda_{i}}(k_i;N-2,L+1,\theta_i) \\
&\times \mathcal{P}_{2L+\frac{N}{2}}\left( \log_{+}\sup\left(\frac{|\vec{p}|}{\LIR},\frac{|\vec{p}|_\LIR}{\min\limits_{j\in\mathcal{I}(T_{1})\cup \mathcal{I}(T_{2})}\lambda_{j}},
 \frac{\max\limits_{j\in\mathcal{I}(T_{1})\cup\mathcal{I}(T_{2})}\lambda_{j}}{\LIR}\right)\right)\, .
\end{split}
\een
Finally, to express the bound in terms of trees $T \in \tkn_{N,2L+2,(v_{1},v_{2})}(\vec{p})$, we merge the two trees ${T_{1} \in \tkn_{n_{1}-1,2l_{1}+1,v_{1}}(\vec{p}_{1})}$ and ${T_{2} \in \tkn_{n_{2}-1,2l_{2}+1,v_{2}}(\vec{p}_{2})}$ by combining the vertices $\mathcal{V}_{1}$ and $\mathcal{V}_{2}$ into one vertex $\mathcal{V}$ of coordination number $N_{\mathcal{V}}=N_{\mathcal{V}_{1}}+N_{\mathcal{V}_{2}}$, see fig.\ref{fig:2insmerge}. 
\begin{figure}[htbp]
\begin{center}
\includegraphics{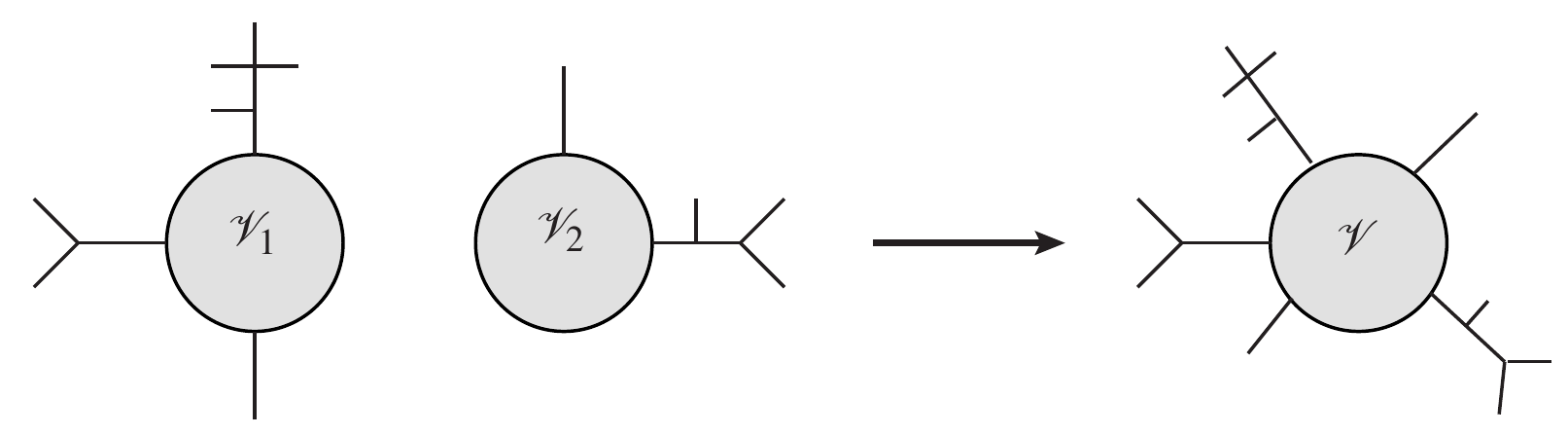}
\end{center}
\caption{We merge two trees ${T_{1} \in \tkn_{n_{1}-1,2l_{1}+1,v_{1}}(\vec{p}_{1})}$ and ${T_{2} \in \tkn_{n_{2}-1,2l_{2}+1,v_{2}}(\vec{p}_{2})}$ by combining the vertices $\mathcal{V}_{1}$ and $\mathcal{V}_{2}$ into one vertex $\mathcal{V}$}
\label{fig:2insmerge}
\end{figure}

%Recalling also that $n_{1}+n_{2}=N+2$ and that $l_{1}+l_{2}=L$, n
Noting that $|\vec{p}_{\mathcal{V}}|\geq \sup(|\vec{p}_{\mathcal{V}_{1}}|,|\vec{p}_{\mathcal{V}_{2}}|)$ and renaming  multi-indices $(v_1,v_2) \to v_1, 
v_3 + v_4 \to v_2$, we arrive at the bound 
\ben\label{CAG23rdfin}
\begin{split}
&|\text{r.h.s. of \eqref{CAG23rdc}}|\\
 &\leq 4^{|w|}4^{([A]+[B])(N+2L+1)}   \sqrt{|w|!\, |w'|!|w''|!}\quad K_{1}^{(2N+8L-4)(|w|+1)}K_{1}^{([A]+[B])(N/2+2L)^3 }  \sum_{v_{1}+v_{2}=w}  \\
&\times \sum_{T \in \tkn_{N,2L+2,v_{1}}(\vec{p})}  M^{[A]+[B]-N_{\mathcal{V}}} \sup(|x|,\frac{1}{M})^{|v_{2}|} \quad
 \sum_{\mu=0}^{([A]+[B])(N+2L+1)}  \frac{1}{\sqrt{\mu!}}\
\left(\frac{|\vec p_{\mathcal{V}}|}{\LIR}\right)^{\mu} \\
&\times  \prod_{i \in\mathcal{I}(T)} \int_{\La}^{\La_{0}} \d\la_{i}\ f_{\lambda_{i}}(k_i;N-2,L+1,\theta_i)  \mathcal{P}_{2L+\frac{N}{2}}\left( \log_{+}\sup\left(\frac{|\vec{p}|}{\LIR},\frac{|\vec{p}|_\LIR}{\min_{j\in\mathcal{I}}\lambda_{j}},
 \frac{\max_{j\in\mathcal{I}}\lambda_{j}}{\LIR}\right)\right)\, .
\end{split}
\een
Here we picked up another factor of $2^{|w|}$ from the redefinition $v_{3}+v_{4}\to v_{2}$ 
%(the redefinition $w_{1}+w_{2}\to w_{1}$ does not yield an additional factor, since different values of $w_{1},w_{2}$ lead to different trees in the ''merging'' procedure)
. We also bounded the sums over $l_{1},l_{2}$ and $n_{1},n_{2}$ by $L+1$ and $N^{3}$, respectively, and absorbed these factors into the logarithmic polynomial. The bound \eqref{CAG23rdfin} is consistent with our induction hypothesis provided that the constant $K$ is chosen large enough that
\ben
K_{1}^{(2N+8L-4)(|w|+1)+([A]+[B])(N/2+2L)^3 }  4^{|w|+([A]+[B])(N+2L+1)} \leq K^{(2N+8L-3)(|w|+1)}K^{([A]+[B])(N/2+2L)^3 }\, .
\een
This establishes theorem \ref{thm3}.
\end{proof}

%\vspace{.5cm}

%\newpage 
\noindent We can again establish a bound where the $\la_{i}$-integrals do not appear anymore:
\begin{cor}\label{cor3}
For any $N,L\in\mathbb{N}$ and  $\bar\eta(\vec{p})>0$ there exists $K>0$ such that
\ben\label{corbound3}
\begin{split}
&| {\cal L}^{0,\La_{0}}_{D,N,L}(\O_{A}(x)\otimes \O_{B}(0),\vec{p};\LIR)| \\
&\le   \sqrt{([A]+[B]-D)!\, ([A]+[B])!} \ 
K^{([A]+[B])(N/2+2L)^3} M^{D} \, 
\inf(\bar\eta(\vec{p}),M)^{-N}\, \sup(M,\frac{1}{|x|})^{[A]+[B]-D}\\
&\times \sup(1, \frac{|\vec{p}|}{M})^{([A]+[B])(N+2L+1)}\mathcal{P}_{2L+\frac{N}{2}}\left( \log_{+}\frac{|\vec{p}|_\LIR}{\inf(\bar\eta(\vec{p}) ,M)}\right) \, ,
\end{split}
\een
 where $\bar\eta$ is defined as in \eqref{etabdef}.
\end{cor}
The proof is analogous to that of corollaries \ref{cor1} and \ref{cor2}. 
 One can also derive a bound on the smeared, connected Schwinger functions with two insertions:
\begin{cor}\label{cor3b}
For any $N,L\in\mathbb{N}$ there exists $K>0$ such that
\ben\label{CAGloopintfinal}
\begin{split}
&\left|\int_{p_{1},\ldots,p_{N}}  \L_{D,N,L}^{0,\infty}\left( \O_{A}(x)\otimes\O_{B}(0);\vec{p}\right)\, \prod_{i=1}^{N}  \frac{\hat{\test}_{i}(p_{i})}{p_i^2} \right| \\
&\leq  \sqrt{[A]![B]!}\ K^{[A]+[B] } \, \sup(M,\frac{1}{|x|})^{[A]+[B]-D} 
  \LIR^{D+N}  \sum_{\mu=0}^{([A]+[B]+2)(N+2L+3)}\hspace{-.5cm} \sum_{ \mu_1+\ldots+\mu_N=\mu   }  \frac{ \prod_{i=1}^{N} \|\hat{\test}_{i}\|_{\frac{\mu_{i}}{2}}} { M^{\mu}}   \, ,
\end{split}
\een
for arbitrary test functions $\test_i\in\mathcal{S}(\mathbb{R}^4)$.
\end{cor}
See appendix \ref{appdistr} for the proof of this corollary.
\section{Convergence of the Operator Product Expansion}\label{sec:convergence}
With the bounds on Schwinger functions with operator insertions at our disposal, we are now in a position to prove convergence of the operator product expansion. We would like to insert the OPE into a correlation function with suitable \emph{spectator fields} and estimate the difference between the left- and right hand side of the expansion. The spectator fields play the role of a quantum state in the Euclidean context. In order to have spectator fields with sufficient regularity, 
we average the $i$-th spectator field against a test function, $\test_i \in {\mathscr S}(\mathbb{R}^4), i=1, ..., N$. We write $\varphi(\test):=\int\d^4 x\,  \varphi(x) \test(x)$.
We can then express the remainder of the OPE, truncated at dimension $D\in\mathbb{N}$, as follows:
\ben\label{OPER}
\begin{split}
& \Big| \Big\langle \O_{A}(x)\O_{B}(0)\, \varphi(\test_{1})\cdots\varphi(\test_{{N}}) \Big\rangle- \sum_{C:[C]\leq D}\C_{A B}^{C}(x)\ \Big\langle \O_{C}(0)\, \varphi(\test_{{1}})\cdots\varphi(\test_{{N}}) \Big\rangle  \Big|=\\
 &\sum_{j=1}^{N}\sum_{\substack{I_{1}\cup\ldots\cup I_{j}=\{1,\ldots, N\}\\ I_{i}\cap I_{j}=\emptyset \\ l_{1}+\ldots+l_{j}=L  }} \int_{\vec q}  \mathcal{R}^{\La,\Lao}_{D,|I_{1}|,l_{1}}(\O_{A}(x)\otimes\O_{B}(0); \vec{q}_{I_{1}}) \bar\L^{\La,\Lambda_{0}}_{|I_{2}|, l_{2}}(\vec{q}_{I_{2}})\cdots \bar\L^{\La,\Lambda_{0}}_{|I_{j}|, l_{j}}(\vec{q}_{I_{j}}) \prod_{i=1}^{N}\hat \test_i(q_i)\, C^{\La,\Lambda_{0}}(q_{i})\, ,
\end{split}
\een
where $\bar\L^{\La,\Lambda_{0}}_{N,L}$ are the moments of the generating functional $\bar{L}^{\La,\Lambda_{0}}(\varphi)=-L^{\La,\Lambda_{0}}(\varphi)+\frac{1}{2}\langle \varphi,\, (C^{\La,\Lambda_{0}})^{-1}\star\varphi \rangle$ without the momentum conservation delta functions taken out [recall \eqref{deltaCAG}], and where we defined the remainder functional (i.e. the functional generating $\mathcal{R}^{\La,\Lao}_{D,N,L}$)
\ben\label{Rdef}
\begin{split}
&R_{D}^{\Lambda,\Lambda_{0}}(\O_{A}(x)\otimes\O_{B}(0)):=\\
&L^{\La,\La_{0}}(\O_{A}(x))L^{\La,\La_{0}}(\O_{B}(0))-L^{\La,\Lambda_{0}}\left( \O_{A}(x)\otimes\O_{B}(0)\right)  -\sum_{C:[C]\leq D}\, \C_{AB}^{C}(x)\, L^{\Lambda,\Lambda_{0}}(\O_{C}(0))\, .
\end{split}
\een
Our definition of the OPE coefficients, def.~\ref{defOPE}, implies the following lemma:
\begin{lemma}\label{remainder}
The remainder functionals $R_{D}^{\Lambda,\Lambda_{0}}$ can be written as
\ben\label{RGD}
R_{D}^{\Lambda,\Lambda_{0}}(\O_{A}(x)\otimes\O_{B}(0))= (1-\sum_{j\leq\Delta}\T^{j}_{x})\left[ L^{\La,\La_{0}}(\O_{A}(x))L^{\La,\La_{0}}(\O_{B}(0)) -L_{D}^{\La,\Lambda_{0}}\left( \O_{A}(x)\otimes\O_{B}(0)\right) \right]
\een
with $\Delta=D-([A]+[B])$.
\end{lemma}
A proof of this lemma for the case of massive fields can be found in~\cite{Hollands:2011gf}. The proof within the massless theory studied in the present paper is analogous and will therefore be omitted here. Lemma \ref{remainder} combined with our bounds on the smeared connected Schwinger functions with up to two insertions, corollaries \ref{cor2b} and \ref{cor3b}, allows us to bound the smeared remainder functional:
\begin{lemma}\label{lemmaRsmeared}
For any $N,L\in\mathbb{N}$ there exists $K>0$ such that
\ben\label{Rsmearedbdeq}
\begin{split}
&\Bigg|\int_{\vec{p}}\mathcal{R}_{D,N,L}^{0,\infty}(\O_{A}(x)\otimes\O_{B}(0); \vec{p}) \prod_{i=1}^{N}\hat \test_i(p_i)\, \frac{1}{p_i^2} \Bigg| \\
&\leq \sqrt{\frac{[A]! [B]!}{\Delta!}} (K\, M)^{[A]+[B]}\, ( K\, M\, |x| )^{\Delta}  \
  \LIR^{N}\  \sum_{\mu=0}^{(D+2)(N+2L+3)} \sum_{ \mu_1+\ldots+\mu_N=\mu   }    \frac{\prod_{i=1}^{N}\|\hat{\test}_{i}\|_{\frac{\mu_{i}}{2}}} { M^{\mu}}\, , 
\end{split}
\een
with $\|\cdot\|_n$ as in \eqref{Schwnorms}.
\end{lemma}
\begin{proof}
Writing $ (1-\sum_{j\leq\Delta}\T^{j}_{x})=  (1-\sum_{j\leq\Delta-1}\T^{j}_{x})- \T_{x}^{\Delta}$ and  using the integral formula
\ben
 (1-\sum_{j\leq\Delta-1}\T^{j}_{x}) f(x) = \sum_{|v|=\Delta} \frac{x^{v}}{(\Delta-1)!} \int_{0}^{1}\d\tau\, (1-\tau)^{\Delta-1} \partial^{v}f(\tau x)
\een
 for the remainder of the Taylor expansion, we can express the r.h.s. of \eqref{RGD} as
\ben\label{Rtaylorfinal}
\begin{split}
&R_{D}^{\La,\La_{0}}(\O_{A}(x)\otimes\O_{B}(0))\\
&= \sum_{|v|=\Delta}x^{v}\int_{0}^{1}\d\tau\, \frac{(1-\tau)^{\Delta-1}}{(\Delta-1)!}
\left[  L_{{D}}^{\La,\Lambda_{0}}\left(\partial^{v} \O_{A}(\tau x)\otimes\O_{B}(0)\right)- L^{\La,\La_{0}}(\partial^{v}\O_{A}(\tau x))L^{\La,\La_{0}}(\O_{B}(0))\right]\\
&+ \sum_{|v|=\Delta}\frac{x^{v}}{v!} \left[  L_{{D}}^{\La,\Lambda_{0}}\left(\partial^{v} \O_{A}(0)\otimes\O_{B}(0)\right)- L^{\La,\La_{0}}(\partial^{v}\O_{A}(0))L^{\La,\La_{0}}(\O_{B}(0))\right]\, .
\end{split}
\een
Here we have also made use of the Lowenstein rules \eqref{Low1} and \eqref{Low2} in order to pull the derivatives into the CAS's. Substituting this formula on the l.h.s. of \eqref{Rsmearedbdeq} one can use our bounds on the smeared Schwinger functions with up to two insertions, corollaries \ref{cor2b} and \ref{cor3b}, to verify the lemma.
\end{proof}
Combining this lemma with our bounds on the smeared connected Schwinger functions without insertion, corollary \ref{cor1b}, we finally arrive at the first main result of this paper:
\setcounter{thm}{0}

\newpage

\begin{thm}\label{propschw}
Let $\test_i\in\mathcal{S}(\mathbb{R}^{4}), i=1, ..., N$. For any $N,L\in\mathbb{N}$ there exists $K>0$ such that remainder of the operator product expansion, carried out up to operators of dimension $D=[A]+[B]+\Delta$, at $L$ loops, is bounded by
\ben
\begin{split}
&\Big| \Big\langle \O_{A}(x)\O_{B}(0)\, \varphi(\test_1)\cdots\varphi(\test_N) \Big\rangle - \sum_{C:[C] \leq D}\C_{A B}^{C}(x) \Big\langle \O_{C}(0)\, \varphi(\test_1)\cdots\varphi(\test_N) \Big\rangle  \Big|_{L-\text{loops}} \\
 &\leq  \sqrt{{[A]! [B]!}} \ (KM)^{D}\   \frac{|x|^{\Delta}}{\sqrt{\Delta!}}  \
  \LIR^{N}  \sum_{\mu=0}^{(D+2)(N+2L+3)} \sum_{ \mu_1+\ldots+\mu_N=\mu   }    \frac{\prod_{i=1}^{N}\|\hat{\test}_{i}\|_{\frac{\mu_{i}}{2}}} { M^{\mu}}   %\mathcal{P}_{2L+\frac{N}{2}}\left( \log_{+} M \right)  
  \, .
\end{split} 
\label{opeschwartz}
\een
%with $\|F\|_n$ as in \eqref{Schwnorms}.
\end{thm}
%
%\vspace{.5cm}
\noindent  To derive our second main result, theorem \ref{OPEbound1}, let us now consider test functions which are localised in momentum space away from exceptional configurations and which have bounded total momentum. 
Thus, we introduce as a condition on the collection of test functions, $\{\test_i(x)\}$, that
\ben\label{condition}
\prod_{i=1}^N \hat \test_i(p_i) \neq 0 \quad \Leftrightarrow \quad \bar\eta(\vec p) \ge \epsilon \ \ \ \ 
\text{and} \ \ \ \ |\vec p| \le P \ . 
\een
Here $P$ characterises, broadly speaking, the maximum momentum admitted and $\epsilon$ the minimal distance to the set of exceptional momenta admitted in the support of $\hat{\test}_{1}(p_{1})\ldots \hat{\test}_{N}(p_{N})$.  Under these conditions, we can simply bound \eqref{OPER} by taking the supremum of the integrand for all $\vec{p}$ in the support of the test functions, and bounding the momentum integrals simply by $P^{4N}$. Combining lemma \ref{remainder} with corollaries \ref{cor1},\ref{cor2} and \ref{cor3} to bound the CAS's with up to two insertions, one thereby immediately arrives at
\begin{thm}\label{OPEbound}
Suppose the smearing functions $F_i, i=1, ..., N$, satisfy \eqref{condition}.
Then the remainder of the operator product expansion, carried out up to operators of dimension $D=[A]+[B]+\Delta$, at $L$ loops, is bounded by
\ben
\begin{split}
&\Big| \Big\langle \O_{A}(x)\O_{B}(0)\, \varphi(\test_1)\cdots\varphi(\test_N) \Big\rangle- \sum_{C:[C] \leq D}\C_{A B}^{C}(x) \Big\langle \O_{C}(0)\, \varphi(\test_1)\cdots\varphi(\test_N) \Big\rangle  \Big|_{L-\text{loops}} \\
 &\leq P^{N} \sqrt{{[A]! [B]!}} \left(K\, M\,  \sup(1,\frac{P}{M})^{(N+2L+1)}\right)^{[A]+[B]}\, \left(\frac{P}{\inf(M, \epsilon)}\right)^{3N} \prod_{i}\sup|\hat{\test}_i| \\
&\times \frac{1}{\sqrt{\Delta!}} \left(K\, M\,|x|\,   \sup(1,\frac{P}{M})^{(N+2L+1)}\right)^{\Delta}\mathcal{P}_{2L+\frac{N}{2}}\left( \log_{+}\frac{P}{\inf(M,\epsilon)}\right) \, ,
\end{split} 
\label{ope3}
\een
where ${K}$ is a constant depending on $N,L$, and $\mathcal{P}_n$ is a polynomial of degree $n$
whose coefficients depend on $N,L$.
\end{thm}
\vspace{1cm}
{\bf Acknowledgements:} The work of S.H. and J.H. was supported by ERC starting grant QC\& C 259562. S.H. is grateful to Ecole Polytechnique for its hospitality and for providing financial support related to his visits. We thank Riccardo Guida for making available to us, and allowing us to use, results from the unpublished manuscript~\cite{GK}, as well as for perceptive comments.
\appendix
\section{Useful formulas}
In this appendix we collect a number of explicit bounds which are used in our inductive proofs.  The first one is an inequality which allows us to bound the ``loop-integrals'' in the induction:
\begin{lemma}
\label{lablemma1}
Let $N,N',L,d\in\mathbb{N}$ and $\la_{1},\ldots,\la_{N'},\la\in\mathbb{R}_{+}$. For  $0\leq N' \le \frac{N}{2}+L$, $ L\ge 1$,  $N \ge 0$, for $\alpha$ as in eq.\eqref{alphadef} and for $\la_{i}\geq \la$,  we have
\ben
\begin{split}
&\int_{\ell} \ e^{-\ell^2/\la^{2} } \prod_{i=1}^{N'} e^{-\frac{(k_i+\ell)^2}{ \al(N+2,L-1) \la_{i}^{2} }}\
 \log^{n}_{+}\sup\left(\frac{|(\vec{p},\ell)|_\LIR}{\min_{j}\lambda_{j}},
 \frac{\max_{j}\lambda_{j}}{\LIR}\right)\cdot \sum_{\mu=0}^{d}\frac{1}{\sqrt{\mu!}} \left(\frac{|(\vec{p},\ell)|}{M}\right)^{\mu} \\
 & \le\ 4\pi^{2} \la^{4} \,\prod_{i=1}^{N'}
 e^{-\frac{k_i^2} { \al(N,L)\la_{i}^{2}}} \left[  \log^{n}\sup\left(\frac{|\vec{p}|_\LIR}{\min_{j}\lambda_{j}},
 \frac{\max_{j}\lambda_{j}}{\LIR}\right) + \sqrt{n!} \right]\, (N'+1)^{(n+4)/2}\\
 &\quad\times 2^{d}\sup\left(1,\frac{ \la}{M}\right)^{d}\,  \sum_{\mu=0}^{d}\frac{1}{\sqrt{\mu!}} \left(\frac{|\vec{p}|}{M}\right)^{\mu}\, .
\end{split}
\een
\end{lemma}
\begin{proof}
To begin with, we rewrite the integral in the form
\ben\label{ldistribute}
\begin{split}
&\int_{\ell} \ \prod_{i=1}^{N'} \left(e^{-\frac{(k_i+\ell)^2}{ \al(N+2,L-1) \la_{i}^{2} }}  e^{-\frac{\ell^2}{2(N'+1)\la^{2}} } \right)
 \left( \log^{n}_{+}\sup\left(\frac{|(\vec{p},\ell)|_\LIR}{\min_{j}\lambda_{j}},
 \frac{\max_{j}\lambda_{j}}{\LIR}\right)   e^{-\frac{\ell^2}{2(N'+1)\la^{2}} } \right) \\
 &\qquad\times \sum_{\mu=0}^{d}\frac{1}{\sqrt{\mu!}} \left(\frac{|(\vec{p},\ell)|}{M}\right)^{\mu}  e^{-\ell^2/2\la^{2} }\,.
\end{split}
\een
We use $x^{n}e^{-x^{2}}\leq \sqrt{n!}$ to bound the polynomial factor in the second line (cf. (77) in~\cite{Hollands:2011gf})
\ben\label{polestimate}
e^{-\ell^{2}/2\la^{2}}\sum_{\mu=0}^{d}\frac{1}{\sqrt{\mu!}} \left(\frac{|(\vec{p},\ell)|}{M}\right)^{\mu}\leq 2^{d}\sup\left(1,\frac{\la}{M}\right)^{d}\sum_{\mu=0}^{d}\frac{1}{\sqrt{\mu!}} \left(\frac{|\vec{p}|}{M}\right)^{\mu}\, .
\een
For the exponential factors, we apply the bound
\ben \label{lemmaexpbound}
\exp\left(-\frac{\ell^2}{2(N'+1)\la^{2}} \right)\cdot \exp\left(-\frac{(k_i+\ell)^2}{ \al(N+2,L-1) \la_{i}^{2} } \right) \leq \exp\left(-\frac{k_i^2}{ \al(N,L) \la_{i}^{2} }\right)\, ,
\een
which can be verified as follows:  Define $\xi>0$ such that $|\ell|=\xi |k_{i}|$. Then
\ben
\exp\left(-\frac{\ell^2}{2(N'+1)\la^{2}} \right)\cdot \exp\left(-\frac{(k_i+\ell)^2}{ \al(N+2,L-1) \la_{i}^{2} } \right) \leq \exp\left(-\frac{k_{i}^{2}}{\la_{i}^{2}} \left[ \frac{\xi^{2}}{2(N'+1)}+ \frac{(1-\xi)^{2}}{\alpha(N+2,L-1)} \right]  \right)
\een
The expression in square brackets is minimal for $\xi=\frac{2(N'+1)}{2(N'+1)+\alpha(N+2,L-1)}$, where it takes the value $\frac{1}{2(N'+1)+\alpha(N+2,,L-1)}$. Thus,
\ben\label{lemexpmin}
\exp\left(-\frac{\ell^2}{2(N'+1)\la^{2}} \right)\cdot \exp\left(-\frac{(k_i+\ell)^2}{ \al(N+2,L-1) \la_{i}^{2} } \right) \leq \exp\left(-\frac{k_{i}^{2}}{\la_{i}^{2}} \left[\frac{1}{2(N'+1)+\alpha(N+2,,L-1)} \right]  \right)\, .
\een
The bound \eqref{lemmaexpbound} then follows by noting that
\ben
\frac{1}{2(N'+1)+\alpha(N+2,,L-1)} \geq \frac{1}{\alpha(N,L)}\, ,
\een
which in turn follows from (recall that $L\geq 1$ by assumption)
\ben\label{alphaineq}
\al(N,L) = 2\left(\frac{N}{2}+2L\right)^2 =\al(N+2,L-1) +4\left(\frac{N}{2}+2L\right)-2  \ge \al(N+2,L-1)+2N'+4 \ .
\een%
It remains to find a bound for the integral over the logarithm in \eqref{ldistribute}. We use the inequality 
\ben\label{hoelderlog}
\begin{split}
 &\int_{{\ell}}e^{-\frac{{\ell}^2}{2(N'+1)\la^{2}} }  \log^{n}_{+}\sup\left(\frac{|(\vec{p},\ell)|_\LIR}{\min_{j}\lambda_{j}},
 \frac{\max_{j}\lambda_{j}}{\LIR}\right)\\
 &= (N'+1)^{2} \int_{{\ell}}e^{-\frac{{\ell}^2}{2\la^{2}} }  \log^{n}_{+}\sup\left(\frac{|(\vec{p},\ell\cdot \sqrt{N'+1})|_\LIR}{\min_{j}\lambda_{j}},
 \frac{\max_{j}\lambda_{j}}{\LIR}\right) \\
& \leq (N'+1)^{(n+4)/2} \int_{{\ell}}e^{-\frac{{\ell}^2}{2\la^{2}}}  \log^{n}_{+}\sup\left(\frac{|(\vec{p},\ell)|_\LIR}{\min_{j}\lambda_{j}},
 \frac{\max_{j}\lambda_{j}}{\LIR}\right) \\
 & \leq  \la^{4}\,(N'+1)^{(n+4)/2}\,
  \left[ \log^{n}_{+}\sup\left(\frac{|\vec{p}|_\LIR}{\min_{j}\lambda_{j}},
 \frac{\max_{j}\lambda_{j}}{\LIR}\right)+ \sqrt{n!} \right]\, ,
 \end{split}
\een
which follows from an inequality in~\cite[cf. (52)]{Hollands:2011gf}. Inserting the bounds \eqref{polestimate}, \eqref{lemmaexpbound} and \eqref{hoelderlog} into \eqref{ldistribute} finishes the proof of the lemma.
\end{proof}
%\vspace{.5cm}
The next three lemmas help us to bound the ``flow integrals'' in the induction:
\begin{lemma}\label{lemmalambdalarge}
Let $n\in\mathbb{N}$, $1\leq s\in\mathbb{N}$  and $a,b,M,\La_{0}\in\mathbb{R}_{+}$ with $0<a<\La_{0}$. Then the following bound holds:
\ben\label{theta1ineq3}
\int_{a}^{\La_{0}}\d\la\, \log^{n}_{+}\left[\sup\left(b, \frac{\lambda}{\LIR}\right)\right]\, \la^{-1-s} \leq   \frac{1}{a^{s}}\sum_{j=0}^{n}\frac{ 3\, n!\, \log^{j}_{+}\left[\sup\left(b, \frac{a}{M} \right)\right]}{j!}\, .
\een
\end{lemma}
\begin{proof}
We decompose the integral:
\ben
\int_{a}^{\La_{0}} \d\la\ \log^{n}_{+}\left[\sup\left(b, \frac{\lambda}{\LIR}\right)\right]\, \la^{-1-s} \leq  \int_{a}^{\La_{0}} \d\la\ \left(  \log^{n}_{+}b  + \log^{n}_{+}\frac{\lambda}{\LIR}     \right)\la^{-1-s}
\een
Combining the trivial inequality
\ben
 \int_{a}^{\La_{0}} \d\la\ \log^{n}_{+}(b)\ \la^{-s-1}\leq a^{-s} \frac{\log^{n}_{+}b}{s}
\een
with the bound
\ben
\begin{split}
 \int_{a}^{\La_{0}} \d\la\ 
\log^{n}_{+}\left(\frac{\lambda}{\LIR}\right) \la^{-1-s} &=\int_{a}^{\sup(a,M)} \la^{-1-s} +  \int_{\sup(a,M)}^{\La_{0}} \d\la\ 
\log^{n}\left(\frac{\lambda}{\LIR}\right)\la^{-1-s}\\
&\leq \frac{1}{a^{s}\cdot s} + \frac{1}{a^{s}}\sum_{j=0}^{n}\frac{  n!\, \log^{j}_{+}\left( \frac{a}{M} \right)}{j!\cdot s^{n-j+1}}\, ,
\end{split}
\een
we verify the lemma (recall that $s\geq 1$).
\end{proof}
\begin{lemma}\label{theta1lem}
Using the notation of section \ref{sec:bounds}, in particular \eqref{fdef}, have the bound
\ben\label{theta1b}
\begin{split}
&\int_{\La}^{\La_{0}} \d\la_{a}\ f_{\lambda_{a}}(k_a;N,L,\theta_a=1) 
\mathcal{P}_{L-1}\left( \log_{+}\sup\left(\frac{|\vec{p}|_\LIR}{\min_{j\in\mathcal{I}}\la_{j}},
 \frac{\max_{j\in\mathcal{I}}\lambda_{j}}{\LIR}\right)\right)\cdot \min_{j\in\mathcal{I}}\la_{j}^{2}\\
 &\leq \mathcal{P}'_{L-1}\left( \log_{+}\sup\left(\frac{|\vec{p}|_\LIR}{\min_{j\in\mathcal{I}\setminus \{a\}}\la_{j}},
 \frac{\max_{j\in\mathcal{I}\setminus \{a\}}\la_{j}}{\LIR}\right)\right)\cdot \min_{j\in\mathcal{I}\setminus \{a\}}\lambda_{j}\, ,
\end{split}
\een
where $\mathcal{P}'_{L-1}$ is a polynomial of the same degree as $\mathcal{P}_{L-1}$
with coefficients depending on $N,\ L\,$.
\end{lemma}
\begin{proof}
The bound \eqref{theta1b}  follows from the decomposition 
\ben
\begin{split}
&\int_{\La}^{\La_{0}} \d\la_{a}\ f_{\lambda_{a}}(k_a;N,L,\theta_a=1) 
\mathcal{P}_{L-1}\left( \log_{+}\sup\left(\frac{|\vec{p}|_\LIR}{\min_{j\in\mathcal{I}}\lambda_{j}},
 \frac{\max_{j\in\mathcal{I}}\lambda_{j}}{\LIR}\right)\right)\cdot \min_{j\in\mathcal{I}}\la_{j}^{2}\\
 &=\int_{\La}^{\min_{j\in\mathcal{I}\setminus \{a\}}\la_{j}} \d\la_{a}\ f_{\lambda_{a}}(k_a;N,L,\theta_a=1) 
\mathcal{P}'_{L-1}\left( \log_{+}\sup\left(\frac{|\vec{p}|_\LIR}{\la_{a}},
 \frac{\max_{j\in\mathcal{I}}\lambda_{j}}{\LIR}\right)\right)\cdot \la_{a}^{2}\\
 &+\int\limits_{\min\limits_{j\in\mathcal{I}\setminus \{a\}}\la_{j}}^{\La_{0}} \d\la_{a}\ f_{\lambda_{a}}(k_a;N,L,\theta_a=1) 
\mathcal{P}''_{L-1}\left( \log_{+}\sup\left(\frac{|\vec{p}|_\LIR}{\min\limits_{j\in\mathcal{I}\setminus \{a\}}\lambda_{j}},
 \frac{\max\limits_{j\in\mathcal{I}}\lambda_{j}}{\LIR}\right)\right) \min_{j\in\mathcal{I}\setminus \{a\}}\la_{j}^{2}
\end{split}
\een
combined with the inequalities
\ben\label{theta1ineq1}
\begin{split}
&\int_{\La}^{\min_{j\in\mathcal{I}\setminus \{a\}}\la_{j}} \d\la_{a}\e^{-\frac{k_{a}^{2}}{\alpha\la_{a}^{2}}} \mathcal{P}_{L-1}\left( \log_{+}\sup\left(\frac{|\vec{p}|_\LIR}{\la_{a}},
 \frac{\max_{j\in\mathcal{I}}\lambda_{j}}{\LIR}\right)\right) \\
 &\leq \int_{0}^{\min_{j\in\mathcal{I}\setminus \{a\}}\la_{j}} \d\la_{a}\mathcal{P}_{L-1}\left( \log_{+}\sup\left(\frac{|\vec{p}|_\LIR}{\la_{a}},
 \frac{\max_{j\in\mathcal{I}\setminus \{a\}}\lambda_{j}}{\LIR}\right)\right)  \\
&\leq  \min_{j\in\mathcal{I}\setminus \{a\}}\lambda_{j} \cdot\mathcal{P}'_{L-1}\left(\log_{+}\sup\left(\frac{|\vec{p}|_{M}}{\min_{j\in\mathcal{I}\setminus \{a\}}\la_{j}}, \frac{\max_{j\in\mathcal{I}\setminus \{a\}}\lambda_{j}}{\LIR}\right)\right)
\end{split}
\een
and
\ben\label{theta1ineq2}
\begin{split}
&\int_{\min_{j\in\mathcal{I}\setminus \{a\}}\la_{j}}^{\La_{0}} \d\la_{a}\ \mathcal{P}_{L-1}\left( \log_{+}\sup\left(\frac{|\vec{p}|_\LIR}{\min\limits_{j\in\mathcal{I}\setminus \{a\}}\lambda_{j}},
 \frac{\max\limits_{j\in\mathcal{I}}\lambda_{j}}{\LIR}\right)\right)\la_{a}^{-2}\\
&= \int_{\min_{j\in\mathcal{I}\setminus \{a\}}\la_{j}}^{\La_{0}} \d\la_{a}\ \mathcal{P}_{L-1}\left( \log_{+}\sup\left(\frac{|\vec{p}|_\LIR}{\min\limits_{j\in\mathcal{I}\setminus \{a\}}\lambda_{j}},
 \frac{\max\limits_{j\in\mathcal{I}\setminus \{a\}}\lambda_{j}}{\LIR}, \frac{\lambda_{a}}{\LIR}\right)\right)\la_{a}^{-2}\\
&\leq \min_{j\in\mathcal{I}\setminus \{a\}}\lambda_{j}^{-1}   \mathcal{P}'_{L-1}\left( \log_{+}\sup\left(\frac{|\vec{p}|_\LIR}{\min\limits_{j\in\mathcal{I}\setminus \{a\}}\lambda_{j}},
 \frac{\max\limits_{j\in\mathcal{I}\setminus \{a\}}\lambda_{j}}{\LIR}\right)\right) \, .
\end{split}
\een
The last inequality in \eqref{theta1ineq1} follows from the relation (note that $N$ and $L$ dependent factors were again absorbed into the new polynomials $\mathcal{P}'_{L-1}$)
\ben\label{smalllabound}
\begin{split}
&\int_{0}^{a}\d\la \log^{n}\sup\left(\frac{|\vec{p}|_{M}}{\la}, b \right) \leq \int_{0}^{a}\d\la \left[\log^{n}\sup\left(\frac{|\vec{p}|_{M}}{\la}\right) +\log^{n}(b) \right]  \\
 &=  a \sum_{m=0}^{n} \left( n \atop m \right) \log^{n-m}\left(\frac{|\vec{p}|_{M}}{a}\right)\, m! + a\cdot \log^{n}(b) \leq  2a \sum_{m=0}^{n} \left( n \atop m \right) \log^{n-m}\sup\left(\frac{|\vec{p}|_{M}}{a}, b\right)\, m! \ \, ,
\end{split}
\een
and \eqref{theta1ineq2} follows from lemma \ref{lemmalambdalarge}. 
\end{proof}
\begin{lemma}\label{lemmaLambdaint}
For any $n,\alpha \in\mathbb{N}$, $1\leq \theta\in\mathbb{N}, 0<M<\La_{0}\in\mathbb{R}$ and for $|k|\leq |\vec{p}|$, the following bound holds:
\ben
\int_{0}^{\La_{0}}\d\la\, \frac{\e^{-\frac{k^{2}}{\alpha \la^{2}}}}{\la^{\theta+1}}\,  \log^{n}\sup\left( \frac{|\vec{p}|_{M}}{\la}, \frac{\la}{M} \right) 
\leq \sqrt{\theta!}\,\left(\frac{\sqrt{\alpha}}{|k|}\right)^{\theta}\, 
\sum_{m=0}^{n} \frac{3\, n!}{(n-m)!} \log^{n-m}\left(\frac{|\vec{p}|_{M}}{\inf(|k|,M)}\right)\, (\frac{\alpha}{2})^{m}\ .
\een
\end{lemma}
\begin{proof}
We decompose the integral into two parts:
\ben\label{lemmalambdadecomp}
\int_{0}^{\La_{0}}\d\la\, \frac{\e^{-\frac{k^{2}}{\alpha \la^{2}}}}{\la^{\theta+1}}\, \log^{n}\sup\left(\frac{|\vec{p}|_\LIR}{\la},
 \frac{\la}{\LIR}\right)=\left( \int_{0}^{|k|}\d\la+ \int_{|k|}^{\La_{0}}\d\la\right) \frac{\e^{-\frac{k^{2}}{\alpha \la^{2}}}}{\la^{\theta+1}}\,\log^{n}\sup\left(\frac{|\vec{p}|_\LIR}{\la},
 \frac{\la}{\LIR}\right)\ .
\een
To bound the integral over large $\la$, we make use of the inequality 
%(see also (53) in [Hollands-Kopper])
%
\ben\label{Lambdaint}
\int_{a}^{\Lambda_{0}}\d\lambda\, \lambda^{-\theta-1} \log^{n}\sup\left(\frac{|\vec{p}|_\LIR}{\la},
 \frac{\la}{\LIR}\right) \leq\, a^{-\theta} \,\sum_{m=0}^{n}\frac{2\,  n!}{(n-m)!}\,\log^{n-m}\sup\left(\frac{|\vec{p}|_{M}}{a}, \frac{a}{M} \right)\, ,
\een
which follows from
\ben
\begin{split}
&\int_{a}^{\Lambda_{0}}\d\lambda\, \lambda^{-\theta-1} \log^{n}\sup\left(\frac{|\vec{p}|_\LIR}{\la},
 \frac{\la}{\LIR}\right) \\
& \leq\int_{a}^{\La_{0}}\d\lambda   \lambda^{-\theta-1} \log^{n}\left(\frac{|\vec{p}|_\LIR}{\la}\right)  +  \int_{\sup(a,\sqrt{M|\vec{p}|_{M}})}^{\La_{0}}\d\lambda \lambda^{-\theta-1} \log^{n}\left(
 \frac{\la}{\LIR}\right)\ ,
\end{split}
\een
and from 
\ben
\int_{a}^{\La_{0}}\d\lambda   \lambda^{-\theta-1} \log^{n}\left(\frac{|\vec{p}|_\LIR}{\la}\right)\leq   a^{-\theta} \log^{n}\left(\frac{|\vec{p}|_\LIR}{a}\right)
\een
\ben
\begin{split}
 \int_{\sup(a,\sqrt{M|\vec{p}|_{M}})}^{\La_{0}}\d\lambda \lambda^{-\theta-1} \log^{n}\left(
 \frac{\la}{\LIR}\right)&\leq  \int_{\sup(a,M)}^{\La_{0}}\d\lambda \lambda^{-\theta-1} \log^{n}\left(
 \frac{\la}{\LIR}\right)\\
&  \leq  \frac{1}{\sup(a,M)^{\theta}}\sum_{m=0}^{n}\frac{  n!\, \log^{n-m}_{+}\left( \frac{a}{M} \right)}{(n-m)!}\ \ .
\end{split}
\een
Setting $a=|k|$ in \eqref{Lambdaint} and recalling that $|k|\leq |\vec{p}|$, we see that this contribution is consistent with the bound claimed in lemma \ref{lemmaLambdaint}. It remains to bound the integral over $\la\leq |k|$ in \eqref{lemmalambdadecomp}. Combining the bound
\ben
\begin{split}
&\log^{n}\left(\frac{|\vec{p}|_{M}}{\la}\right)\, \e^{-\frac{k^{2}}{2\alpha\la^{2}}}
%= \left[\log\left(\frac{|\vec{p}|_{M}}{k}\right)+\log\left(\frac{k}{\la}\right)\right]^{n} \e^{-\frac{k^{2}}{2\alpha\la^{2}}}\\
=\sum_{m=0}^{n}\left(n\atop m\right) \log^{n-m}\left(\frac{|\vec{p}|_{M}}{|k|}\right)\,\log^{m}\left(\frac{|k|}{\la}\right)\e^{-\frac{k^{2}}{2\alpha\la^{2}}}\\
&\leq \sum_{m=0}^{n}\left(n\atop m\right) \log^{n-m}\left(\frac{|\vec{p}|_{M}}{|k|}\right)\,\left(\frac{k^{2}}{2\la^{2}}\right)^{m}\e^{-\frac{k^{2}}{2\alpha\la^{2}}}\leq \sum_{m=0}^{n}\left(n\atop m\right) \log^{n-m}\left(\frac{|\vec{p}|_{M}}{|k|}\right)\, (\frac{\alpha}{2})^{m}\, m!
\end{split}
\een
with the relation
\ben
\int_{0}^{\infty} \d\la \, \la^{-\theta-1} \cdot  \e^{-\frac{k^{2}}{2\alpha \la^{2}}} =\frac{\Gamma(\theta/2)}{2} \left(\frac{\sqrt{2\alpha}}{|k|}\right)^{\theta}\leq \sqrt{\theta !} \left(\frac{\sqrt{\alpha}}{|k|}\right)^{\theta}
\een
we obtain the bound
\ben
\begin{split}
&\int_{0}^{|k|}\d\la\, \frac{\e^{-\frac{k^{2}}{\alpha \la^{2}}}}{\la^{\theta+1}}\,\log^{n}\sup\left(\frac{|\vec{p}|_\LIR}{\la},
 \frac{\la}{\LIR}\right)\\
 & \leq \sqrt{\theta!}\,\left(\frac{\sqrt{\alpha}}{|k|}\right)^{\theta}\, 
\sum_{m=0}^{n}\frac{n!}{(n-m)!} \log^{n-m}\sup\left(\frac{|\vec{p}|_{M}}{|k|}, \frac{|k|}{M}\right)\, (\frac{\alpha}{2})^{m}
\end{split}
\een
which, combined with \eqref{Lambdaint}, establishes the lemma in view of our assumption $|k|\leq |\vec{p}|$.
\end{proof}
\section{Proof of corollary \ref{cor3b}}\label{appdistr}
%{Proof of proposition \ref{propschw}}
Here we prove the bound \eqref{CAGloopintfinal} stated in corollary \ref{cor3b}. For $N=0$ there is nothing to show, as the claimed bound follows directly from corollary \ref{cor3}.
For $N>0$ we start from the bounds given in theorem~\ref{thm3},  which immediately give for some constant $K=K(N,L)$
\ben\label{CAGloopint1}
\begin{split}
&\left|\int_{p_{1},\ldots,p_{N}}  \L_{D,N,L}^{0,\infty}\left( \O_{A}(x)\otimes\O_{B}(0);\vec{p}\right)\, \prod_{i=1}^{N}  \frac{\hat{\test}_{i}(p_{i})}{p_i^2} \right|  \leq  \sqrt{[A]![B]!}\ K^{[A]+[B] } \\
&\times \int_{\vec{p}}\prod_{i=1}^{N} \frac{|\hat{\test}_{i}(p_{i})|}{p_i^2} \sum_{ T \in \tkn_{N,2L+2,0}(\vec{p})}
  \sup\left(M,\frac{1}{|x|}\right)^{[A]+[B]-D}\,  \LIR^{D-N_{\mathcal{V}}}  \sum_{\mu=0}^{([A]+[B])(N+2L+3)}\frac{1}{\sqrt{\mu!}}\
\left(\frac{|\vec p_{\mathcal{V}}|}{\LIR}\right)^{\mu}   \\
&\times \prod_{i \in \mathcal{I}(T)} \int_{0}^{\infty} \d\la_{i}\ {f}_{\lambda_{i}}(k_i;N-2,L+1,\theta_i)  \,
 \mathcal{P}_{2L+\frac{N}{2}}\left( \log_{+}\sup\left(\frac{|\vec{p}|}{\LIR},\frac{|\vec{p}|_\LIR}{\min_{j\in\mathcal{E}}\lambda_{j}},
 \frac{\max_{j\in\mathcal{E}}\lambda_{j}}{\LIR}\right)\right) \, ,
\end{split}
\een
where ${\test}_i\in\mathcal{S}(\mathbb{R}^4)$. Since the test functions $\hat{\test}_i(p)$ decay rapidly for large $p$, they will allow us to control the UV-behaviour of the momentum integrals. 
 It is convenient to exploit this decay with the inequality
\ben
\begin{split}
|p|^{\mu-2}\,  |\hat{\test}(p)| \leq \frac{\|\hat{\test}\|_{\frac{s+\mu}{2}}}{p^2\cdot (M^2+p^{2})^{s/2}} &=\frac{s}{2}\cdot \int_0^1 \d\tau \frac{\tau^{s/2-1}\, \|\hat{\test}\|_{\frac{s+\mu}{2}}}{[\tau (M^2+p^2)+(
1-\tau) p^2]^{\frac{s+2}{2}}}\\
& =\|\hat{\test}\|_{\frac{s+\mu}{2}}\cdot \frac{2}{\Gamma({s/2})} \int_0^\infty \d\la \, \la^{-s-3} e^{-\frac{p^2}{\la^2}} \int_0^1\d\tau \, \tau^{s/2-1} e^{-\frac{\tau M^2}{\la^2}} \, ,
\end{split}
\een
which holds for $s>0$ and $\mu\geq 0$, and where $\|\cdot\|_n$ is defined in \eqref{Schwnorms}. Upon substitution into \eqref{CAGloopint1} and using the inequalities $|\vec{p}_{\mathcal{V}}|\leq |\vec{p}|\leq |p_1|+\ldots+|p_N|$, we get
\ben\label{CAGloopint2}
\begin{split}
&\left|\int_{p_{1},\ldots,p_{N}}  \L_{D,N,L}^{0,\infty}\left( \O_{A}(x)\otimes\O_{B}(0);\vec{p}\right)\, \prod_{i=1}^{N} \frac{\hat{\test}_{i}(p_{i})}{p_i^2} \right|\leq  \sqrt{[A]![B]!}\ K^{[A]+[B] }  \\
&\times \sum_{ T \in \tkn_{N,2L+2,0}(\vec{p})}
  \sup\left(M,\frac{1}{|x|}\right)^{[A]+[B]-D}\,  \LIR^{D-N_{\mathcal{V}}}  \sum_{\mu=0}^{([A]+[B])(N+2L+3)}\sum_{\mu_1+\ldots+\mu_N=\mu} \frac{\prod_{i=1}^N  \|\hat{\test}_{i}\|_{ \frac{\mu_{i}+s_{i}} {2}}} { M^{\mu}}    \\
&\times \int_{\vec{p}}\prod_{i \in \mathcal{E}(T)} \int_{0}^{\infty} \d\la_{i}\ \tilde{f}_{\lambda_{i}}(k_i;\alpha,\theta_i)  \,
 \mathcal{P}_{2L+\frac{N}{2}}\left( \log_{+}\sup\left(\frac{|\vec{p}|}{\LIR},\frac{|\vec{p}|_\LIR}{\min_{j\in\mathcal{E}}\lambda_{j}},
 \frac{\max_{j\in\mathcal{E}}\lambda_{j}}{\LIR}\right)\right) \, ,
\end{split}
\een
where $\mathcal{E}(T)$ denotes the set of all lines (internal and external) of the tree $T$, and where we define modified weight factors by [compare eq.\eqref{fdef}]
\ben\label{ftildedef}
 \tilde{f}_{\lambda_{i}}(k_{i};\alpha,\theta_{i}):= \begin{cases}
 \la_{i}^{-\theta_{i}-1} e^{-k_{i}^{2}/\alpha\la_{i}^{2}} & \text{for }i\in \mathcal{I} \text{ internal line}  \\
  \la_{i}^{-s_{i}-3} \,e^{-k_{i}^{2}/\alpha\la_{i}^{2}} \int_0^1\d\tau \, \tau^{s_i/2-1} e^{-\frac{\tau M^2}{\la_i^2}}  & \text{for }i \text{ an external line of }T \text{ if } s_i>0 \\
    \la_{i}^{-3} \,e^{-k_{i}^{2}/\alpha\la_{i}^{2}}  & \text{for }i \text{ an external line of }T \text{ if } s_i=0 \, ,
 \end{cases}
\een
with $\alpha=\alpha(N-2,L+1)$ as before. The parameters $s_{i}\in\mathbb{N}$ can be chosen freely at this stage. 
\noindent Next we decompose the $\lambda_{i}$ integrations into 'sectors', i.e. we write the last line of \eqref{CAGloopint2} as 
\ben\label{CAGloopint3}
\begin{split}
& \sum_{ \pi\in\mathfrak{S}(\mathcal{E}) } \int_{\vec{p}}\int_{0}^{\infty} \d\la_{\pi_{E}}\,  \tilde{f}_{\la_{\pi_E}}(k_{\pi_E};\alpha,\theta_{\pi_E}) \int_{0}^{\la_{\pi_{E}}} \d\la_{\pi_{E-1}}\ldots\int_{0}^{\la_{2}} \d\la_{\pi_{1}}   \tilde{f}_{\la_{\pi_1}}(k_{\pi_1};\alpha,\theta_{\pi_1})   \\
&\times \mathcal{P}_{2L+\frac{N}{2}}\left( \log_{+}\sup\left(\frac{|\vec{p}|}{\LIR},\frac{|\vec{p}|_\LIR}{{\lambda_{\pi_1}}},	
 \frac{\lambda_{\pi_E}}{\LIR}\right)\right) \, ,
\end{split}
\een
where $\mathfrak{S}(\mathcal{E})$ is the set of permutations of the internal and external lines of $T$ [i.e. the \emph{symmetric group} on $\mathcal{E}(T)$] and where we denote by $E=|\mathcal{E}(T)|$ the total number of lines of $T$. Writing the $\la_i$-integrals in this form will allow us to bound the momentum integrals one-by-one successively. Since the $p_i$ are related to the momenta $k_{\pi_j}$ by momentum conservation in the tree $T$, we may assume  that $p_1$ is the momentum flowing through\footnote{We note that there are no other lines in the tree with the same associated momentum $k_{\pi_{1}}$, since vertices of coordination number two are not allowed in trees $\tkn_{N,L,0}$. This follows from the conflicting requirements $\rho=0$ and $\theta>0$ for lines associated to such a vertex.} the line $\pi_1$. Then our independent integration variables are $(k_{\pi_{1}},p_{2},\ldots,p_{N})$. To get a bound, we use lemma \ref{lablemma1}.
% (note that $\la_{\pi_{1}}\leq \la_{\pi_{
% i}}$ for all $i\leq N$).  
 We distinguish two cases:
\begin{enumerate}
\item Let us first assume that $\pi_{1}$ is an internal line of $T$. To bound the integral over $k_{\pi_1}$ we proceed as in the proof of lemma \ref{lablemma1}, i.e.  we use \eqref{lemexpmin} to bound the exponential factors and \eqref{hoelderlog} to bound the remaining integral over the logarithms. Adapting these inequalities to the case at hand, we have to replace $(N'+1)$ by $E \alpha$, because there are $E$ exponential weight factors in \eqref{CAGloopint3} (as opposed to $(N'+1)$ in lemma \ref{lablemma1}), and because we have a factor $\exp(-k_{\pi_1}^2/\alpha\la_{\pi_1}^2)$ in \eqref{CAGloopint3} (as opposed to a factor $\exp(-\ell^2/\la^2)$ in lemma \ref{lablemma1}). Thus, one finds for each term inside the sum in \eqref{CAGloopint3} (i.e. for each sector) the bound 
\ben\label{CAGloopint4}
\begin{split}
& \int_{(p_{2},\ldots,p_{N})}  \int_{0}^{\infty} \d\la_{\pi_{E}}\int_{0}^{\la_{\pi_{E}}} \d\la_{\pi_{E-1}}\ldots\int_{0}^{\la_{\pi_{3}}} \d\la_{\pi_{2}}     \prod_{\substack{i\in\mathcal{E}(T)\\ i\neq \pi_{1}  } } \tilde{f}_{\la_{i}}(k_{i}; \tilde{\alpha}, \theta_{i}) \\
&\qquad\times\int_{0}^{\la_{\pi_{2}}} \d\la_{\pi_{1}}    \la_{\pi_{1}}^{3-\theta_{\pi_{1}}} \ \mathcal{P}_{2L+\frac{N}{2}}\left( \log_{+}\sup\left(\frac{|\vec{p}|}{\LIR},\frac{|\vec{p}|_\LIR}{{\lambda_{\pi_1}}},	
 \frac{\lambda_{\pi_E}}{\LIR}\right)\right) \Bigg|_{k_{\pi_{1}}=0} \, ,
\end{split}
\een
where $\tilde{\alpha}=\tilde{\alpha}(N,L)$ has to satisfy $\frac{1}{\tilde{\alpha}}\leq \frac{1}{2E\alpha+ \alpha} $, and where we have absorbed $N$ and $L$ dependent constants in the polynomials $\mathcal{P}$. The condition on $\tilde{\alpha}$ follows from \eqref{lemexpmin}. 

 Next we observe that the exponent of $\la_{\pi_1}$ is positive, because one has $\theta_{i}\leq 2$ for the trees under consideration. We can therefore bound the integral over $\la_{\pi_{1}}$, using $\la_{\pi_{1}}^{3-\theta_{\pi_{1}}} \leq \la_{\pi_{2}}^{3-\theta_{\pi_{1}}} $ and using \eqref{smalllabound}. We obtain a factor $\la_{\pi_{2}}^{4-\theta_{\pi_{1}}}$ in the process. To put the resulting expression into a more convenient form in terms of trees, we next   perform a cutting procedure on our tree $T$ (see fig.\ref{figcut}):
\begin{enumerate}
\item Split the tree $T$ into two pieces by cutting the internal line $\pi_{1}$. 
\item If the line $\pi_1$ had a weight $\theta_{\pi_{1}}=2$, then this automatically yields a pair of trees $T_{1}\in \mathcal{T}_{n_{1}+1,l_{1},0}$ and $T_{2}\in \tkn_{n_{2}+1,l_{2},0}$ with $n_{1}+n_{2}=N, l_{1}+l_{2}=L$, and one can proceed to step \ref{stepc}. If however $\theta_{\pi_{1}}=1$, then one of the trees will have an illegal weight: Recall from condition \ref{it9} of def. \ref{deftrees1} that we associate to vertices of coordination number $<4$ adjacent internal lines with weight $\rho_i<2$. Cutting a line with $\rho_i=1$ will leave a 3 vertex without an associated ``partner-line''. We use a factor $\la_{\pi_{2}}$ in order to reduce the weight of another line adjacent to that vertex (if that line reaches weight $0$, we delete it and fuse the adjacent vertices), which yields two legally weighted trees $T_{1}\in \mathcal{T}_{n_{1}+1,l_{1},0}$ and $T_{2}\in \tkn_{n_{2}+1,l_{2},0}$. See fig.~\ref{figcut} for an example.
\begin{figure}[htbp]
\begin{center}
\includegraphics[width=\textwidth]{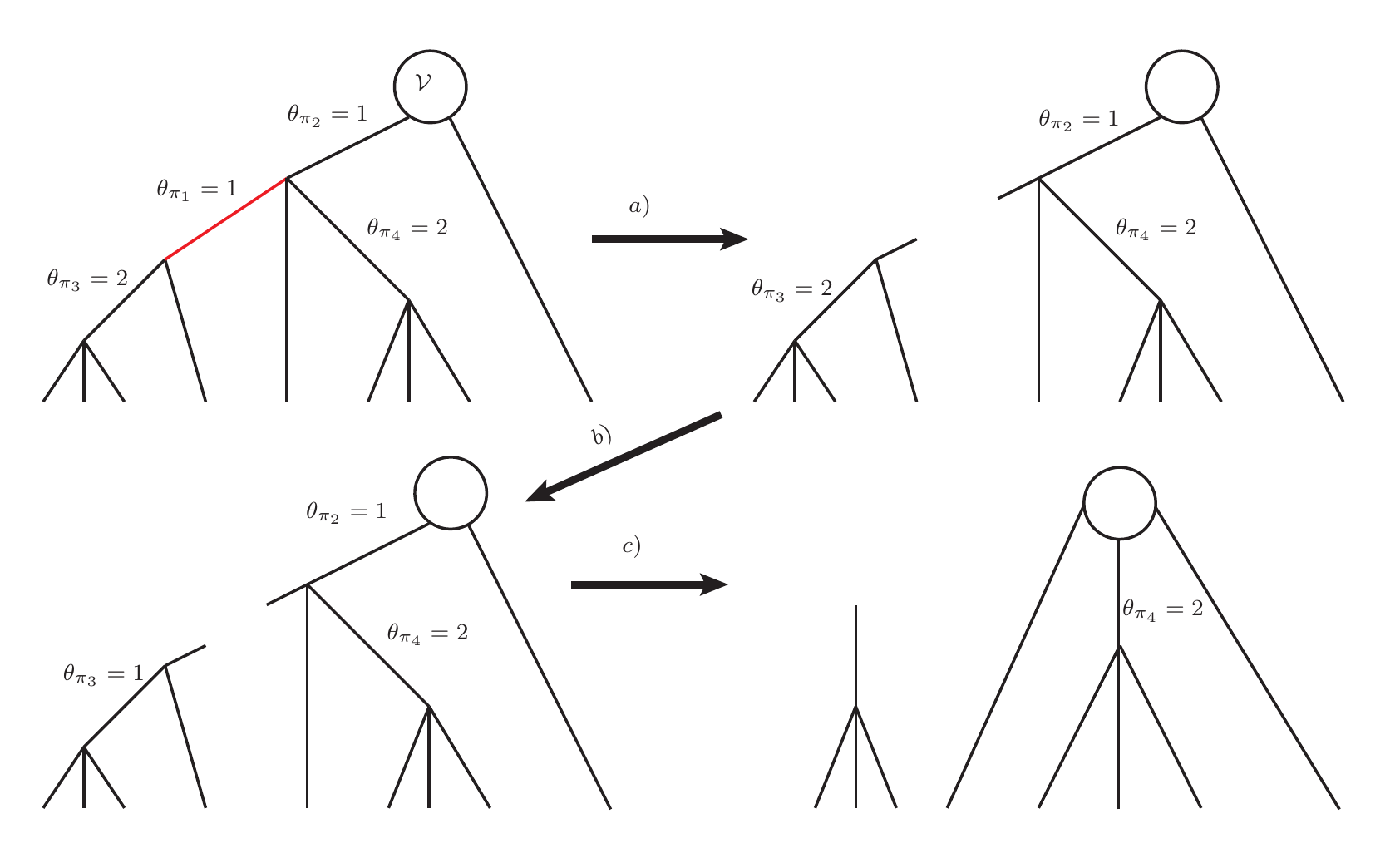}
\end{center}
\caption{We cut the tree $T$ along the line $\pi_{1}$ (marked red), obtaining two trees $T_{1}$ and $T_{2}$ as a result [step a)]. The line $\pi_{3}$ of the tree $T_{1}$ has an illegal weight $\theta_{\pi_{3}}=2$. We reduce this weight by one, using up a factor $\la_{\pi_{2}}$, and we arrive at two legally weighted trees [step b)]. We then perform a reduction procedure on each of those trees, removing the external line $\pi_{1}$ [step c)]. }
\label{figcut}
\end{figure}
\item\label{stepc} On each of those trees $T_{1}$ and $T_{2}$ perform a reduction procedure of the kind introduced previously (see pages \pageref{reduct} and \pageref{fig:LoopCase3}), removing the external line $\pi_{1}$ from each tree. This procedure eats a factor $\la_{\pi_{2}}^{2}$ and it produces two trees in  $\mathcal{T}_{n_{1},l_{1}+1,0}$ and $\tkn_{n_{2},l_{2}+1,0}$, respectively. If $\pi_1$ is directly attached to the vertex $\mathcal{V}$ in the tree $T_2$, then we can remove that line without any reduction procedure, and we only use up one power of $\la_{\pi_2}$.
\item\label{itc} So far we have in fact only defined trees $T\in \mathcal{T}_{n,l,0}$ with $n\geq 4$ , see def. \ref{deftrees1}. The extension to $n<4$ is simple: The sets $\mathcal{T}_{3,l,0},\mathcal{T}_{2,l,0},\mathcal{T}_{1,l,0}$ contain only the 3-vertex, the 2-vertex and the single (external) line, respectively. None of these trees have any internal lines. To extend the reduction procedure, which has so far only been defined for trees $T\in \mathcal{T}_{n,l,0}$ with $n>4$, to trees with $n\leq 4$, we proceed as follows  (see fig.\ref{figreductionspecial}):
\begin{figure}[htbp]
\begin{center}
\includegraphics[width=.6\textwidth]{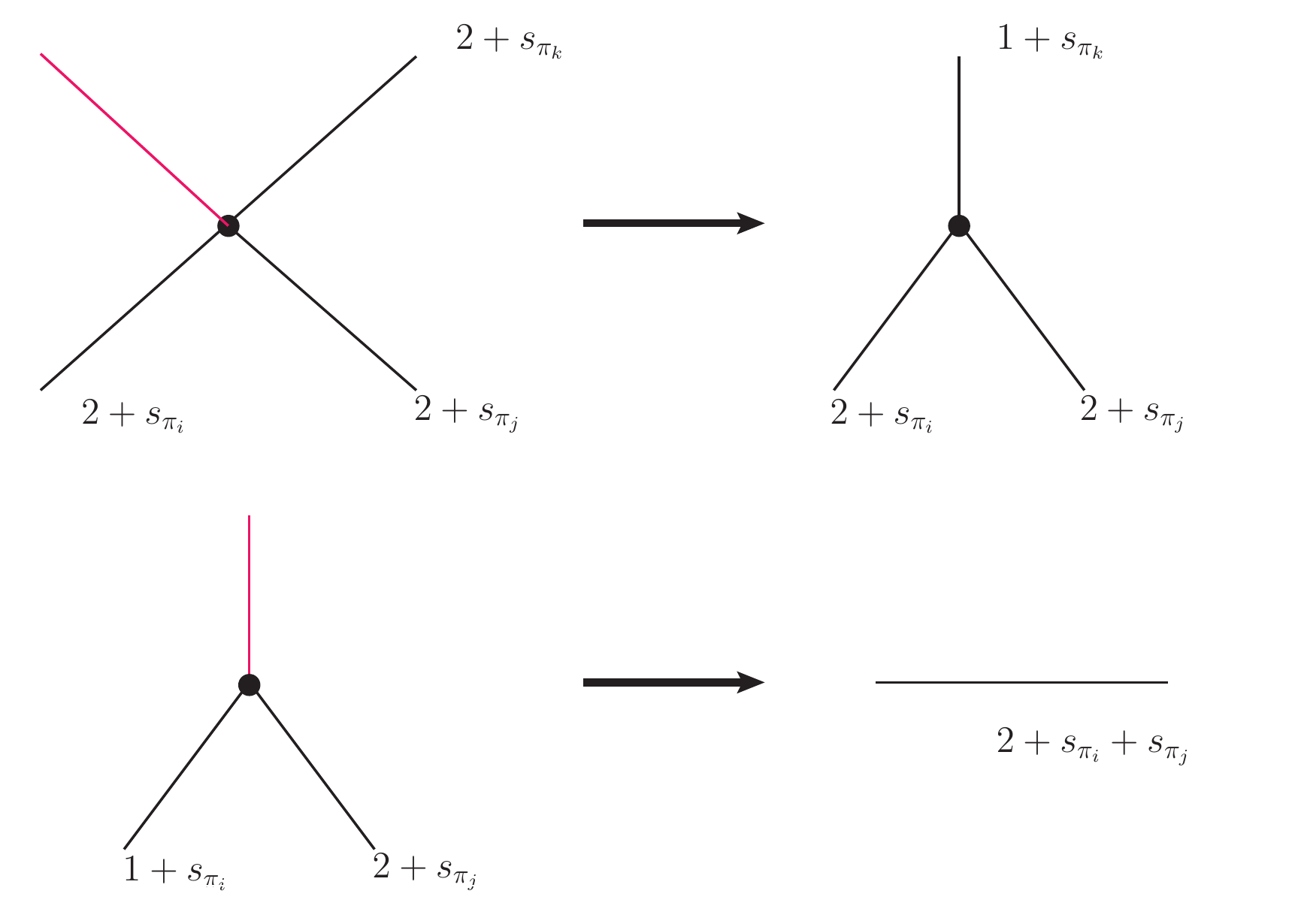}
\end{center}
\caption{Reduction of trees in $\mathcal{T}_{n,l,0}$ with $n\leq 4$: We remove the red line (corresponding to $\pi_1$ in the discussion above) and reduce the weight of another external line by one. A line labelled $a+s_{\pi_i}$ comes with a factor $ \la_{\pi_i}^{-s_{\pi_i}-a-1} \,e^{-k_{\pi_i}^{2}/\alpha\la_{\pi_i}^{2}} \int_0^1\d\tau \, \tau^{s_{\pi_i}/2-1} e^{-{\tau M^2}/{\la_{\pi_i}^2}}$ if $s_{\pi_i}>0$ and with a factor  $ \la_{\pi_i}^{
-a-1} \,e^{-k_{\pi_i}^{2}/\alpha\la_{\pi_i}^{2}} $ if $s_{\pi_i}=0$. }
\label{figreductionspecial}
\end{figure}
%
%
% We assume here for simplicity that only one of the additional weights $s_{\pi_i}$ on the external lines of the trees (cf. \eqref{ftildedef}) can be non-zero. Indeed, in the discussion below (see item \ref{itemext}) we will set all but one of these weights equal to zero. 
We delete the external line $\pi_1$ of the tree and reduce by one the weight (i.e. the power of $1/\la_{\pi_i}$ in the factor \eqref{ftildedef}) of another external line of the tree. For concreteness, we use the convention that the weight of the line with the largest $\la_{\pi_i}$ is reduced. If a vertex of coordination number 2 is created in the process, we delete that vertex and fuse the adjacent external lines, adding up their weights.  Thus, no vertices with coordination number two appear in our trees.
\end{enumerate}
To sum up, we can bound the expression inside the sum in \eqref{CAGloopint3} by
\ben
\begin{split}
& \int_{(p_{2},\ldots,p_{N})}  \int_{0}^{\infty} \d\la_{\pi_{E_{1}+E_{2}+1}}\ldots\int_{0}^{\la_{\pi_{3}}} \d\la_{\pi_{2}}     \prod_{{i\in\mathcal{E}(T_{1})\cup\mathcal{E}(T_{2})  } } \tilde{f}_{\la_{i}}(k_{i};\tilde{\alpha},\theta_{i}) \\
&   \times  \mathcal{P}_{2L+\frac{N}{2}}\left( \log_{+}\sup\left(\frac{|\vec{p}|}{\LIR},\frac{|\vec{p}|_\LIR}{{\lambda_{\pi_2}}},	
 \frac{\lambda_{\pi_{E_1+E_2+1}}}{\LIR}\right)\right) \times 
 \begin{cases}
 \la_{\pi_2} & \text{ if } \pi_1 \text{ is directly attached to } \mathcal{V} \\
 1 & \text{ otherwise}
 \end{cases}
   \, ,
\end{split}
\een
where $E_{i}$ is the number of lines in the tree $T_{i}$. 
\item\label{itemext} %
Consider now the case where $\pi_{1}$ is an external line of the tree $T$. Recall that we are free to choose the parameter $s_{\pi_1}$ affecting the weight of this line [cf. \eqref{ftildedef}]. The most convenient choice at this stage is simply $s_{\pi_1}=0$. Using lemma \ref{lablemma1}, we obtain again the bound \eqref{CAGloopint4}, but with a factor $\la_{\pi_{1}}$ instead of $ \la_{\pi_{1}}^{3-\theta_{\pi_{1}}}$ in the last line. Integrating over $\la_{\pi_{1}}$, we therefore obtain a factor $\la_{\pi_{2}}^{2}$. If $\pi_1$ is not directly attached to $\mathcal{V}$, we can use one power of $\la_{\pi_{2}}$ to remove the external line $\pi_{1}$ from the tree, using the now familiar reduction procedure. If $\pi_1$ is directly attached to $\mathcal{V}$, then we can simply delete that line without any reduction procedure. As a result, we obtain a tree $T_3\in\tkn_{N-1,2L+4,0}$, i.e. we have the following bound on the terms inside the sum in \eqref{CAGloopint3}
\ben
\begin{split}
& \int_{(p_{2},\ldots,p_{N})}  \int_{0}^{\infty} \d\la_{\pi_{E(T_3)+1}}\ldots\int_{0}^{\la_{\pi_{3}}} \d\la_{\pi_{2}}     \prod_{{i\in\mathcal{E}(T_3)  } } \tilde{f}_{\la_{i}}(k_{i};\tilde{\alpha},\theta_{i}) \\
& \times  \mathcal{P}_{2L+\frac{N}{2}}\left( \log_{+}\sup\left(\frac{|\vec{p}|}{\LIR},\frac{|\vec{p}|_\LIR}{{\lambda_{\pi_2}}},	
 \frac{\lambda_{\pi_{E(T_3)+1}}}{\LIR}\right)\right) \times 
 \begin{cases}
  \la_{\pi_{2}}^2 & \text{ if } \pi_1 \text{ is directly attached to }\mathcal{V}\\
  \la_{\pi_{2}} &  \text{ otherwise }
  \end{cases}
   \, .
\end{split}
\een
\end{enumerate}
Since the trees $T_{1},T_{2}$ and $T_3$ are again in either of
the sets $\mathcal{T}$ and $\tkn$, we can repeat the procedure with the line $\pi_{2}$. After $N-1$ iterations of this scheme, only one momentum integral remains. As regards the trees in the resulting bound we are left with only one external line for all of them. Therefore the bound is expressed in terms of $k\leq N$ trees $T_{1}\in \tkn_{n_{1},l_{1},0}$,  $T_{2}\in \mathcal{T}_{n_{2},l_{2},0}, \ldots , T_{k}\in \mathcal{T}_{n_{k},l_{k},0}$ with $n_i=1$ for some $i\leq k$ and $n_j=0$ for all $j\neq i$. Since no vertices of coordination number $2$ are allowed, these trees contain no internal lines. The corresponding bound hence contains only one momentum integral and one $\la$-integral. 

While iterating the cutting procedure, we have picked up positive powers of the $\la_{\pi_i}$ at various stages. The final exponent of $\la_{\pi_N}$ after $(N-1)$-iterations can be checked via power counting: We performed $(N-1)$ momentum integrals, which yield a power $\la_{\pi_N}^{4(N-1)}$. The removed $(N-1)$ external lines 
contribute a power of $\la_{\pi_N}^{-2(N-1)}$, and the internal lines a power of $\la_{\pi_N}^{N_{\mathcal{V}}-N}$. Thus, each term in the sum in \eqref{CAGloopint3} is bounded by
\ben\label{momentumfinal}
\begin{split}
   \int_{0}^{\infty} \d\la_{\pi_N} \int_p  \la_{\pi_N}^{N+N_{\mathcal{V}}-2}  \tilde{f}_{\la_{\pi_N}} (p;\tilde{\alpha},\theta_{\pi_N})  \mathcal{P}_{2L+\frac{N}{2}}\left( \log_{+}\sup\left(\frac{|{p}|}{\LIR},\frac{|{p}|_\LIR}{\lambda_{\pi_N}},	
 \frac{\lambda_{\pi_N}}{\LIR}\right)\right) \, .
\end{split}
\een
The final momentum integral can be bounded as before. Choosing the weight $s_{\pi_N}$ of the final external line larger than zero, we arrive at
\ben\label{sectorbound1}
\begin{split}
   \int_{0}^{\infty} \d\la_{\pi_N}  \  \la_{\pi_N}^{N+N_{\mathcal{V}}-1-s_{\pi_N}} \int_0^1\d\tau\, \tau^{s_{\pi_N}/2-1} e^{-\tau M^2/\la_{\pi_N}^{2}}  \mathcal{P}_{2L+\frac{N}{2}}\left( \log_{+}\sup\left(\frac{\LIR}{\la_{\pi_N}},	
 \frac{\lambda_{\pi_N}}{\LIR}\right)\right)\ .
\end{split}
\een  
For the integral over $\la_{\pi_N}$ to be convergent, we need to choose $s_{\pi_N}> N+N_{\mathcal{V}}$. In particular we are free to pick for example $s_{\pi_N}=N+N_{\mathcal{V}}+1$, which allows us to bound the final $\la$-integral as follows:
\ben
\begin{split}
 &\int_{0}^{\infty} \d\la_{\pi_N}    \la_{\pi_N}^{N+N_{\mathcal{V}}-1-s_{\pi_N}}e^{-\tau M^2/\la_{\pi_N}^{2}} \mathcal{P}_{2L+\frac{N}{2}}\left( \log_{+}\sup\left(\frac{\LIR}{\la_{\pi_N}},	
 \frac{\lambda_{\pi_N}}{\LIR}\right)\right) \\
 &\leq \tilde{K}
\int_{0}^{\infty} \d\la_{\pi_N}  \  \la_{\pi_N}^{-2} e^{-\tau M^2/2\la_{\pi_N}^{2}}\ \left( \frac{M}{\la_{\pi_N}}+\frac{\la_{\pi_N}}{M} \right)^{1/2}  
%\leq \mathcal{P}_{2L+\frac{N}{2}}\left(1\right) \frac{ \left(\tau^{-3/2}+ \tau^{-1/2} \right)}{M^2}
\leq    \tilde{K}\ %\mathcal{P}_{2L+\frac{N}{2}}\left( 1\right) 
\, \frac{\tau^{-1/2}}{ M^{3/2}}\ ,
\end{split}
 \een
for some positive constant $\tilde K=\tilde K(N,L)$. Using this inequality, it is easy to bound the $\tau$-integral in \eqref{sectorbound1} (recall that we consider the case $N>0$):
\ben
\eqref{sectorbound1} \leq \tilde K\, M^{-3/2}\, \int_0^1\d\tau \, \tau^{\frac{N+N_{\mathcal{V}} -2 }{2}}   \leq \tilde K\,  M^{-3/2}\, .
\een
To sum up, we have found that, choosing the parameters $s_{\pi_i}$ as discussed above, every term in the sum in \eqref{CAGloopint3} is bounded simply by $\tilde K\,  M^{-3/2}$.
Since the bound holds  %for any sector in \eqref{CAGloopint3} and
 for any tree $T$ in the sum over $\tkn_{N,2L+2,0}$ in \eqref{CAGloopint2}, we obtain in total
 \ben\label{CAGloopintfinal3}
\begin{split}
&\left|\int_{p_{1},\ldots,p_{N}}  \L_{D,N,L}^{0,\infty}\left( \O_{A}(x)\otimes\O_{B}(0);\vec{p}\right)\, \prod_{i=1}^{N}  \frac{\hat{\test}_{i}(p_{i})}{p_i^2} \right| \\
&\leq  \sqrt{[A]![B]!}\ K^{[A]+[B] } \, \sup(M,\frac{1}{|x|})^{[A]+[B]-D} \sum_{n=0}^N
  \LIR^{D-n-3/2}  \sum_{\mu=0}^{([A]+[B])(N+2L+3)}\hspace{-.5cm} \sum_{\substack{ s_1+\ldots+s_N=N+n+1 \\ \mu_1+\ldots+\mu_N=\mu }  }    \frac{\prod_{i=1}^{N}\|\hat{\test}_{i}\|_{\frac{\mu_{i}+s_{i}}{2}}} { M^{\mu}}   \, .
\end{split}
\een
The constant $K=K(N,L)$ absorbs $\tilde K$ as well as additional $N,L$ dependent factors from the summations over $\tkn_{N,2L+2,0}$ and $\mathfrak{S}(\mathcal{E})$.  The desired inequality \eqref{CAGloopintfinal}, and thereby corollary \ref{cor3b}, finally follows after combining the sums over $n$ and $\mu$ into one, where we use the inequality $([A]+[B])(N+2L+3)+N+n+1\leq ([A]+[B]+2)(N+2L+3)$.

The proofs of corollaries \ref{cor1b} and \ref{cor2b} for the smeared connected Schwinger functions with one and no insertion are analogous, but simpler. \hfill \proofSymbol
\addcontentsline{toc}{chapter}{Bibliography}
%\bibliographystyle{utphys}
%\bibliography{deform}

\begin{thebibliography}{10}
\bibitem{Hollands:2011gf}
S.~Hollands and Ch.~Kopper, ``{The operator product expansion converges in
  perturbative field theory},'' {\em Commun.Math.Phys.} {\bf 313} (2012)
  257--290.
\bibitem{Wilson:1969ub}
K.~Wilson, ``{Non-Lagrangian models of current algebra},'' {\em Physical
  Review} {\bf 179} (1969)  1499--1512.
\bibitem{Zimmermann:1973wp}
W.~Zimmermann, ``{Normal products and the short distance expansion in the
  perturbation theory of renormalizable interactions },'' {\em Ann.Phys.} {\bf
  77} (1973)  570--601.
%
\bibitem{Polchinski:1983gv}
J.~Polchinski, ``{Renormalization and Effective Lagrangians},'' {\em
  Nucl.Phys.} {\bf B231} (1984)  269--295.
\bibitem{Keller:1990ej}
G.~Keller, Ch.~Kopper, and M.~Salmhofer, ``{Perturbative renormalization and
  effective Lagrangians in $\Phi^4$ in four-dimensions},'' {\em Helv.Phys.Acta}
  {\bf 65} (1992)  32--52.
%
\bibitem{Keller:1991bz}
G.~Keller and Ch.~Kopper, ``{Perturbative renormalization of composite operators
  via flow equations. 1.},'' {\em Commun.Math.Phys.} {\bf 148} (1992)
  445--468.
 \bibitem{Keller:1992by}
G.~Keller and Ch.~Kopper, ``{Perturbative renormalization of composite operators
  via flow equations. 2. Short distance expansion},'' {\em Commun.Math.Phys.}
  {\bf 153} (1993)  245--276.
\bibitem{GK}
R.~Guida and Ch.~Kopper, ``{All order uniform momentum bounds for the massless $\phi_{4}^{4}$ field theory},'' {\em in preparation}. %{\bf 179} (1969)  1499--1512.
\bibitem{Kopper:2001to}
Ch.~Kopper and F.~Meunier, ``{Large Momentum bounds from Flow Equations},'' {\em Annales Henri
  Poincare} {\bf 3} (2002)  435--449.
\bibitem{Lovasz:2003ua}
Lov{\'a}sz, L., Pelik{\'a}n, J. and Vesztergombi, K., ``{Discrete Mathematics: Elementary and Beyond},'' {\em Undergraduate Texts in Mathematics, } {Springer} (2003).
%
%
%
\bibitem{Wilson:1971bg}
K.~G. Wilson, ``{Renormalization group and critical phenomena. 1.
  Renormalization group and the Kadanoff scaling picture},'' {\em Phys.Rev.}
  {\bf B4} (1971)  3174--3183.
\bibitem{Wilson:1971dh}
K.~G. Wilson, ``{Renormalization group and critical phenomena. 2. Phase space
  cell analysis of critical behavior},'' {\em Phys.Rev.} {\bf B4} (1971)
  3184--3205.
\bibitem{Wegner:1972ih}
F.~J. Wegner and A.~Houghton, ``{Renormalization group equation for critical
  phenomena},'' {\em Phys.Rev.} {\bf A8} (1973)  401--412.
%
%
%\bibitem{Wetterich}
%C.~Wetterich, ``{Exact evolution equation for the effective potential},'' {\em Phys.Lett.} {\bf
%  B301} (1993)  90--94.
%
\bibitem{Muller:2002he}
V.~F. M{\"u}ller, ``{Perturbative renormalization by flow equations},'' {\em
  Rev.Math.Phys.} {\bf 15} (2003)  491.
%
\bibitem{Kopper:1997vg}
Ch.~Kopper, ``{Renormierungstheorie mit Flussgleichungen}.'' Shaker, 1998.
%
\bibitem{Holland:2012vw}
J.~Holland and S.~Hollands, ``{Operator product expansion algebra},'' {\em
  J.Math.Phys.} {\bf 54} (2013)  072302.
\bibitem{glimm}
Glimm J. and Jaffe, A., ``{Quantum Physics: A Functional Integral Point of View},''  {Springer} (1987).
%\bibitem{Kopper:2000qm}
%C.~Kopper, V.~F. M{\"u}ller, and T.~Reisz, ``{Temperature independent
%  renormalization of finite temperature field theory},'' {\em Annales Henri
%  Poincare} {\bf 2} (2001)  387--402.
%
  %
\bibitem{low}
J.H.~Lowenstein, ``{Normal Product Quantization of Currents in Lagrangian Field Theory},'' {\em Phys.Rev. }{\bf D4} (1971)  2281-2290.
%
%\bibitem{zimbrand}
%W.~Zimmermann, ``{Local operator products and renormaliztion in quantum field
%  theory },'' in {\em Lectures on Elementary Particles and Quantum Field
%  Theory}, pp.~395--589.
%\newblock Deser, Stanley (ed.). Cambridge, Mass, 1970.
%
%\bibitem{Olbermann:2012uf}
%H.~Olbermann, ``{Quantum field theory via vertex algebras},'' {\em PhD Thesis
%  (Cardiff University)} (2010)  .
%
%\bibitem{Kopper:2009um}
%C.~Kopper, ``{On the local Borel transform of Perturbation Theory},'' {\em
%  Commun.Math.Phys.} {\bf 295} (2010)  669--699.
%
%\bibitem{stein1993harmonic}
%E.~M. Stein and T.~S. Murphy, {\em Harmonic analysis: real-variable methods,
%  orthogonality, and oscillatory integrals}, vol.~3.
%\newblock Princeton University Press, 1993.
\end{thebibliography}
\providecommand{\href}[2]{#2}\begingroup\raggedright\endgroup
\end{document}